\documentclass[10pt]{amsart}

\usepackage{mathtools}
\usepackage{amssymb,amsfonts,amsmath}
\usepackage{hyperref}
\usepackage{amsthm}

\usepackage{booktabs}
\usepackage{bm}

\usepackage{graphicx, verbatim, centernot}
\usepackage{caption}
\usepackage{xcolor}

\usepackage{multicol}
\usepackage[shortlabels]{enumitem}

\usepackage{dsfont}

\usepackage[left=1in,right=1in,top=1in,bottom=1in]{geometry}

\usepackage{fancyhdr}

\usepackage{tikz}
\usetikzlibrary{arrows.meta,positioning}

\usepackage{cleveref}

\newtheorem{theorem}{Theorem}[section]
\newtheorem*{theorem*}{Theorem}

\newtheorem{corollary}[theorem]{Corollary}
\newtheorem{lemma}[theorem]{Lemma}

\newtheorem{proposition}[theorem]{Proposition}

\newtheorem*{claim*}{Claim}
\newtheorem{fact}[theorem]{Fact}

\theoremstyle{definition}
\newtheorem{definition}[theorem]{Definition}

\theoremstyle{remark}

\DeclareMathOperator*{\essinf}{ess\,inf}
\DeclareMathOperator*{\esssup}{ess\,sup}

\def\eps{{\varepsilon}}

\newcommand{\pointsets}{\mathcal{N}}
\newcommand{\events}{\mathfrak{R}}
\newcommand{\dirichlet}{\mathcal{E}}
\newcommand{\ind}[1]{\mathds{1}_{#1}}
\newcommand{\innerProduct}[2]{\left\langle #1, #2 \right\rangle}
\newcommand{\Pois}{\textup{Pois}}
\newcommand{\Ai}{\mathrm{Ai}}
\newcommand{\spec}{{\mathrm{spec}}}

\newcommand{\tnorm}[1]{{\left\vert\kern-0.25ex\left\vert\kern-0.25ex\left\vert #1 
		\right\vert\kern-0.25ex\right\vert\kern-0.25ex\right\vert}}

\newcommand{\one}{\boldsymbol{1}}
\newcommand{\bX}{\mathbf{X}}
\newcommand{\bY}{\mathbf{Y}}
\newcommand{\bZ}{\mathbf{Z}}
\newcommand{\Vol}{\operatorname{Vol}}
\newcommand{\blambda}{{\boldsymbol{\lambda}}}
\newcommand{\dist}{\mathrm{dist}}

\renewcommand{\P}{\mathbb{P}}
\newcommand{\E}{\mathbb{E}}
\newcommand{\N}{\mathbb{N}}
\newcommand{\R}{\mathbb{R}}
\newcommand{\C}{\mathbb{C}}

\newcommand{\region}{\Lambda}
\newcommand{\potential}{\phi}
\newcommand{\activity}{\lambda}
\newcommand{\activities}{\pmb{\lambda}}
\newcommand{\borel}{\mathcal{B}}
\newcommand{\size}[1]{\left\lvert#1\right\rvert}
\newcommand{\diff}{\mathrm{d}}
\newcommand{\vol}{\mathrm{Vol}}
\newcommand{\randpoint}{\mathbf{y}}
\newcommand{\domain}{\mathcal{D}}
\newcommand{\norm}[1]{\left\lVert#1\right\rVert}
\newcommand{\absolute}[1]{\left\lvert#1\right\rvert}
\newcommand{\epsbound}{\varepsilon^{*}}
\newcommand{\V}[1]{\mathbf{#1}}

\newcommand{\Var}{{\operatorname{Var}}}
\newcommand{\Cov}{\operatorname{Cov}}

\AtBeginDocument{\pagestyle{plain}}

\begin{document}
	
	\title{Uniqueness, analyticity and mixing for Gibbs point processes  via spectral gaps}

	\author{Andreas G\"obel}
	\address{Hasso Plattner Institute, University of Potsdam, Germany}
	\email{andreas.goebel@hpi.de}
	\author{Matthew Jenssen}
	\address{King's College London, Department of Mathematics}
	\email{matthew.jenssen@kcl.ac.uk}
	\author{Marcus Michelen}
	\address{Northwestern University, Department of Mathematics}
	\email{michelen@northwestern.edu}
	\author{Marcus Pappik}
	\address{Hasso Plattner Institute, University of Potsdam, Germany}
	\email{marcus.pappik@hpi.de}
	\author{Will Perkins}
	\address{Georgia Institute of Technology, School of Computer Science}
	\email{math@willperkins.org}
	\author{Leon Schiller}
	\address{Hasso Plattner Institute, University of Potsdam, Germany}
	\email{leon.schiller@hpi.de}

	\begin{abstract}
		A Gibbs point process is a model for particles interacting in the continuum through a potential.  Among the most classical examples is the hard-sphere model, where given an activity parameter $\lambda$, a radius $r$, and a bounded set $\Lambda \subset \R^d$ one samples a Poisson process of intensity $\lambda$ in $\Lambda$ conditioned on the points forming the centers of a packing of  spheres of radius $r$.  We prove uniqueness of infinite-volume Gibbs measure, analyticity of the pressure, and various notions of spatial and temporal mixing for activities up to what we define as the spectral threshold $\lambda_{\spec}$ of the potential.  For each fixed dimension $d \geq 2$, this improves the uniqueness and analyticity bounds for the hard-sphere model.  As $d \to \infty$, we obtain not only the first bound with a growing improvement over the classical bounds, but one whose improvement grows exponentially fast.  
		{We also prove an optimal mixing time bound for heat bath dynamics for the hard-sphere model up to an expected density of $\Theta(d / 2^d)$, which is the first result that asymptotically matches the maximum density for rapid mixing predicted by Parisi and Zamponi.}

		We also exhibit repulsive, radial pair potentials for which $\lambda_{\spec} = + \infty$, thereby showing that there are Gibbs point processes with such potentials that provably exhibit no phase transition at any activity $\lambda > 0$.  Further, in dimensions $8$ and $24$ we exhibit such a potential with no phase transition for which the work of Cohn-Kumar-Miller-Radchenko-Viazovska proves that the unique ground state at any fixed density is given by the $E_8$ and Leech lattices, respectively.  
		Our work builds upon a 2013 work of Kondratiev--Kuna--Ohlerich which implicitly defined $\lambda_{\spec}$ and demonstrated a spectral gap for a Glauber-like continuum birth-death dynamics.  Our main work consists of showing that such a spectral gap implies {several strong notions of absence of phase transition,} along with analyzing the behavior of $\lambda_\spec$ {for some interesting potentials}.
	\end{abstract}
	\date{}
	
	\maketitle
	
	\section{Introduction}

	A Gibbs point process, or classical gas, is a model for particles interacting in the continuum through a potential.  For a pair potential given by a symmetric function $\phi: \R^d \times \R^d \to \R \cup \{+\infty\}$, we define the energy of a finite configuration $\bX \subset \R^d$ via $$H(\bX):= \sum_{\{x,y\} \subset \bX } \phi(x,y)\,.$$
	For a bounded, measurable region $\Lambda \subset \R^d$ and an activity $\lambda > 0$, the \emph{partition function} is defined by \begin{equation*}
		Z_{\Lambda}(\lambda) = \sum_{k \geq 0} \frac{\lambda^k}{k!} \int_{\Lambda^k} e^{-H(\mathbf{x})} \,\diff \mathbf{x}\,.
	\end{equation*}
	The \emph{finite-volume Gibbs measure} is defined by $$\mu_{\Lambda,\lambda}(A) = \frac{1}{Z_\Lambda(\lambda)} \sum_{k \geq 0} \frac{\lambda^k}{k!} \int_{\Lambda^k} \one{\{\mathbf{x} \in A}\} \cdot e^{-H(\mathbf{x})} \,\diff \mathbf{x}\,. $$
	
	One of the central thermodynamic quantities is the pressure, first defined by defining the \emph{finite-volume pressure} at activity $\lambda$ by $$p_{\Lambda}(\lambda) := \frac{1}{\vol(\Lambda)} \cdot \log Z_{\Lambda}(\lambda)$$
	and subsequently defining $p(\lambda)$ via $\lim_{\Lambda \nearrow \R^d} p_{\Lambda}(\lambda)$ over a sequence of sets $\Lambda \nearrow \R^d$ such as increasing cubes.

	In statistical physics, the central questions for these models concern properties of the infinite-volume limits: the uniqueness of infinite-volume Gibbs measures, analyticity of the pressure, and spatial decay of correlations.  In the computer science literature, the central questions concern sampling and approximation in finite volume: mixing times for Markov chains such as Glauber dynamics and approximation algorithms for the partition function.  Essentially, each of these properties listed says something about the model behaving in a ``gas-like'' or ``fluid'' fashion; in fact, strategies that prove one of these properties for a specific potential $\phi$ and activity $\lambda$ often prove many others  simultaneously.  For many choices of pair potential $\phi$, it is expected that the model undergoes a \emph{phase transition}: one definition of this is that for small $\lambda$ there is a unique infinite-volume Gibbs measure while at other $\lambda$ there is more than one such measure.

	The most well-studied choice of potential $\phi$ is the \emph{hard-sphere model}, where one sets $\phi(x,y) = + \infty$ if $\|x - y\|_2 < 2r$ and $0$ otherwise.  The hard-sphere model is supported on pointsets that are $2r$-separated, meaning that the points form the centers of a sphere packing of balls of radius $r > 0$. 
	It is widely believed in the physics literature that in dimensions $2$ and $3$ there is a phase transition for the hard-sphere model going back at least to the seminal 1957 work of Alder and Wainwright \cite{alder1957phase}; various physics works suggest phase transitions in other  dimensions (see e.g.\ \cite{finken2001freezing} for a treatment up to $d = 50$) as well as asymptotically as $d \to \infty$~\cite{frisch1999high,parisi2010mean}.  Dimension $d=2$ is particularly subtle; here Richthammer \cite{richthammer2007translation} proved a Mermin-Wagner type theorem, showing that all infinite-volume Gibbs measures are translation invariant.  This omits the possibility of crystallization, but does not rule out the possibility of phase transition breaking rotational order, which has been predicted by recent works in physics~\cite{bernard2011two}. See~\cite{lowen2000fun} for a survey on the physics literature on the hard sphere model.

	Classical techniques for understanding when a Gibbs point process is in the fluid phase---and in particular for showing uniqueness of infinite-volume Gibbs measure and analyticity of the pressure---include cluster expansion methods: one expands $\log Z_{\Lambda}(\lambda)$ as a power series for $\lambda \in \C$ and proves analyticity of $\log Z_{\Lambda}(\lambda)$ in a complex disk $|\lambda| < \lambda_{0}$ for some fixed $\lambda_{0} > 0$ uniformly in the domain $\Lambda$.  For these techniques, the relevant quantity is the \emph{temperedness constant} defined by \begin{equation}
		C_\phi = \sup_{x} \int_{\R^d}|1 - e^{-\phi(x,y)}| \,\diff y
	\end{equation}
	which is always assumed to be finite. 
	In the case where $\phi$ is non-negative, we say that it is \emph{repulsive}.  For repulsive potentials Groeneveld, Penrose, and Ruelle \cite{groeneveld1962two,penrose1963convergence,ruelle1963correlation} showed in various forms that one can take $\lambda_{0} = 1/(e C_\phi)$, which shows analyticity and uniqueness for $\lambda \in [0,\lambda_0)$. 
	
	Traditional routes to proving uniqueness of infinite-volume Gibbs measures go through the Kirkwood--Salsburg equations, a hierarchy of equestions for the correlation functions satisfied by infinite-volume Gibbs measures.  Approaches that directly analyze the Kirkwood--Salsburg equations are similar in spirit to cluster expansion approaches, and such results prove uniqueness up to $1/(e C_\phi)$. 
	We refer the reader to Ruelle's classic book for further context on this approach \cite{ruelle1969statistical}.  A 1970 work of Meeron \cite{meeron1970bounds} proved uniqueness and analyticity up to $1/C_\phi$ using a novel interpolation approach along with a variant of the Kirkwood--Salsburg equations.  A sequence of works by Michelen--Perkins \cite{michelen2022strong,michelen2023analyticity,michelen2025potential} proved uniqueness and analyticity along with various spatial and temporal mixing up to  $e/C_\phi$ as well as $e/\Delta_\phi$, where $\Delta_\phi < C_\phi$ is a quantity introduced in \cite{michelen2025potential} called the potential-weighted connective constant; these works were built upon tree-like integral recursions and were inspired by work of Weitz \cite{Wei06} in an analogous discrete setting.  Other approaches to proving uniqueness involve Dobrushin uniqueness \cite{houdebart2022explicit},    disagreement percolation \cite{betsch2023uniqueness,dereudre2019introduction,hofertemel2019disagreement}, and Markov chain mixing \cite{helmuth2022correlation}.

	Our contribution is to establish uniqueness and analyticity through a different approach.  A remarkable but little-noticed work by Kondratiev, Kuna and Ohlerich \cite{kondratiev2013spectral} shows that under a certain positive-definiteness condition on an integral operator defined in terms of $\phi$ and $\lambda$, a continuous-time Markov chain known as \emph{spatial birth-death dynamics} exhibits a spectral gap.  
	In~\cite{kondratiev2013spectral}, the authors used this  to show that for certain potentials there is a spectral gap for all activities and the corresponding models ``should not show any phase transition...at any activity.''  
	Our main work {proves this in a very strong form}, showing that this spectral gap condition for the dynamics implies uniqueness of infinite-volume Gibbs measure, analyticity of the pressure, and various strong forms of spatial and temporal mixing.  {These implications are reminiscient of the implications in Dobrushin and Shlosman's theory of complete analyticity in the setting of bounded-range lattice spin models~\cite{dobrushin1985completely}; the novelty is in the weakness of the assumption (the spectral gap of birth--death dynamics) and the generality to continuum models with unbounded range potentials.}
	
	Using these implications (and analyzing specific potentials), we will prove that for certain pair potentials there is uniqueness and analyticity at \emph{all} activities, thus providing the first rigorous proof of the existence of a Gibbs point process with a non-trivial, radial pair potential  that lacks a phase transition.  As we will see below, we will exhibit such a pair potential for which the breakthrough work of Cohn-Kumar-Miller-Radchenko-Viazovska \cite{cohn2022universal} proves that the ground state at any fixed density is uniquely given by the $E_8$ and Leech lattices in dimension $8$ and $24$ respectively. 
	Additionally, for classical examples such as the hard-sphere model, our approach demonstrates dramatically improved bounds for the fluid regime for every dimension $d \geq 2$, with improvements that increase exponentially as the dimension grows. 
	To our knowledge, this is the first asymptotically growing improvement in the literature.

	\subsection{Main results}
	
	To state our results, as well as those of \cite{kondratiev2013spectral}, 
	for a repulsive, tempered potential $\phi$, define the \emph{spectral threshold} $\lambda_{\spec}$ via 
	\begin{equation}\label{eq:lambda-spec-definition}
		\lambda_{\spec} := \left(\sup_{f \in L^2 : \|f \|_2 = 1} \iint f(x) f(y) (e^{-\phi(x,y)} - 1)\,\diff x\,\diff y\right)^{-1}
	\end{equation}
	where the supremum is over real-valued functions $f$ and if the supremum on the right-hand side is $0$ then we interpret $\lambda_{\spec} = +\infty$.  In the special case where $\phi$ is translation invariant, we may define $g(x) = 1 - \exp(-\phi(x,0))$ and we will have \begin{equation}\label{eq:lambda-spec-translation-invariant}
		\lambda_\spec = 
		\left(\esssup_{\xi} (-\widehat{g}(\xi) )\right)^{-1}
	\end{equation}
	where $\widehat{g}$ is the Fourier transform of $g$ (see Lemma \ref{lem:fourier-criterion} for a proof).   We will also assume that $\phi$ decays exponentially, meaning $1 - \exp(-\phi(x,y)) \leq \exp(-\Omega(\dist(x,y)))$ (see Definition \ref{def:potential_decay} for a precise definition).

	Our main theorem is as follows. 
	
	\begin{theorem}\label{thm:main}
		Let $\phi$ be a repulsive potential that decays exponentially.  Then for all $\lambda < \lambda_\spec$, there is a unique infinite-volume Gibbs measure and the infinite-volume pressure is analytic on $[0,\lambda_\spec)$.
	\end{theorem}
	
	{To illustrate the uses of Theorem~\ref{thm:main}, we give several examples, computing or bounding $\lambda_\spec$ for specific potentials and deducing the consqeuences.}
	
	\subsubsection{Hard sphere model}
	
	A first example is the \emph{hard sphere model}, {perhaps the most studied specific potential, and one that illustrates the quantiative improvements given by Theorem~\ref{thm:main}}.  For an interaction radius $2r_d$ we define \begin{equation} \label{eq:hard-sphere-potential}
		\phi(x,y) = \begin{cases}
			+\infty &\text{ if }\|x - y\|_2 \leq 2r_d \\
			0 &\text{ otherwise}
		\end{cases}
	\end{equation}
	where we often take $r_d$ to be the radius of the ball of volume $1$ in which case the temperedness constant is given by $C_\phi = 2^{d}$.    
	
	\begin{proposition}\label{prop:hard-spheres}
		For the case of the hard sphere potential, as $d \to \infty$ we have $$
		\lambda_\spec = \left(\frac{e}{2}+o(1)\right)^{d/2}\cdot C_\phi^{-1}\, .
		$$
	\end{proposition}
	
	In particular, we obtain not only the first bound for  uniqueness and analyticity thresholds for the hard sphere model that allows $\activity C_{\potential}$ to increase with higher dimensions, but also an exponential improvement over previous results.
	We in fact compute the leading order asymptotics of $\lambda_\spec$ (see \eqref{eq:hard-sphere-lambda-spec-asympt}).  For each dimension $d$, we explicitly compute $\lambda_\spec$ in terms of Bessel functions and their zeros below. 
	
	One may also parameterize the hard sphere model in terms of the expected density; for a region $\Lambda$ and activity $\lambda$, define $\alpha_{\Lambda}(\lambda)$ by \begin{equation}\label{eq:point-density-definition}
		\alpha_\Lambda(\lambda) =  \frac{\E_{\bX \sim \Lambda,\lambda}[|\bX|]}{\Vol(\Lambda)}\,.
	\end{equation}
	
	A basic calculation shows that $\lambda \mapsto \alpha_{\Lambda}(\lambda)$ is monotone increasing for all $\Lambda$ (see e.g.\ \cite[Lemma 6]{jenssen2019hard}).  
	For repulsive potentials such as the hard sphere model, the random point process $\bX$ is always stochastically dominated by a Poisson point process of intensity $\lambda$ (see Lemma \ref{lem:Poisson-domination}).  In particular, this implies that we always have the upper bound 
	\begin{equation}\label{eq:upper-bound-density}
		\alpha_{\Lambda}(\lambda) \leq \lambda\,.
	\end{equation}  
	All prior results for the hard sphere model obtained results for $\lambda = \Theta(2^{-d})$, which translates to $\alpha_{\Lambda}(\lambda) = \Theta(2^{-d})$ (see \cite{michelen2023analyticity} for a lower bound on the implicit constant).  
	
	Outside of the trivial upper bound \eqref{eq:upper-bound-density}, bounding the density is often challenging.  A theorem of Jenssen--Joos--Perkins \cite[Theorem~2]{jenssen2019hard} shows that for the hard sphere model, for $\lambda \geq 3^{-d/2}$ one has \begin{equation}\label{eq:LB-density}
		\alpha_\Lambda(\lambda) \geq (1 + o_d(1))\frac{\log(2/\sqrt{3}) d}{2^d} \,.
	\end{equation}
	Since Proposition \ref{prop:hard-spheres} shows $$\lambda_{\spec} \geq \left( \frac{e^{1/2}}{2^{3/2}} + o(1)\right)^{d}  \geq 3^{-d/2}\,, $$
	this shows that in fact when parametrizing in terms of the density, our results hold up to $\alpha_{\Lambda}(\lambda) = \Theta(d/ 2^d)$, a factor of $d$ denser than the reach of previous results.  As we will see below, this includes showing that we efficiently sample from the hard-sphere model using a Markov chain at expected density $\Theta(d/2^d)$ in high dimensions.  A physics prediction due to Parisi and Zamponi \cite[Eq.~(54)]{parisi2010mean} suggests that Markov chains should mix slowly at density above $\approx 4.8 \cdot d / 2^d$.  This shows that Parisi and Zamponi predict our result is sharp up to a constant factor of the density.   
	We also note that the spectral threshold yields an improvement in the uniqueness and analyticity thresholds for the hard sphere model in low dimensions as well.

	For the case of hard spheres in fixed dimension $d$, we compute $\lambda_\spec$ in terms of Bessel functions.  In particular, setting $\nu = d/2$ we show in Section \ref{sec:lambda-spec-HS} that 
	\begin{equation*}
		C_\phi \lambda_\spec = \Vol(B_{\R^d}(1)) (2\pi)^{-\nu} ( j_{\nu + 1,1})^{d/2} (-J_{\nu}(j_{\nu + 1,1}))^{-1}
	\end{equation*}
	where $J_\nu$ is the Bessel function of the first kind of order $\nu$ and $j_{\nu+1,1}$ is the first positive zero of $J_{\nu+1}$.  In dimension $2$ this yields $C_\phi \lambda_\spec \approx 7.56$.  There were multiple previous bounds in dimension $2$ in terms of quantities like critical percolation parameters and connective constants, and $\lambda_\spec$ is larger than all: Michelen--Perkins \cite{michelen2025potential} proved uniqueness and analyticity up to $C_\phi \lambda \leq 3.233$ using the potential-weighted connective constant (Monte Carlo simulations in \cite{michelen2025potential} suggest that the connective constant bound is roughly of size $4.38$). Betsch-Last \cite{betsch2023uniqueness} proved uniqueness up to the phase transition for continuum percolation  $C_\phi\lambda < \lambda_{\mathrm{perc}}$; the best rigorous lower bound on this quantity appears to be due to Ziesche \cite{ziesche2018sharpness} with $\lambda_{\mathrm{perc}} \geq 4.48$ and extensive simulations suggest that $\lambda_{\mathrm{perc}} \approx 4.51$ (see \cite{mertens2012continuum} for a survey of these results and \cite{balister2005continuum} for simulations with rigorous confidence intervals).  This is in comparison to the radius of convergence of the cluster expansion, which is only known to hold for $C_\phi \lambda < 1/e \approx 0.37$ and provably cannot hold for $C_\phi \lambda > 1$.
	
	We include in Table \ref{tab:hard-sphere-low-dimensions} the values of $C_\phi \lambda_\spec$ up to the third decimal place.  The previous records for each fixed dimension $d \geq 3$ were essentially only slightly larger than $e$ due to \cite{michelen2025potential}.
	\begin{table}[h!]        
		\centering
		\begin{tabular}{|c||c|c|c|c|c|c|c|c|c|c|}
			\hline
			$\bm{d}$ & 2 & 3 & 4 & 5 & 6 & 7 & 8 & 9 & 10 & 11 \\
			\hline
			$\bm{C_\phi\lambda_{\spec}}$ & 7.559 & 11.604 & 17.059 & 24.327 & 33.921 & 46.481 & 62.809 & 83.911 & 111.038 & 145.749\\
			\hline
		\end{tabular}
		\caption{Rigorous numerical approximations to $\lambda_{\spec}$ for the hard sphere potential in low dimensions.} 
		\label{tab:hard-sphere-low-dimensions}
	\end{table}

	\subsubsection{Gibbs point processes with no phase transition}
	
	We next give examples of repulsive Gibbs point processes for which we prove there is no phase transition.  In particular, if $\lambda_{\spec} = +\infty$ then one obtains uniqueness and analyticity for all $\lambda > 0$.
	
	\begin{corollary}\label{cor:positive-definite}
		Let $\phi$ be a repulsive potential that decays exponentially.  Assume $\phi$ is translation invariant and define $g(x) = 1- e^{-\phi(x,0)}$.  If $\widehat{g}(\xi) \geq 0$ for all $\xi \in \R^d$ then for all $\lambda \geq 0$ there is a unique infinite volume Gibbs measure and the infinite volume pressure is analytic on $[0,\infty)$.
	\end{corollary}

	We include a short Fourier argument  that shows that it is impossible for $\phi$ to  satisfy the hypotheses of Corollary \ref{cor:positive-definite} while simultaneously having a hard-core (Lemma \ref{lem:impossible-core}).

	{One can ask if potentials satisfying the conditions of Corollary~\ref{cor:positive-definite} exist, and if so, whether they avoid phase transitions by not possesing sufficiently ordered ground states that might drive crystallization.  We will in fact show that there are potentials with no phase transition whose ground states are ordered in the strongest sense.}
	Consider the potential $\phi$ defined by \begin{equation} \label{eq:positive-def-phi}
		\phi(x,y) = \log\left(\frac{1}{1 - \exp(-\|x - y\|^2)}\right)\,.
	\end{equation}
	In this case, we see that $g(x) = \exp(-\|x \|^2)$ and so $\widehat{g}(\xi) \geq 0$ for all $\xi \in \R^d$, and so  Corollary~\ref{cor:positive-definite} implies that the potential $\phi$ does not exhibit a phase transition at any $\lambda > 0$.   While there are many choices for potential $\phi$ which yield $\widehat{g}(\xi) \geq 0$ for all $\xi$, the potential \eqref{eq:positive-def-phi} is in fact a completely monotonic function of distance squared as can be seen by the expansion \begin{align*}
		\phi(x,y) = \log\left(\frac{1}{1 - \exp(-\|x - y\|^2)}\right) = \sum_{k \geq 1} \frac{e^{-k\|x - y\|^2}}{k} \,.
	\end{align*}
	
	As such, the breakthrough work of Cohn-Kumar-Miller-Radchenko-Viazovska \cite{cohn2022universal} shows that the potential $\phi$ defined in \eqref{eq:positive-def-phi} has a unique ground state in dimensions $8$ and $24$ up to isometries; the ground states are the $E_8$ and Leech lattices, respectively, which were shown to be the unique optimal sphere packing configurations in work by Viazovska \cite{viazovska2017sphere} and Cohn-Kumar-Miller-Radchenko-Viazovska \cite{cohn2017sphere}.  More precisely, this means that if we fix any density $\rho$, the unique energy minimizer over point configurations at density $\rho$ in dimensions $8$ and $24$ are given by a rescaled version of the $E_8$ and Leech lattice respectively.  Despite the potential having highly structured and unique ground states at each density, Corollary \ref{cor:positive-definite} shows that at \emph{all} activities the Gibbs point process remains gas-like. 
	
	It is conjectured that in dimension $2$ completely monotonic potentials have unique ground state given by the hexagonal lattice---which is the unique optimizer for the sphere packing problem---but this remains open (see Montgomery \cite{montgomery1988minimal} for a proof that the hexagonal lattice is optimal among lattices, and \cite{cohn2022universal} for further context and references).

	\subsubsection{Gaussian core model}  
	The Gaussian core model is another widely studied potential for which we obtain a significant quantitative improvement. 
	The Gaussian core model is given by the potential $$\phi(x,y) = \beta e^{-\|x - y\|_2^2 / 2}$$ where $\beta > 0$ is an inverse temperature parameter. We will work in the regime where $\beta$ is fixed and $d \to \infty$.  In this setting, work of Cohn and de Courcy-Ireland \cite{cohn2018gaussian} show that the approximate ground states at each fixed density are given by the point configurations of random lattices.
	
	In the regime where $\beta$ is fixed and $d \to \infty$ we obtain an exponential improvement over prior results.  In particular, defining the temperedness constant $$C_\phi = \int_{\R^d} 1 - e^{-\phi(x,0)} \,\diff x$$
	then prior results \cite{michelen2023analyticity,michelen2025potential} showed uniqueness and analyticity up to $\lambda \leq e/C_\phi$.  We will see in Lemma \ref{lem:gaussian-core-bound} that there is a constant $C_\beta > 0$ so that $\lambda_{\spec} \geq C_\beta 2^d / C_\phi$.

	\subsection{Temporal and spatial mixing}

	We also prove optimal mixing (i.e., asymptotically optimal mixing times) of a natural Markov chain known as \emph{block dynamics} or \emph{heat bath dynamics}.  Block dynamics with update radius $L$ on a bounded measurable set $\Lambda$ is the Markov chain where given a configuration $X$, the next configuration is obtained by choosing a point $x \in \Lambda$ uniformly at random and updating the configuration in $B_L(x)$ according to the finite-volume Gibbs measure associated to $\lambda,\phi$ subject to the boundary condition $X \cap B_L(x)^c$ (see Definition \ref{def:block_dynamics} for the formal details).  We prove optimal mixing of block dynamics up to $\lambda_{\spec}$.
	
	\begin{theorem}\label{thm:block-dynamics}
		Let $\phi$ be a repulsive potential that decays exponentially.  Then for all $\lambda < \lambda_{\spec}$, there is some $L_0 = L_0(\lambda,\phi) > 0$ so that for all activity functions $\blambda \leq \lambda$ and block size $L \geq L_0$, the block dynamics with update radius $L$ have mixing time $O(\Vol(\Lambda) \log \Vol(\Lambda) )$ for all boxes $\Lambda = [-n,n]^d \subset \R^d$.
	\end{theorem}
	
	While the connection between a spectral gap and mixing bounds for the corresponding dynamics is classical, it often provides a sub-optimal mixing bound.  In particular, using only the information of the spectral gap provided by \cite{kondratiev2013spectral}, one would obtain\footnote{In the case where the potential does not have a hard core, one can bound the spectral gap of block dynamics to that of the birth-death dynamics by a comparison argument.} a mixing time of $O(\Vol(\Lambda)^2)$ rather than our bound of  $O(\Vol(\Lambda) \log \Vol(\Lambda))$, which is optimal.

	Underpinning all of our results is the use of a notion of spatial mixing.  In the case $\phi$ is finite-range, the notion of strong spatial mixing is sufficient to capture the change in a Gibbs measure under the inclusion of an additional point.  Since our potentials may have infinite-range, we introduce an additional notion of spatial mixing to capture exponential decay of the influence of including one additional point.  We call this notion \emph{single-site Strong Spatial Mixing} (single-site SSM). 
	
	To formalize this, for a given measurable function $\blambda: \R^d \to [0,\lambda]$, we let $\mu_{\Lambda,\blambda}$ denote the Gibbs distribution on $\Lambda$ with activity function $\blambda$ (see Section \ref{sec:finite-volume-prelims} for a precise definition of $\mu_{\Lambda,\blambda}$). 
	For a point $x \in \R^d$, and activity function $\blambda$, define $\blambda_x : \R^d \to [0,\lambda]$ via $\blambda_x(v) = \blambda(v) e^{-\phi(x,v)}$ (see Section \ref{sec:pinnings} for our more general definition of an activity function with pinnings).  For a Borel set $\Delta$, a bounded set $\Lambda$, and activity functions $\blambda$ and $\blambda'$ we define the projected total variation distance $\|\mu_{\Lambda,\blambda} -  \mu_{\Lambda,\blambda'}\|_{\Delta}$ as the total variation distance between the pushforwards of $\mu_{\Lambda,\blambda}$ and $\mu_{\Lambda,\blambda'}$ under the map $X \mapsto X \cap \Delta$ for a point set $X$.

	\begin{definition}[Single-site Strong Spatial Mixing]\label{def:spatial_mixing}
		A potential $\potential$ satisfies \emph{single-site Strong Spatial Mixing} (single-site SSM)  up to some activity $\activity \in \R_{\ge 0}$ if there exist constants $C < \infty$ and $m > 0$ such that, for all bounded measurable $\region' \subseteq \region \subset \R^d$, all $x \in \R^d$ and all activity functions $\activities \le \activity$ it holds that 
		\[
		\|\mu_{\region, \activities} - \mu_{\region, \activities_x}\|_{\region'} \le C  \vol(\region')  e^{- m \cdot \dist(x, \region')},
		\]
		where for the distance $\dist(x, \region')$ we take $\inf\{\|x-y\|_2\mid y\in\Lambda'\}$.
	\end{definition}
	
	We also recall the notion of Strong Spatial Mixing (SSM) for point measures, introduced in \cite{michelen2022strong}.

	\begin{definition}[Strong Spatial Mixing] \label{def:SSM}
		A potential $\phi$ satisfies \emph{Strong Spatial Mixing} (SSM) up to some activity $\lambda \in \R_{\geq 0}$ if there exist constants $C > 0$ and $m > 0$ so that for all bounded and measurable $\Lambda' \subseteq \Lambda \subset \R^d$ and all activity functions $\blambda,\blambda' \leq \lambda$ it holds that $$\|\mu_{\Lambda,\blambda} - \mu_{\Lambda,\blambda'}\|_{\Lambda'} \leq C \Vol(\Lambda') e^{-m \cdot\dist(\Lambda', \mathrm{supp}(\blambda - \blambda'))}\,.$$
	\end{definition}
	
	We note that when $\phi$ has finite-range, SSM implies single-site SSM.  In general, the two are not immediately equivalent.  
	We will show that both occur up to $\lambda_\spec$ when $\phi$ decays exponentially, which yields various algorithmic implications discussed below.
	
	\begin{theorem}\label{thm:spatial-mixing}
		Let $\phi$ be a repulsive potential that decays exponentially.  Then $\phi$ satisfies SSM and single-site SSM up to $\lambda_\spec$.
	\end{theorem}
	
	In fact, we will show that in the setting of Theorem \ref{thm:spatial-mixing}, single-site SSM implies SSM (see Proposition \ref{prop:SSM}) but this appears to be non-trivial and uses the machinery built for SSM in \cite{michelen2022strong}.  A similar argument shows that SSM implies single-site SSM in the setting of Theorem \ref{thm:spatial-mixing}.

	As a final notion of spatial mixing, we will see that the $k$-point correlation functions of the infinite-volume measure (which is unique by Theorem \ref{thm:main}) enjoy exponential decay.  For an activity function $\blambda$ for which there is a unique infinite-volume Gibbs measure, let $\rho_{\blambda}(v_1,\ldots,v_k)$ denote the $k$-point correlation function for the unique infinite-volume Gibbs measure.  
	
	We will see exponential decay of $k$-point correlation functions up to $\lambda_{\spec}$.
	
	\begin{theorem}\label{thm:decay-of-k-point-correlations}
		Let $\phi$ be a repulsive potential that decays exponentially.  For each $\lambda < \lambda_{\spec}$ there are constants $C,m > 0$ depending on $\phi,d,\lambda$ so that for all $k,\ell \in \mathbb{N}$ and $\mathbf{v} = (v_1,\ldots,v_k)$ and $\mathbf{w} = (w_1,\ldots,w_\ell)$ we have $$\left|\rho_{\blambda}(\mathbf{v},\mathbf{w}) - \rho_{\blambda}(\mathbf{v})\rho_{\blambda}(\mathbf{w}) \right| \leq C^{k + \ell} e^{-m \cdot\dist(\mathbf{v},\mathbf{w})}$$
		where we write $\dist(\mathbf{v},\mathbf{w}) = \min_{i,j} \dist(v_i,w_j).$
	\end{theorem}

	\subsection{Algorithmic implications}
	Efficient algorithms for sampling configurations from a Gibbs point process and approximating partition functions have been studied extensively~\cite{friedrich2022spectral,michelen2022strong,jerrum2019perfect,anand2023perfect,FGKKP2025,jenssen2024quasipolynomial}.
	While most algorithms are restricted to repulsive, bounded-range pair potentials
	~\cite{friedrich2022spectral,michelen2022strong,jerrum2019perfect,anand2023perfect}, a few algorithmic results also extend beyond repulsive~\cite{jenssen2024quasipolynomial} or beyond bounded-range potentials~\cite{FGKKP2025}. 
	The aforementioned results give asymptotic run-time guarantees up to $\lambda<e/C_\phi$\footnote{Or $\lambda<e/\Delta_\phi$ using the results of \cite{michelen2025potential}.}.
	We note that, by efficiently implementing each step of the heat bath dynamics via rejection sampling, Theorem \ref{thm:block-dynamics} immediately yields and efficient approximate sampling algorithm all the way up to $\lambda_\spec$.
	Combining this with a standard equivalence between counting and sampling for spin systems (see e.g.\ the self-reducibility argument in \cite{michelen2022strong}) this further yields an efficient Monte Carlo approximation of the partition function in the same activity regime.
	Lastly, when restricted to bounded-range potentials, combining our SSM result in Theorem \ref{thm:spatial-mixing} with \cite{anand2023perfect} yields an exact sampling algorithm with near linear running time up to $\lambda_\spec$.

	\subsection{The canonical ensemble}
	In this short section we record some results which establish analyticity in the \emph{canonical ensemble} associated to the Gibbs points process considered above (which is the grand canonical ensemble). 
	
	In the canonical ensemble, the density of points is fixed and configurations are chosen with probability proportional to $e^{-H(\mathbf{x})}$. For $k$ and $\Lambda$ define the finite-volume canonical partition function $$\widehat{Z}_{\Lambda}(k) = \frac{1}{k!} \int_{\Lambda^{k}} e^{-H(\mathbf{x})}\,\diff \mathbf{x}\,.$$
	For $\alpha > 0$, define the pressure of the canonical ensemble to be  
	$$\psi(\alpha) := \lim_{n \to \infty} \frac{1}{\Vol(\Lambda_n)}\log \widehat{Z}_{\Lambda_n}(\lfloor \alpha \Vol(\Lambda_n)\rfloor ) 
	$$
	where existence of the limit\footnote{In the case where $\phi$ has a hard core, this limit may be given by $-\infty$ for $\alpha$ large enough.} was proven by Fisher \cite{fisher1964free} (see also Ruelle \cite{ruelle1969statistical}).   
	
	Using Theorem \ref{thm:main} we will see that the the limiting infinite-volume density in the grand canonical ensemble $$\alpha(\lambda):= \lim_{n \to \infty} \alpha_{\Lambda_n}(\lambda)
	$$
	exists and is analytic for $\lambda \in (0,\lambda_\spec)$.  
	We will also see that $\alpha$ is strictly increasing on $(0,\lambda_\spec)$ and therefore the limit
	\[
	\alpha_{\spec}
	\coloneqq
	\lim_{\lambda\uparrow\lambda_{\spec}}\alpha(\lambda)\, .
	\]
	exists in \([0,\infty]\).  We first deduce analyticity for the pressure of the canonical ensemble.

	\begin{theorem}\label{thm:analyticity-canonical}
		Let $\phi$ be a repulsive potential that decays exponentially.  Then $\psi$ is analytic for all $\alpha \in (0,\alpha_\spec)$.
	\end{theorem}

	In the case of the hard-sphere potential, we obtain analyticity for the canonical pressure up to density $\Theta(d2^{-d})$ in comparison to previous approaches which only go up to density $\Theta(2^{-d})$. 
	
	\begin{corollary}\label{cor:canon-hs}
		For the hard-sphere potential, $\psi$ is analytic for all $\alpha \in (0,cd2^{-d})$ for each fixed $c<\log(2/\sqrt{3})$ and $d$ sufficiently large. 
	\end{corollary}

	We also show that if $\lambda_\spec=\infty$ then $\alpha_\spec=\infty$, and so for such potentials $\psi$ is analytic for \emph{all} $\alpha>0$.
	
	\begin{corollary}\label{cor:canon-pd}
		Let $\phi$ be a repulsive, translation invariant potential that decays exponentially. If $\lambda_\spec=+\infty$ then $\psi$ is analytic for all $\alpha\in (0,\infty)$.
	\end{corollary}

	\subsection{Discussion and overview of the proofs}

	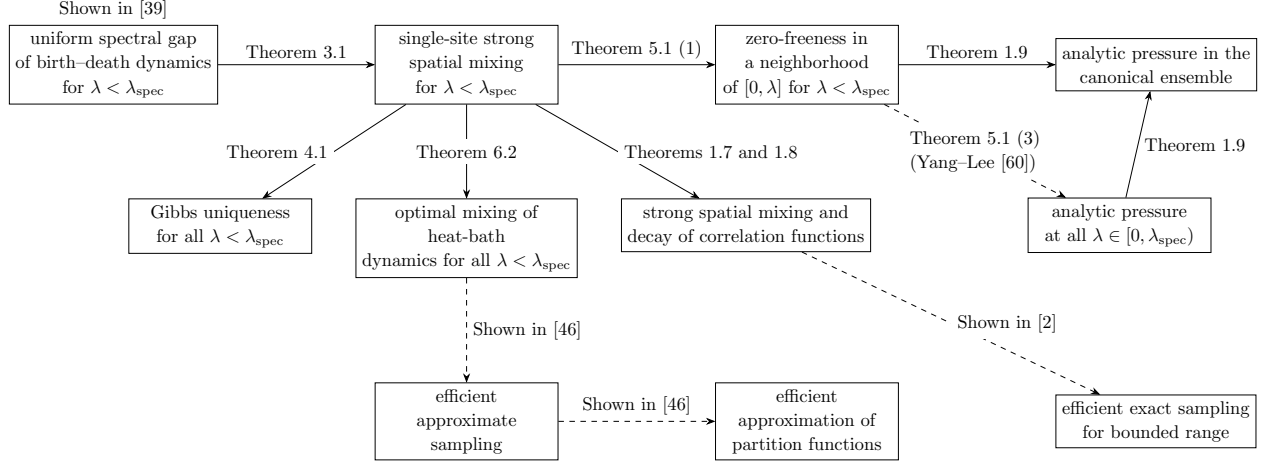
\begin{figure}
		\centering
		\resizebox{\textwidth}{!}{%
			\begin{tikzpicture}[
				node distance=3cm and 3cm,
				box/.style={
					draw,
					rectangle,
					align=center,
					minimum width=3.5cm,
					minimum height=1cm
				},
				->,
				>=Stealth
				]
				
				\node[box, label=above:{Shown in \cite{kondratiev2013spectral}}] (gap)
				{uniform spectral gap\\ 
					of birth--death
					dynamics\\ for $\lambda < \lambda_\spec$};
				
				\node[box,right=of gap] (ssm)
				{single-site strong\\
					spatial mixing\\
					for $\lambda < \lambda_\spec$};
				
				\node[box,right=of ssm] (zf)
				{zero-freeness in\\
					a neighborhood\\
					of $[0,\lambda]$ for $\lambda < \lambda_{\spec}$};
				
				\node[box,right=of zf] (analyticity)
				{analytic pressure in the\\
					canonical ensemble};
				
				\node[box,below left=1.8cm and 1.2cm of ssm] (gibbs)
				{Gibbs uniqueness\\
					for all $\lambda < \lambda_\spec$};
				
				\node[box,below=1.8cm of ssm] (hb)
				{optimal mixing of\\
					heat-bath\\
					dynamics for all $\lambda < \lambda_{\spec}$};
				
				\node[box,below right=1.8cm and 1.2cm of ssm] (other)
				{strong spatial mixing and \\
					decay of correlation functions};
				
				\node[box,right=of other] (pressure)
				{analytic pressure\\
					at all $\lambda \in [0,\lambda_\spec)$};
				
				\node[box,below=2cm of hb] (approxsample)
				{efficient\\
					approximate\\
					sampling};
				
				\node[box,right=3cm of approxsample] (approxpart)
				{efficient\\
					approximation of\\
					partition functions};
				
				\node[box,right=3cm of approxpart] (exactsample)
				{efficient exact sampling\\
					for bounded range};
				
				\draw (gap) --node[midway, above] {\Cref{{thm:ssssm}}} (ssm);
				\draw (ssm) --node[midway, above] {\Cref{thm:zero_freeness} (1)} (zf);
				\draw (zf) --node[midway, above] {\Cref{thm:analyticity-canonical}}  (analyticity);
				
				\draw (ssm) --node[fill=white,midway,left] {\Cref{thm:ssssm_uniqueness}} (gibbs);
				\draw (ssm) --node[fill=white,midway] {\Cref{{th:mixing-block-dynamics}}} (hb);
				\draw (ssm) --node[fill=white,midway,right] {\Cref{thm:spatial-mixing,thm:decay-of-k-point-correlations}} (other);
				
				\draw[dashed] (zf) --node[fill=white,midway,align=left]{\Cref{{thm:zero_freeness}} (3) \\ (Yang--Lee \cite{yang1952statistical})} (pressure);
				\draw (pressure) --node[midway, right] {\Cref{thm:analyticity-canonical}} (analyticity);
				
				\draw[dashed] (hb) --node[midway,right]{Shown in \cite{michelen2022strong}} (approxsample);
				\draw[dashed] (other) --node[fill=white,midway, right]{Shown in \cite{anand2023perfect}} (exactsample);
				
				\draw[dashed] (approxsample) --node[midway,above]{Shown in \cite{michelen2022strong}} (approxpart);
				
		\end{tikzpicture}}
		\caption{The graph of implications. Dashed lines indicate implications known from the literature. \label{fig:overview}}
	\end{figure}
	
	The implications we prove in this paper are summarised in \Cref{fig:overview}.
	
	As discussed earlier, our starting point is a theorem of Kondratiev--Kuna--Ohlerich which shows that one can bound the spectral gap for spatial birth-death dynamics.  We will formally define these dynamics in Section \ref{sec:birth-death-dynamics}, but informally describe them here: in a domain $\Lambda$, we attempt to add points according to a Poisson process of intensity measure $\blambda$ in continuous time with rate $1$; when we attempt to add a  point at $x$ to a configuration $\bX$, add it with probability $\exp(-\sum_{y \in \bX} \phi(x,y))$;  points are independently killed at rate $1$.  The generator $L_{\Lambda,\blambda}$ for these dynamics is defined in \eqref{eq:generator-definition}.  We set its Dirichlet form to be $\dirichlet_{\Lambda,\blambda}$ (defined in \eqref{eq:dirichlet-def}) and $\mathcal{D}_{\Lambda,\blambda}(\dirichlet_{\Lambda,\blambda})$ to be its domain.  
	
	Among the main results of \cite{kondratiev2013spectral} is a so-called coercivity identity which is equivalent to a \emph{spectral gap}.  In particular, for function $f \in \mathcal{D}_{\Lambda,\blambda}(\dirichlet_{\Lambda,\blambda})$, we define $\Var_{\lambda,\blambda}(f)$ to be the variance of the random variable $f(\bX)$ where $\bX$ is sampled from the finite-volume Gibbs measure $\mu_{\Lambda,\blambda}$.  Our starting point is the following special case of \cite[Theorem~3.4]{kondratiev2013spectral}:
	
	\begin{theorem}[Kondratiev--Kuna--Ohlerich, \cite{kondratiev2013spectral}] \label{thm:Kondratiev}
		Let $\phi$ be a repulsive and tempered potential and suppose that $\blambda$  and $\gamma \in (0,1)$ are such that for all $f \in L^2$ with $\|f \|_2 = 1$ we have
		\begin{equation}\label{eq:positive-def-assumption}
			\iint f(x) f(y) \sqrt{\blambda(x) \blambda(y)} (1 - e^{-\phi(x,y)}) \,dx\,dy + (1 - \gamma) \geq 0\,.
		\end{equation}
		Then for all bounded and  measurable $\Lambda \subset \R^d$, it holds that \begin{equation}\label{eq:spectral-gap}
			\Var_{\Lambda,\blambda}(f) \leq \gamma^{-1} \cdot \mathcal{E}_{\Lambda,\blambda}(f) \text{ for all }f \in \mathcal{D}(\mathcal{E}_{\Lambda,\blambda})\,.\end{equation}
	\end{theorem}

	The paper \cite{kondratiev2013spectral} works in a higher level of generality than Gibbs point processes, and the version of Theorem \ref{thm:Kondratiev} stated there is for arbitrary point processes that admit \emph{Papangelou intensities}.  Roughly, the Papangelou intensity is the likelihood ratio between a point set $\bX$ and $\bX \cup v$. In the case of Gibbs point processes, this simplifies in terms of the activity function $\blambda$ and the added energy.  
	
	The proof of \cite{kondratiev2013spectral} is via a Bochner--Bakry--\'Emery approach: one defines the \emph{Carr\'e-du-champ} operators $\Gamma$ and $\Gamma_2$, applies a discrete integration by parts and takes expectation with respect to the Gibbs measure.  Showing a certain inequality between the two expectations---known as a coercivity identity---is equivalent to a spectral gap.  The manipulations are elementary in nature but fairly involved, and ultimately in the main inequality the only two inequalities used are the positive-definiteness assumption \eqref{eq:positive-def-assumption} along with bounding an integral of squares below by zero. 
	
	Incredibly, a nearly identical condition to \eqref{eq:positive-def-assumption} appeared in a recent breakthrough by Chen--Chen--Chen--Yin--Zhang \cite{chen2025rapid} which showed rapid mixing of Glauber dynamics for the hard-core model on random regular graphs (and also on other graphs which satisfy a bound on their minimum eigenvalue).  The proof of \cite{chen2025rapid} combines technical ingredients  and introduces new ideas related to several recently growing topics in local-to-global arguments: a trickle-down theorem for field dynamics, spectral stability, and a comparison between field dynamics and Glauber dynamics.  In a companion paper \cite{goebel2026simple}, we provide a short and self-contained proof of \cite{chen2025rapid} by streamlining the elementary argument from \cite{kondratiev2013spectral} and adapting it to the discrete case.
	
	Our first step is to show that a spectral gap \eqref{eq:spectral-gap} uniformly for \emph{all} $\blambda \leq \lambda$ implies both single-site SSM and SSM.  In their foundational 2002 work on spatial birth-death dynamics for Gibbs point processes, Bertini--Cancrini--Cesi showed decay of covariance for local test functions (see their \cite[Theorem 4.1]{bertini2002spectral} and our version Proposition \ref{prop:covariance}).  Crucially, their proof required the assumption that the potential is of bounded range; since we are interested in potentials of infinite range---such as the potential \eqref{eq:positive-def-phi}---we required some new ideas.  In particular, in order to obtain such an exponential bound, we restrict to potentials that decay exponentially (see Definition \ref{def:potential_decay} for a precise definition) and introduce a notion of approximately local functions (Definition \ref{def:approxlocal}).  At the core of both Proposition \ref{prop:covariance} and \cite[Theorem 4.1]{bertini2002spectral} is a propagation-of-influence estimate for the birth-death semigroup (Lemma \ref{lemma:unbounded_4.2}, the analogue of \cite[Lemma 4.2]{bertini2002spectral}); the notions of potential decay and approximately local functions were introduced precisely to extend this propagation-of-influence estimate.  Single-site SSM for sets of volume bounded below follows from a quick application of Proposition \ref{prop:covariance}.  Using single-site SSM for large sets will show correlation decay for one-point correlation functions (Lemma \ref{lem:1-pt-SSSSM}).  We then adapt the strategy of \cite{michelen2022strong} to prove single-site SSM for all sets (Lemma \ref{lem:upgrade-to-small-vol}) and SSM (Proposition \ref{prop:SSM}).   This is carried out in Section \ref{sec:spatial-mixing}.
	
	We then move in Section \ref{sec:uniqueness} to prove that there is a unique infinite-volume Gibbs measure, and that it has decaying correlation functions.  
	In the analogous discrete case---or similarly when the potential is of finite range---uniqueness follows essentially immediately from strong spatial mixing along with the DLR equations.  Our setting is a bit trickier, and in fact served as one of the main motivations of our definition of single-site SSM.  To show uniqueness, the DLR equations state that for an infinite-volume Gibbs measure, the point configuration on a finite set $\Lambda$ may be taken to be the finite-volume Gibbs measure on $\Lambda$ subject to the boundary conditions of the point configuration outside of $\Lambda$.  We show that uniqueness follows from single-site SSM by summing over the influence of each point in the boundary condition and using Poisson dominance along with exponential decay of the potential to obtain a summable tail that decreases as $\Lambda \uparrow \R^d$ (Theorem \ref{thm:ssssm_uniqueness}).  We then deduce decay of correlation functions for the infinite-volume Gibbs measure (Proposition \ref{prop:exponential-decay-k-point}) using our bounds on correlation functions from Section \ref{sec:spatial-mixing}.
	
	We have now arrived at the most technical section of the paper, Section \ref{sec:zero-free}, where we prove a zero-free region for the partition function $Z_{\Lambda}(\lambda)$ along with analyticity of the infinite-volume pressure (Theorem \ref{thm:zero_freeness}).  At a very high level, our proof takes some inspiration from a pair of influential works of Dobrushin and Shlosman \cite{dobrushin1985completely,dobrushin1987completely} which prove that for finite-range pair interactions in the discrete setting of $\mathbb{Z}^d$ various strong notions of a ``fluid regime'' are equivalent; most relevant for us is a notion of strong spatial mixing implying a zero-free region.  We ultimately show that single-site SSM implies a uniform zero-free region for cubes $\Lambda = [-n,n]^d$. 
	
	As in the case of Dobrushin--Shlosman, our proof is by induction and deals
	with ratios of partition functions applied for complex activities, but we have two serious obstacles: first, since we are working in the continuum, the nature of such an induction is significantly less clear; second, our interactions need not have bounded range.  This requires us to introduce various intermediate steps that must be proven simultaneously in our inductive proof in order to overcome these two challenges.  The core of our inductive proof is Proposition \ref{prop:expectations} containing four statements concerning zero-freeness of partition functions along with approximations of the influence of altering the activity function $\blambda$ by introducing a single point.  We then prove Proposition \ref{prop:expectations} with a fairly involved four-part strong inductive step (see Figure \ref{img:proofstruct} for a schematic depiction of the structure of the inductive step). Following that, we show analyticity of the infinite-volume pressure by following the standard Lee--Yang approach \cite{yang1952statistical}, which the same duo subsequently applied to the ferromagnetic Ising model in their Lee--Yang paper \cite{lee1952statistical}.  
	
	In Section \ref{sec:block-dynamics} we show that optimal mixing of block dynamics follows from single-site SSM. For this, we use a version of \emph{path coupling} due to Bubley--Dyer \cite{bubley1997path}.  For completeness, we provide a self-contained proof of a basic version of Bubley--Dyer's path coupling lemma (Lemma \ref{lemma:path_coupling}) making use of Kantorovich--Rubinstein duality to witness the Wasserstein metric in terms of discrepancy of the expectation of Lipschitz functions.  To apply the method of path coupling, we follow \cite{michelen2022strong,helmuth2022correlation} and adapt the ideas of \cite{dyer2004mixing} to the continuum.  As for the previous sections, the lack of a finite-range assumption complicates the analysis.
	
	Finally, in Section \ref{sec:lambda-spec} we show how to compute $\lambda_\spec$ in the case of the hard-sphere model and the Gaussian core model. We first show that in the case of a translation-invariant potential, one can interpret $\lambda_\spec$ in terms of the most negative value of the fourier transform of $x \mapsto 1 - e^{-\phi(x,0)}$.  
	From there we compute that $\lambda_\spec$ for the hard sphere model may be expressed in terms of an optimization problem related to Bessel functions.  
	Some basic calculus and properties of Bessel functions carried out in Appendix \ref{sec:bessel} allow us to compute this minimum precisely in terms of Bessel functions and their zeros.  
	This establishes Proposition \ref{prop:hard-spheres} along with its more exact version \eqref{eq:hard-sphere-lambda-spec-asympt}.  
	The analysis for the Gaussian core model is simpler, and consists of applying the Fourier transform to the individual terms of the Taylor expansion of $1 - e^{-\potential(x, 0)}$.
	This leaves us with an alternating series, which we bound by eliminating all but one negative term.
	
	We now begin with some preliminaries, where we define our central objects more formally and introduce some elementary properties.

	\section{Preliminaries}
	We write $\borel$ for the Borel sets on $\R^{d}$, $\borel_b$ for the bounded Borel sets and $\vol$ for the Borel--Lebesgue measure on $(\R^d, \borel)$.
	We denote by $\pointsets$ the set of locally finite and by $\pointsets_f$ the finite point sets from $\R^d$.
	Further, given $\Delta \in \borel$, we write $\pointsets_{\Delta}$ for all $X \in \pointsets$ with $X \subseteq \Delta$.
	We write $\events$ for the $\sigma$-algebra on $\pointsets$ generated by the maps $N_{\Delta}: \pointsets \to \N_0, X \mapsto \size{X \cap \Delta}$ for $\Delta \in \borel_b$, and denote by $\events_{\Delta}$ the trace of $\pointsets_{\Delta}$ in $\events$.

	\subsection{Finite-volume Gibbs measure}\label{sec:finite-volume-prelims}
	A repulsive pair potential on $\R^d$ is a symmetric, measurable function $\potential: \R^d \times \R^d \to \R_{\ge 0} \cup \{\infty\}$.
	Given a repulsive pair potential $\potential$, we define the associated energy function as  
	\[
	H(X) \coloneqq \sum_{\{x, y\} \subseteq X} \potential(x, y) 
	\]
	for finite point configuration $X \in \pointsets_f$.
	With some abuse of notation, we extend this definition to finite tuples $\V{x} = (x_1, \dots, x_k) \in (\R^d)^k$, writing $H(\V{x}) = H(\{x_1, \dots, x_k\})$.

	We will impose the following decay property on the potential $\phi$.
	\begin{definition}[Exponential decay of the potential]\label{def:potential_decay}
		We say that a repulsive potential \emph{decays exponentially} with constants $B, \alpha > 0$ if
		\[
		1 - e^{-\potential(x, y)} \le B  e^{- \alpha  \cdot\dist(x, y)} \text{ for } x, y \in \R^d.
		\]
	\end{definition}

	Fix some repulsive potential $\potential$.
	Given a bounded measurable region $\region \subset \R^d$ and a bounded and measurable activity function $\activities: \R^d \to \R_{\ge 0}$, we write $\mu_{\region, \activities}$ for the Gibbs distribution on $\region$ with activity function $\activities$, i.e.,
	\[
	\mu_{\region, \activities} (\diff X) \propto \bigg(\prod_{x \in X} \activities(x)\bigg) \cdot e^{-H(X)} P_{\region}(\diff X),
	\]
	where $P_{\region}$ is the distribution of a Poisson point process of constant intensity $1$ on $\region$.
	The right-hand side of the expression above is normalized by $\frac{e^{\vol(\region)}}{Z_{\region}(\activities)}$, where 
	\[
	Z_{\region}(\activities) \coloneqq 1 + \sum_{k \ge 1} \frac{1}{k!} \int_{\region^k} \activities^{\V{x}} \cdot e^{-H(\V{x})} \diff \V{x} ,
	\]
	is called the partition function on $\region$, and where we abbreviate $\activities^{\V{x}} \coloneqq \prod_{1 \le j \le k} \activities(x_j)$ for $\V{x} = (x_1, \dots, x_k) \in (\R^d)^k$.
	
	We write $\E_{\region, \activities}[f]$ and $\Var_{\region, \activities}[f]$ for the expectation and variance, and $\Cov_{\region, \activities}(f, g)$ for the covariance of measurable functions $f, g : \pointsets \to \C$ with respect to $\mu_{\region, \activities}$. 
	Lastly, if the activity function is constant $\activities \equiv \activity$, we simplify notation and write $Z_{\region}(\activity)$, $\mu_{\region, \activity}$, $\E_{\region, \activity}$ and $\Var_{\region, \activity}$.

	We will also make use of the fact that repulsive Gibbs measures are stochastically dominated by Poisson point processes (see, e.g., \cite[Lemma 2]{dereudre2019introduction}) and refer to this property as \emph{Poisson domination}: 
	
	\begin{lemma}\label{lem:Poisson-domination}
		Fix a repulsive potential $\phi$, a bounded measurable region $\Lambda \subset \R^d$ and a measurable activity function $\blambda : \R^d \to \R_{\geq0}$.  Then the Gibbs point process $\bX$ attained from sampling $\mu_{\Lambda,\blambda}$ is stochastically dominated by the Poisson process $P_{\Lambda,\blambda}$ of intensity measure $\blambda$ in $\Lambda$.
	\end{lemma}

	\subsection{Pinnings, finite-volume DLR equation, and correlation functions} \label{sec:pinnings}
	Given a point configuration $X \in \pointsets$, we write 
	\[
	\activities_X: \R^d \to \R_{\ge 0}, \quad y \mapsto \activities(y) \cdot \prod_{x \in X} e^{-\potential(x, y)}
	\]
	for the activity function $\activities$ under pinning $X$.
	If $X = \{x\}$ for some $x \in \R^d$, we simply write $\activities_x \coloneqq \activities_{\{x\}}$.
	
	Pinnings can be used to conveniently express conditional expectations in Gibbs point processes.
	In particular, let $\region' \subseteq \region \subset \R^d$ be bounded and measurable, and let $\bX \sim \mu_{\region, \activities}$, then $\mu_{\region', \activities_{\bX \setminus \region'}}$ is a regular conditional distribution for $\bX \cap \region'$ given $\bX \setminus \region'$.
	Equivalently, for every measurable function $f$ such that $\E_{\region, \activities}[f]$ exists, it holds that
	\begin{align}
		\E_{\region, \activities}[f] = \E_{\bX \sim \region, \activities}[\E_{\bY \sim \region', \activities_{\bX \setminus \region'}}[f(\bY \cup (\bX \setminus \region'))]] , \label{eq:finite_DLR}
	\end{align}
	which is known as the finite-volume DLR equation. We further use the Georgii, Nguyen, Zessin (GNZ) equation (see e.g.\  \cite{daley2008introduction}) stating that for every measurable $f: \R^d \times \pointsets \rightarrow \R$, we have 
	\begin{align}\label{eq:gnz}
		\E_{\bX \sim \region, \activity}\Big[\sum_{x \in \bX}f(x, \bX)\Big] = \int_{\region} \activities(x) \ \E_{\bX\sim\region, \activities} \big[ f(x, \bX \cup \{x\}) e^{-\sum_{y \in \bX}\phi(x, y)}\big] \diff x.
	\end{align}
	
	One consequence of the GNZ equations \eqref{eq:gnz} is that we can write the\footnote{\emph{A priori}, for a point process, the $1$-point correlation function is defined as an element of $L^1$ and hence is only defined almost-everywhere; for Gibbs point processes, the GNZ equations assure that the function in \eqref{eq:1-point-identity} can be taken as the correlation function and as such is defined everywhere.} $1$-point correlation function as 
	\begin{equation}\label{eq:1-point-identity}
		\rho_{\Lambda,\blambda}(v) = \blambda(v) \E_{\bX \sim \Lambda,\blambda}\left[\exp\left(-\sum_{x \in \bX} \phi(v,x) \right) \right]
	\end{equation}
	(see e.g.\ \cite[Eq.~(11)]{michelen2022strong}).  Iterating the GNZ equations yields a similar identity for the $k$-point correlation function and as such allow us to write the $k$-point correlation function in terms of $1$-point corrleation functions of pinned activities: for points $v_1,\ldots,v_k$ we have the identity \begin{equation}\label{eq:multipoint-correlation-identity}
		\rho_{\Lambda,\blambda}(v_1,\ldots,v_k) = \prod_{j = 1}^k \rho_{\Lambda,\blambda_{v_1,\ldots,v_{j-1}}}(v_j)
	\end{equation}
	which is the displayed equation in the proof of Corollary 15 in \cite{michelen2022strong}. 
	
	\subsection{Infinite-volume Gibbs measures and the infinite-volume pressure}
	
	There are many equivalent definitions of infinite-volume Gibbs measure (see \cite{ruelle1969statistical} for various notions and their equivalences, or \cite{jansen2018gibbsian} for a modern survey).  For us, the most convenient definition will be that an infinite-volume Gibbs measure is one that satisfies a version of the DLR equations similar to \eqref{eq:finite_DLR}.
	
	\begin{definition}[Infinite-volume Gibbs measure]\label{def:dlr}
		Given a repulsive potential $\potential$ and an activity function $\activities$, we say a probability measure $\mu$ on $(\pointsets, \events)$ satisfies the Dobrushin-Lanford-Ruelle (DLR) equations if, for all non-negative measurable functions $f: \pointsets \to \R_{\ge 0}$ and all bounded measurable $\region \subset \R^d$, it holds that
		\[
		\E_{\mu}[f] = \E_{\bX \sim \mu}\big[ \E_{\bY \sim \mu_{\region, \activities_{\bX \setminus \region}}}[f(\bY \cup (\bX \setminus \region))]\big].
		\]
		Every $\mu$ satisfying the DLR equations is called an \emph{infinite-volume Gibbs measure} compatible with $\potential$ and $\activities$.
	\end{definition}
	
	Equivalently, one could define infinite-volume Gibbs measure in terms of the GNZ equations \eqref{eq:gnz}.  As a consequence, one obtains identities similar to \eqref{eq:1-point-identity} and \eqref{eq:multipoint-correlation-identity}.  In particular, if $\mu$ is an infinite-volume Gibbs measure compatible with $\phi$ and $\blambda$, then its $k$-point correlation function is
	\begin{equation}\label{eq:k-pt-correlation-definition}
		\rho_{\mu,\blambda}(v_1,\ldots,v_k) = \prod_{j = 1}^k \blambda(v_j)  e^{- \sum_{1 \leq i < j \leq k} \phi(v_i,v_j)} \E_{\bY \sim \mu} \left[\exp\left(- \sum_{j = 1}^k \sum_{y \in \bY} \phi(v_j,y) \right) \right]
	\end{equation}
	(see \cite[Prop.~5.12]{jansen2018gibbsian}).  When $\mu$ is the unique infinite-volume Gibbs measure compatible with $\phi$ and $\blambda$ we simply write $\rho_{\blambda} = \rho_{\mu,\blambda}$.

	\subsection{Birth-death dynamics}\label{sec:birth-death-dynamics}

	Spatial birth-death dynamics are a natural Markov process associated with a potential $\potential$, an activity function $\activities$, and a bounded region $\region \subseteq \R^d$.
	For a function $f$ on $\pointsets$, we define the operators $D_x^+$ and $D_x^-$ as follows.
	\[
	D_x^{+} f (X) \coloneqq f(X \cup \{x\}) - f(X) \,,\hspace{1em} D_x^{-} f (X) \coloneqq f(X \setminus \{x\}) - f(X) \hspace{1em}\text{ for } x \in \R^d, X \in \pointsets.
	\]
	Given a repulsive potential $\potential$, the birth-death dynamics for an activity function $\activities$ on a bounded, measurable region $\region \subset \R^d$ is determined by the generator
	\begin{equation}\label{eq:generator-definition}
		L_{\region, \activities} f (X) \coloneqq \sum_{x \in X} D_x^- f (X) + \int_{\region} \activities(x)  e^{-D_x^+ H(X)} D_x^+ f (X)\, \diff x
	\end{equation}
	acting on $L^2(\mu_{\region, \activities})$ with dense domain
	\[
	\domain_0(L_{\region, \activities}) \coloneqq \{f \in L^2(\mu_{\region, \activities}) \mid \exists C \ge 0: \absolute{f} \le C \text{ and } f(X) = 0 \text{ for all } X \in \pointsets \text{ with } \size{X} > C\}.
	\]
	The Dirichlet form associated with $L_{\region, \activities}$ is given by
	\begin{equation}\label{eq:dirichlet-def}
		\dirichlet_{\region, \activities}(f, g) = \innerProduct{- L_{\region, \activities} f}{g}_{\mu_{\region, \activities}} \text{ for } f, g \in \domain_0(L_{\region, \activities}),
	\end{equation}
	and we write $\dirichlet_{\region, \activities}(f) = \dirichlet_{\region, \activities}(f, f)$.
	The following useful properties of $L_{\region, \activities}$ and $\dirichlet_{\region, \activities}$ were proven in \cite[Proposition 2.1]{bertini2002spectral}.
	\begin{enumerate}[(1)]
		\item $\dirichlet_{\region, \activities}$ is closable, and its closure is associated with a self-adjoint extension of $L_{\region, \activities}$ with respect to $\innerProduct{\cdot}{\cdot}_{\mu_{\region, \activities}}$.
		By abuse of notation, we write $\dirichlet_{\region, \activities}$ and $L_{\region, \activities}$ for the extensions and denote their respective domains by $\domain(L_{\region, \activities})$ and $\domain(\dirichlet_{\region, \activities})$.
		\item for every $p \in \N \cup \{\infty\}$, $P_t^{\region, \activities} \coloneqq e^{t L_{\region, \activities}}$ for $t \in \R_{\ge 0}$ defines a Markov semi-group on $L^p(\mu_{\region, \activities})$. 
		\item For every $t \in \R_{\ge 0}$, $P_t^{\region, \activities}$ is self-adjoint on $L^2(\mu_{\region, \activities})$.
	\end{enumerate}
	Moreover, we have the following useful connection between a Poincaré inequality for $L_{\region, \activities}$ and contraction of the semi-group $(P_t^{\region, \activities})_{t \ge 0}$ on $L^2(\mu_{\region, \activities})$.
	
	\begin{lemma}[{\cite[Theorem 4.2.5]{bakry2013analysis}}] \label{lemma:L2_contraction}
		Suppose there is some $\gamma < \infty$ such that $\Var_{\region, \activities}(f) \leq \gamma \cdot \dirichlet_{\region, \activities}(f)$ for all $f \in \domain(\dirichlet_{\region, \activities})$.
		Then, for all $f \in L^2(\mu_{\region, \activities})$ and $t \ge 0$, it holds that 
		\[
		\norm{P_t^{\region, \activities} f - \E_{\region, \activities}[f]}_{L^2(\mu_{\region, \activities})} \le \sqrt{\Var_{\region, \activities}(f)} \ e^{-t/\gamma} .
		\]
	\end{lemma}
	
	Furthermore, the following explicit representation of the Dirichlet form will be useful.
	\begin{lemma}\label{lem:dirichlet_explicit}
		For all $f,g \in \domain(\dirichlet_{\region, \activities})$, we have 
		$$
		\dirichlet_{\region, \activities}(f, g) = \int_{\region} \activities(x) \E_{\region, \activities} \big[  e^{-D_x^+ H} D_x^+f D_x^+g  \big] \diff x.
		$$
	\end{lemma}
	\begin{proof}
		By definition of the generator $L$, we have 
		\begin{align*}
			-\dirichlet_{\region, \activities}(f, g) = \E_{\bX \sim \region, \activity}\Big[\sum_{x \in \bX}g(\bX)D_x^-f(\bX)\Big] + \int_{\region} \activities(x) \E_{\bX\sim\region, \activities} \big[  e^{-D_x^+ H(\bX)} g(\bX)D_x^+f(\bX)  \big] \diff x.
		\end{align*} Applying the GNZ equation \eqref{eq:gnz} to the first term, we get 
		\begin{align*}
			&\E_{\bX \sim \region, \activity}\Big[\sum_{x \in \bX}g(\bX)D_x^-f(\bX)\Big] = \int_{\region} \activities(x) \E_{\bX\sim\region, \activities} \big[  e^{-D_x^+ H(\bX)} g(\bX \cup \{x\})D_x^-f(\bX \cup \{x\})  \big] \diff x.
		\end{align*}
		Noting that $D_x^-f(\bX \cup \{x\}) = -D_x^+f(\bX)$ and plugging the result back into our expression for $-\dirichlet_{\region, \activities}(f,g)$ then yields the lemma. 
	\end{proof}

	\subsection{Elementary inequalities and Poisson moments}
	We will make frequent use of the following standard elementary inequality:
	\begin{lemma}
		\label{lem:complexproductdifference}
		Let $\{z_j\}_{j =1 }^n$ and $\{w_j\}_{j =1 }^n$ be complex numbers with $|z_j|, |w_j| \le \theta$ for all $j$. Then,
		$$
		\left| \prod_{j=1}^n z_j - \prod_{j=1}^n w_j \right| \le \theta^{n-1} \sum_{j=1}^n |w_j - z_j|.
		$$
	\end{lemma}
	\begin{proof}
		We proceed by induction on $n$, noting the $n = 1$ case is trivial. For tuples $\{z_j\}_{j =1 }^{n+1}$ and $\{w_j\}_{j =1 }^{n+1}$ with $|z_j|, |w_j| \le \theta$ for all $j \leq n+1$ define $Z = \prod_{j = 1}^{n} z_j$ and $W = \prod_{j = 1}^{n} w_j$ and note \begin{equation*}
			\left|\prod_{j = 1}^{n+1} z_j - \prod_{j = 1}^{n+1} w_j \right| \leq |z_{n+1}| \cdot |Z - W| + |W| \cdot |z_{n+1} - w_{n+1}|\,.
		\end{equation*}
		The proof follows by induction.
	\end{proof}
	
	We will also make frequent use of various exponential moment identities for Poisson random variables due to our frequent use of Poisson domination (Lemma \ref{lem:Poisson-domination}). 
	
	\begin{lemma}\label{lemma:poisson_moments}
		Let $\gamma \ge 0$ and let $\xi \sim \Pois(\gamma)$.
		For all $b \in \R$ and $K \in \N_0$ it holds that:
		\begin{enumerate}[(1)]
			\item\label{lemma:poisson_moments:exp} $\E[ e^{b\xi}] =  \exp((e^b-1) \gamma)$ and $\displaystyle \E[\ind{\xi > K} e^{b\xi}] =  \exp((e^b-1) \gamma) \cdot \P(\xi'>K)$, where $\xi' \sim \Pois(\gamma e^b)$.
			\item \label{lemma:poisson_moments:linexp}  $\E[\xi e^{b\xi}] = \gamma e^b \exp((e^b-1) \gamma)$.
			\item \label{lemma:poisson_moments:quadexp}$\displaystyle \E[\xi^2 e^{b\xi}] =  (\gamma e^b + \gamma^2 e^{2b}) \exp((e^b-1) \gamma)$. 
			\item \label{lemma:poisson_moments:cubeexp}$\displaystyle \E[\xi^3 e^{b\xi}] =  (\gamma e^b + 3\gamma^2 e^{2b} + \gamma^3 e^{3b}) \exp((e^b-1) \gamma)$. 
		\end{enumerate}
	\end{lemma}
	\begin{proof}
		For \ref{lemma:poisson_moments:exp}, we compute
		\begin{align*}
			\E[\ind{\xi > K} e^{b \xi}] 
			= e^{-\gamma} \sum_{k > K} \frac{\gamma^k}{k!}  e^{b k} 
			= e^{-\gamma} e^{\gamma e^b}\cdot  \P(\xi'>K).
		\end{align*}
		To compute $\E[e^{b \xi}]$ we follow the same calculation without the indicator. 
		
		For \ref{lemma:poisson_moments:linexp}, we note that
		\begin{align*}
			\E[\xi e^{b \xi}] 
			= e^{-\gamma} \sum_{k \geq 1} \frac{\gamma^k}{k!} k e^{b k} 
			= \gamma e^b e^{-\gamma} \sum_{k \geq 1} \frac{\gamma^{k-1}}{(k-1)!} e^{b (k-1)} 
			= \gamma e^b e^{-\gamma} e^{\gamma e^b}.
		\end{align*}
		Next, for \ref{lemma:poisson_moments:quadexp}, we first write $\E[\xi^2 e^{b\xi}] = \E[\xi (\xi-1) e^{b\xi}] + \E[\xi e^{b\xi}]$.
		For calculating the second term, we use \ref{lemma:poisson_moments:linexp}.
		For the first term, we have
		\[
		\E[\xi (\xi-1) e^{b\xi}] = e^{-\gamma} \sum_{k \ge 2} \frac{\gamma^k}{k!} k (k-1) e^{bk} = \gamma^2 e^{2b} e^{-\gamma} e^{\gamma e^b}.
		\]
		Lastly, for \ref{lemma:poisson_moments:cubeexp}, we write $\xi^3 = \xi(\xi-1)(\xi-2) + 3\xi(\xi-1) + \xi$ and then proceed just like before.
	\end{proof}
	
	\section{Spatial mixing from spectral gap} \label{sec:spatial-mixing}
	In this section, we prove that a uniform Poincaré inequality for the spatial birth-death dynamics implies single-site SSM.
	The main result of this section is the following theorem.
	\begin{theorem} \label{thm:ssssm}
		Let $\activity \in \R_{\ge 0}$ and let $\potential$ be a repulsive potential that decays exponentially with constants $B$ and $\alpha$ (see Definition \ref{def:potential_decay}).
		Suppose there is a constant $\gamma < \infty$ such that, for all activity functions $\activities \le \activity$ and all bounded, measurable $\region \subset \R^d$, it holds that
		\[
		\Var_{\region, \activities}(f) \leq \gamma \ \dirichlet_{\region, \activities}(f)  \text{ for all } f \in \domain(\dirichlet_{\region, \activities}).
		\]
		Then $\potential$ satisfies single-site SSM up to $\activity$ for constants $C \coloneqq C(\gamma,\activity, B, \alpha, d) < \infty$ and $m \coloneqq m(\gamma, \activity, B, \alpha) > 0$ (see Definition \ref{def:spatial_mixing}).
	\end{theorem}
	
	Recall that single-site SSM is ultimately a statement at projected total variation distance.  In particular, recall that given $\Delta \in \borel$ we define the projected total variation distance $\|{P - Q}\|_{\Delta}$ between probability distributions $P, Q$ on $(\pointsets, \events)$ as the total variation distance between the pushforwards of $P$ and $Q$ under the map $\pointsets \to \pointsets, X \mapsto X \cap \Delta$.  To simplify notation throughout this section, we will write $\norm{\cdot}_{p}$ for $\norm{\cdot}_{L^p(\mu_{\region, \activities})}$ for every $p \in \N \cup \{\infty\}$ whenever the region $\region \subset \R^d$ and the activity function $\activities$ are clear from the context.

	At the core of our proof of Theorem \ref{thm:ssssm} is the derivation of a correlation-decay result that generalizes that in \cite[Theorem 4.1]{bertini2002spectral} in two ways. 
	Firstly, it applies without assuming a bounded range of the potential and only needs an exponential decay as in Definition \ref{def:potential_decay}. Secondly, instead of only applying to local functions, it applies to functions $f,g$ on $\pointsets$ that are ``approximately local'' on a region $\Delta$ in the sense that $|f(X \cup \{x\}) - f(X)|$ decays exponentially in the distance $\dist(x, \Delta)$ for all $x \notin \Delta$ and $X \in \pointsets$. By contrast, an (ordinary) $\Delta$-local function $f$ is required to satisfy $|f(X \cup \{x\}) - f(X)| = 0$ whenever $x \notin \Delta$. 
	
	\begin{definition}[Approximately local functions]\label{def:approxlocal}
		Given a bounded region $\Delta \subset \R^d$ and $K, \kappa \in \mathbb{R}$ such that $0 \le K < \infty$ and $\kappa > 0$, we say that a function $f$ on $\pointsets$ is $(\Delta,\kappa, K)$-local if for all $x \notin \Delta$,
		\[
		\|D_x^+ f\|_\infty \leq K e^{- \kappa\cdot\dist(x,\Delta)}\,.
		\]
		We note that a $\Delta$-local function is $(\Delta,\kappa, 0)$-local for any choice of $\kappa > 0$, and that a $(\Delta,\kappa, K)$-local function, it is also $(\Delta',\kappa', K')$-local for every $\Delta' \supseteq \Delta$, $\kappa' \le \kappa$ and $K' \ge K$.
	\end{definition}
	To introduce our correlation decay result, we further need to introduce the semi-norm $\tnorm{\cdot}$.
	Given a measurable function $f$ on $\pointsets$, we define 
	\[
	\tnorm{ f } \coloneqq \int \|D_x^+ f\|_\infty \,dx\,.
	\]
	Our correlation decay result now reads as follows.
	\begin{proposition}\label{prop:covariance}
		Given an activity function $\activities \le \activity$, a repulsive potential $\potential$ satisfying exponential decay for constants $\alpha > 0$ and $B < \infty$ (see Definition \ref{def:potential_decay}), and a bounded measurable region $\region \subset \R^d$, suppose there is a constant $\gamma < \infty$ such that 
		\[
		\Var_{\region, \activities}(f) \leq \gamma \ \dirichlet_{\region, \activities}(f) \text{ for all } f \in \domain(\dirichlet_{\region, \activities})
		\]
		then, for each $\kappa > 0,  K \in [0,\infty)$, there are constants $m_0 = m_0(\gamma, \activity, \alpha, B, \kappa) > 0$ and $C_0 = C_0(\gamma, \activity, \alpha, B, \kappa) < \infty$ so that the following holds.  
		For all bounded, measurable $\region_f,\region_g \subset \region$ with $\dist(\region_f, \region_g) > 0$ and all measurable functions $f,g$ that are $(\region_f,\kappa, K)$- and $(\region_g,\kappa, K)$-local, we have 
		\[
		\absolute{\Cov_{\region, \activities}(f,g)} \leq C_0  (\tnorm{f} \cdot \tnorm{g} + K (\tnorm{f} + \tnorm{g}) + \sqrt{\Var_{\region,\activities}(f)} \sqrt{\Var_{\region,\activities}(g)})  e^{- m_0  \cdot\dist(\Lambda_f,\Lambda_g)}\,.
		\]
	\end{proposition}

	The main difference between our proof of Proposition \ref{prop:covariance} and that of \cite[Theorem 4.1]{bertini2002spectral} lies in proving the following lemma, which generalizes \cite[Lemma 4.2]{bertini2002spectral}.
	\begin{lemma} \label{lemma:unbounded_4.2}
		In the setting of Proposition \ref{prop:covariance}, there are constants $M = M(\activity, \alpha, B, d) < \infty$, $C_1 = C_1(\activity,B,\alpha,\kappa,d) < \infty$ and $m_1 = m_1(\alpha, \kappa) > 0$ such that, for all $t \ge 0$
		\[
		\big| \E_{\region, \activities}\big[P_t^{\region, \activities}(f g) - P_t^{\region, \activities}f \ 
		P_t^{\region, \activities}g\big]\big| \le C_1 (\tnorm{f} \cdot  \tnorm{g} + K  (\tnorm{f} + \tnorm{g})) \cdot e^{t  M - m_1  \cdot\dist(\region_f, \region_g)} .
		\]
	\end{lemma}
	
	We first deduce Proposition \ref{prop:covariance} and Theorem \ref{thm:ssssm} and then prove Lemma \ref{lemma:unbounded_4.2}.
	
	\newcommand{\infnorm}[1]
	{\big\| #1 \big\|_\infty}
	\newcommand{\twonorm}[1]
	{\big\| #1 \big\|_2}
	
	\subsection{Proving Proposition \texorpdfstring{\ref{prop:covariance}}{3.3} and Theorem  \texorpdfstring{\ref{thm:ssssm}}{3.1}}

	The proofs of Proposition \ref{prop:covariance} and Theorem \ref{thm:ssssm} using Lemma \ref{lemma:unbounded_4.2} will be similar to the proof of \cite[Theorem 4.1]{bertini2002spectral}.  We begin with Proposition \ref{prop:covariance}:
	
	\begin{proof}[Proof of Proposition \ref{prop:covariance}]
		Without loss of generality, we may assume $\E_{\region, \activities}[f] = \E_{\region, \activities}[g] = 0$.
		For every $t \ge 0$, using that $P_t^{\region, \activities}$ is self-adjoint with respect to $\innerProduct{\cdot}{\cdot}_{\mu_{\region, \activities}}$, we then get
		\[
		\Cov_{\region, \activities}(f, g) = \E_{\region, \activities}[f  g] =  \E_{\region, \activities}[P_t^{\region, \activities} (f g)] .
		\]
		Using the triangle inequality and the Cauchy--Schwarz inequality yields
		\begin{align*}
			\big| \E_{\region, \activities}[P_t^{\region, \activities} (f g)] \big| &= \big| \E_{\region, \activities}[P_t^{\region, \activities} f \ P_t^{\region, \activities} g] + \E_{\region, \activities}[P_t^{\region, \activities} (f g) - P_t^{\region, \activities} f \ P_t^{\region, \activities} g] \big|\\
			&\le \twonorm{ P_t^{\region, \activities} f} \twonorm{ P_t^{\region, \activities} g} + \big| \E_{\region, \activities}\big[P_t^{\region, \activities} (f g) - P_t^{\region, \activities} f \  P_t^{\region, \activities} g\big] \big|.
		\end{align*}
		Using the fact that $P_t^{\region, \activities}$ satisfies a Poincar\'e inequality with constant $\gamma$, we can apply Lemma \ref{lemma:L2_contraction} to the first term, and Lemma \ref{lemma:unbounded_4.2} to the second term to obtain
		\[
		\absolute{\Cov_{\region, \activities}(f, g)}
		\le \norm{f}_2 \norm{g}_2  e^{-2 t/\gamma} + C_1  (\tnorm{f} \cdot \tnorm{g} + K  (\tnorm{f} + \tnorm{g}))  e^{t  M - m_1  \cdot\dist(\region_f, \region_g)}.
		\]
		Choosing $t = \frac{m_1}{2M}  \dist(\region_f, \region_g)$, $m_0 = \min(\frac{m_1}{M \gamma}, \frac{m_1}{2})$ and $C_0 = \max(1, C_1)$ concludes the proof.
	\end{proof}
	
	We use Proposition \ref{prop:covariance} to prove the main theorem of this section.  The presence of the term $\|f\|_2 \|g\|_2$ on the right-hand side of Proposition \ref{prop:covariance} will yield a bound of the form $$\| \mu_{\Lambda,\blambda} - \mu_{\Lambda,\blambda_x}\|_{\Lambda'} \leq C(\vol(\Lambda') + C) e^{-m\cdot\dist(x,\Lambda')}$$
	which will ostensibly only yield single site SSM for sets $\Lambda'$ of volume bounded below.  We will see that using the strategy of \cite{michelen2022strong} will upgrade single site SSM to all sets $\Lambda'$.
	
	\begin{lemma}\label{lem:upgrade-to-small-vol}
		Suppose that $\phi$ satisfies single site strong spatial mixing up to some activity $\lambda$ for all bounded measurable $\Lambda'$ with $\vol(\Lambda') \geq 1$.  Then $\phi$ satisfies single site strong spatial mixing up to $\lambda$ for all bounded and measurable $\Lambda'$.
	\end{lemma}
	
	Lemma \ref{lem:upgrade-to-small-vol} will be proved in parallel to strong spatial mixing (Proposition \ref{prop:SSM}) and so we defer both proofs to  Section \ref{sec:cor-SSM}.
	We now prove Theorem \ref{thm:ssssm}:

	\begin{proof}[Proof of Theorem \ref{thm:ssssm}]
		By Lemma \ref{lem:upgrade-to-small-vol} we may assume that $\vol(\Lambda') \geq 1$.  
		Note that
		\[
		\|{\mu_{\region, \activities_x} - \mu_{\region, \activities}}\|_{\region'} = \sup_{A \in \events} \absolute{\E_{\bY \sim \region, \activities_x}[\ind{\bY \cap \region' \in A}] - \E_{\bY \sim \region, \activities}[\ind{\bY \cap \region' \in A}]} .
		\]
		Given an event $A \in \events$, we then write
		\begin{align*}
			&\absolute{\E_{\bY \sim \region, \activities_x}[\ind{\bY \cap \region' \in A}] - \E_{\bY \sim \region, \activities}[\ind{\bY \cap \region' \in A}]} 
			\\ &\hspace{2cm}= \absolute{\frac{\E_{\bY \sim \region, \activities}[\ind{\bY \cap \region' \in A}  \exp(-\sum_{y \in \bY} \potential(x, y))]}{\E_{\bY \sim \region, \activities}[\exp(-\sum_{y \in \bY} \potential(x, y))]} - \E_{\bY \sim \region, \activities}[\ind{\bY \cap \region' \in A}]} = \frac{\absolute{\Cov_{\region, \activities}(f, g)}}{\E_{\region, \activities}[g]} ,
		\end{align*}
		where $f: Y \mapsto \ind{Y \cap \region' \in A}$ and $g: Y \mapsto \exp(-\sum_{y \in Y} \potential(x, y))$.
		
		To bound the denominator from below, note that $Y \mapsto \exp(-\sum_{y \in Y} \potential(x, y))$ is non-increasing with respect to inclusion.
		Thus, by Poisson domination (Lemma \ref{lem:Poisson-domination}), the Laplace functional of Poisson point processes and the exponential decay of $\potential$, we have
		\begin{align}\label{eq:lower-bound-exponential-added-energy}
			\E_{\region, \activities}[g]
			\ge \exp\left( - \activity \int_{\region} 1 - e^{-\potential(x, y)} \diff y\right)
			\ge e^{-\activity  B  C_{\alpha, d}} ,
		\end{align}
		where $C_{\alpha, d} = \int_{\R^d} e^{-\alpha  \norm{y}} \diff y < \infty$.
		
		It remains to upper bound the numerator, for which we aim to use Proposition \ref{prop:covariance}.
		To this end, note that $\sqrt{\Var_{\region,\activities}(f)}  \le \norm{f}_{\infty} \le 1$ and $\sqrt{\Var_{\region,\activities}(g)}  \le \norm{g}_{\infty} \le 1$.
		Since $f$ is $\region'$-local, it holds that $\tnorm{f} \le 2 \vol(\region')$.
		Finally, using that $\potential$ is repulsive we have $\|D^+_y g\|_\infty \le 1 - e^{-\potential(x, y)}$.
		Since $\potential$ further satisfies exponential decay according to Definition \ref{def:potential_decay}, it is $(\{x\}, \alpha, B)$-local, implying that  $\tnorm{g} \le B  C_{\alpha, d}$.
		Applying Proposition \ref{prop:covariance} and choosing $C < \infty$ large enough and $m > 0$ small enough then proves the theorem.
	\end{proof}

	\subsection{Proving Lemma \texorpdfstring{\ref{lemma:unbounded_4.2}}{3.4}}

	We begin with an identity proven in \cite[Eq.~(4.3)]{bertini2002spectral}, which may be understood as a standard identity for the carr\'e-du-champ operator associated to $L_{\Lambda,\blambda}$:
	
	\begin{lemma}\label{lem:dirichlet-form-curvature}
		For all $f,g \in \domain(\dirichlet_{\Lambda,\blambda})$ we have \begin{align*}
			\E_{\region, \activities}\big[P_t^{\region, \activities}(f  g) - P_t^{\region, \activities}f \ P_t^{\region, \activities}g\big] = -2 \int_0^t \dirichlet_{\region, \activities}(P_s^{\region, \activities}f,P_s^{\region, \activities}g) \diff s\,.
		\end{align*}
	\end{lemma}
	\begin{proof}
		We expand on the proof of \cite[Eq.~(4.3)]{bertini2002spectral}.  We approximate $L_{\Lambda,\blambda}$ by an $L^2(\diff \mu)$ bounded operator where we write $\mu = \mu_{\Lambda,\blambda}$.  For each $k > 0$ define $L_k$ via $L_k = - \int_0^k s \diff E_s$ where $\{E_s : s \in [0,\infty)\}$ are the spectral projections associated to $-L_{\Lambda,\blambda}$ guaranteed by the spectral theorem.  Define $P^{k}:= e^{t L_k}$.  Define $\Phi(s) = P^k_{t -s}(P_s^k f \cdot P_s^k g)$ and note that since $L_k$ is bounded, $\Phi(s)$ is differentiable.  By the fundamental theorem of calculus we have \begin{align*}
			P_t^k(f) P_t^k(g) - P_t^k(fg) &=  \Phi(t) - \Phi(0) = \int_0^t \Phi'(s)\diff s \\ 
			&= -\int_0^t P_{t - s}^k\left[ L_k (P_s^k f \cdot P_s^k g) - P_s^k g \cdot L_k P_s^k f - P_s^k f \cdot  L_k  P_s^k g \right]\diff s
		\end{align*}
		where we use that $P_s^k$ and $L_k$ commute. Using the fact that $L_k$ is self-adjoint in $L^2(\diff \mu)$, we have \begin{align*}
			\E_\mu&\left[P_{t - s}^k\left[ L_k (P_s^k f \cdot P_s^k g) - P_s^k g \cdot L_k P_s^k f - P_s^k f \cdot  L_k  P_s^k g \right] \right] \\
			&\qquad=  \E_\mu\left[ L_k (P_s^k f \cdot P_s^k g) - P_s^k g \cdot L_k P_s^k f - P_s^k f \cdot  L_k  P_s^k g  \right] \\
			&\qquad= \langle L_k (P_s^k f \cdot P_s^k g), 1\rangle_\mu - \langle P_s^k g, L_kP_s^k f\rangle_\mu - \langle P_s^k f, L_k P_s^k g\rangle_\mu \\
			&\qquad= -2 \mathcal{E}_k(P_s^k f, P_s^k g)
		\end{align*}
		where $\mathcal{E}_k$ is the Dirichlet form associated to $L_k$.  Taking $k \to \infty$ completes the proof.
	\end{proof}
	
	Combining Lemma \ref{lem:dirichlet_explicit} and  \ref{lem:dirichlet-form-curvature} shows 
	\begin{align}
		\big|\E_{\region, \activities}\big[P_t^{\region, \activities}(f  g) - P_t^{\region, \activities}f \ P_t^{\region, \activities}g\big]\big| \le 
		2 \activity \int_0^t \int_\Lambda F_s(x) G_s(x) \diff x \diff s \label{eq:pure_magic}
	\end{align}
	where $F_t(x) \coloneqq \infnorm{D_x^+ P^{\region, \activities}_t f}$ and $G_t(x) \coloneqq \infnorm{D_x^+ P^{\region, \activities}_t g}$.  We upper bound $F_t$, analogously to \cite[Eq.~(4.8)]{bertini2002spectral}.
	
	\begin{lemma}
		In the notation above we have $$     F_t(x) \le F_0(x) + \int_{0}^{t} \left(F_s(x) + \int_{\region} \activity \infnorm{D_x^+ e^{- D_y^+ H}} F_s(y) \diff y \right) \diff s \,.$$
	\end{lemma}
	\begin{proof}
		This bound is implicit in \cite{bertini2002spectral}.  In particular, if we let $T_x$ denote the operator defined by $T_x f(\omega) = f(\omega \cup \{x\})$ and apply the fundamental theorem of calculus to $P_{t -s}^{\Lambda,\blambda}(D_x^+ P_s^{\Lambda,\blambda} f)$ to obtain \begin{equation*}
			D_x^+ P_t^{\Lambda,\blambda} f = P_t^{\Lambda,\blambda} D_x^+ f + \int_0^t P_{t - s}^{\Lambda,\blambda} \left[-D_x^+ P_s^{\Lambda,\blambda} f + \int_{\Lambda} \blambda(y) (D_x^+ e^{-D_y^+ H}) T_x D_y^+ P_s^{\Lambda,\blambda} f \diff y \right] \diff s
		\end{equation*}
		which is precisely \cite[Eq.~(4.7)]{bertini2002spectral}.  Since $P_t^{\Lambda,\blambda}$ is a contraction in $L^\infty(d\mu_{\Lambda,\blambda})$ we obtain the stated bound.
	\end{proof}

	Using that $\infnorm{D_x^+ e^{- D_y^+ H}} \le 1 - e^{- \potential(x, y)}$ for repulsive $\potential$, we get
	\begin{equation}\label{eq:F-t-pointwise}
		F_t(x) \le F_0(x) + \int_{0}^{t} \left[F_s(x) + (\Gamma F_s)(x) \right] \diff s ,
	\end{equation}
	where $\Gamma$ is the Hilbert--Schmidt operator with kernel $\gamma(x, y) = \ind{y \in \region} \ \activity  (1 - e^{-\potential(x, y)})$ defined by \begin{equation*}
		(\Gamma F)(x) = \int_{\Lambda} F(y) \gamma(x,y) \diff y\,.
	\end{equation*}
	Iterating \eqref{eq:F-t-pointwise} then yields the following convenient bound on $F_t$ in terms of $F_0$.
	\begin{lemma}
		Iterating the inequality \eqref{eq:F-t-pointwise} yields
		\begin{align} \label{eq:ft_bound}
			F_t(x) \le e^{t} \sum_{k=0}^{\infty} \frac{t^k}{k!} (\Gamma^k F_0)(x) .
		\end{align}  
	\end{lemma}
	\begin{proof}
		Set $G = I + \Gamma$ and note that \eqref{eq:F-t-pointwise} can be rewritten as $$F_t(x) \leq F_0(x) + \int_0^t (G F_s)(x)\, \diff s\,.$$  By induction this implies that for each $k \geq 1$ we have \begin{equation*}
			F_t(x) \leq \sum_{j = 0}^k \frac{t^j}{j!} (G^j F_0)(x) + \int_0^t \frac{(t - s)^k}{k!} (G^{k+1} F_s)(x) \,ds\,.
		\end{equation*}
		Taking $k \to \infty$ and using that $\|G^{k+1} F_s \|_\infty \leq (1 + |\Lambda \lambda|)^{k+1}  \|F_s\|_\infty$ shows that \begin{equation*} F_t(x) \leq  e^{t G} F_0(x) = e^{t} e^{t\Gamma}F_0(x) = e^{t} \sum_{k = 0}^\infty \frac{t^k}{k!} (\Gamma^k F_0)(x)\,. \qedhere \end{equation*}
	\end{proof}
	
	We proceed by observing that $\Gamma^k$ is again a Hilbert--Schmidt operator with kernel $\gamma^k$, i.e.\  recursively defined by
	\[
	(\Gamma^k F)(x) = \int_{\region} \gamma^{k}(x,y) F(y) \diff y \quad \text{ where } \quad 
	\gamma^k(x, y) = \int_{\region} \gamma^{k-1}(x, z) \cdot \gamma(z, y) \diff z .
	\]
	Our last stop before proving Lemma \ref{lemma:unbounded_4.2} is to bound $\gamma^k(x, y)$.

	\begin{lemma} \label{claim:gamma}
		Suppose $\potential$ is repulsive and decays exponentially (see Definition \ref{def:potential_decay}) with constants $B < \infty$ and $\alpha > 0$, then there is some $b \coloneqq b(\activity, B, \alpha, d) < \infty$ so for all $k \geq 1$ we have
		\[
		\gamma^k(x, y) \le b^k e^{- \frac{\alpha}{2}  \dist(x, y)}
		\]
	\end{lemma}
	\begin{proof}
		Set 
		$
		b \coloneqq \max\{\activity  B, \activity B \int_{\R^d} e^{- \frac{\alpha}{2} \norm{x}} \diff x\} .
		$
		We prove the claim by induction on $k$.
		The case $k=1$ follows immediately from the exponential decay of the potential, which yields
		\[
		\gamma(x, y) 
		= \activity  (1 - e^{- \potential(x, y)}) \le \activity B e^{ - \alpha\cdot\dist(x, y)}.
		\]
		For the induction step, we use the induction hypothesis and the exponential decay of the potential to get
		\begin{align*}
			\gamma^k(x, y) 
			= \int_{\region} \gamma^{k-1}(x, z)  \gamma(z, y) \diff z 
			\le b^{k-1}  \activity B \int_{\region} e^{-\frac{\alpha}{2}  \dist(x, z)}  e^{- \alpha\cdot\dist(y, z)} \diff z
		\end{align*}
		By the triangle inequality, we have  $\frac{1}{2}  \dist(x, z) + \dist(y, z) \ge \frac{1}{2} \dist(x, y) + \frac{1}{2}  \dist(y, z)$ and hence
		\[
		e^{-\frac{\alpha}{2}  \dist(x, z)}  e^{- \alpha  \dist(y, z)} \le e^{- \frac{\alpha}{2}  \dist(x, y)} e^{- \frac{\alpha}{2}  \dist(y, z)} 
		\]
		for every $z \in \R^d$, which then yields
		\[
		\gamma^k(x, y) 
		\le b^{k-1} \activity B e^{- \frac{\alpha}{2}  \dist(x, y)} \int_{\region} e^{- \frac{\alpha}{2}  \dist(y, z)}  \diff z
		\le b^k e^{- \frac{\alpha}{2} \dist(x, y)}. \qedhere
		\]
	\end{proof}
	
	In order to apply Lemma \ref{claim:gamma} to bound the expansion in \eqref{eq:ft_bound}, we will need to isolate the $k = 0$ term.  With this in mind, define the operator $\mathcal{K}$ by $$(\mathcal{K} h)(x) = \int_{\Lambda} e^{-\frac{\alpha}{2} \dist(x,y)}h(y) \diff y\,.$$

	We are now ready to prove Lemma \ref{lemma:unbounded_4.2}:
	
	\begin{proof}[Proof of Lemma \ref{lemma:unbounded_4.2}]
		Substituting the bound from Lemma \ref{claim:gamma} into \eqref{eq:ft_bound} yields
		\begin{align}
			F_t(x) \le e^t \sum_{k = 0}^{\infty} \frac{t^k}{k!} \int_{\region} b^k e^{- \frac{\alpha}{2}  \dist(x, y)} F_0(y) \diff y
			\leq e^{t  (1 + b)}(F_0(x) + (\mathcal{K}F_0)(x)) \label{eq:F-diffuse}
		\end{align}  
		and similarly for $G_t$.  We bound \begin{align*}
			\int_0^t \int_\Lambda F_s(x) G_s(x) \,\diff x \,\diff s \leq \int_0^t \int_\Lambda  e^{2(1 + b)s} (F_0(x) + (\mathcal{K} F_0)(x))(G_0(x) + (\mathcal{K} G_0)(x))\, \diff x \,\diff s
		\end{align*}
		
		When multiplying out, there are four terms to integrate, namely $F_0 \cdot G_0$, $(\mathcal{K} F_0)\cdot G_0$, $F_0 \cdot (\mathcal{K} G_0)$ and $(\mathcal{K} F_0) \cdot (\mathcal{K} G_0)$.  We show how to bound the latter term, as the others are similar (or in fact simpler).

		We bound 
		\begin{align*}
			\int_0^t \int_\Lambda e^{2(1 + b)s} (\mathcal{K} F_0)(x) (\mathcal{K} G_0)(x) \,\diff x \,\diff s 
			&=  \int_0^t e^{2 (1+b) s} \iiint_{\Lambda^3}e^{-\frac{\alpha}{2}  \dist(x,y)}   e^{-\frac{\alpha}{2} \dist(x,z)}  \| D_y^+ f\|_\infty  \|D_z^+ g\|_\infty \,dx\,dy\,dz  \,ds \\
			&\leq e^{2 (1+b) t} \iiint_{\Lambda^3}e^{-\frac{\alpha}{2}  \dist(x,y)}   e^{-\frac{\alpha}{2}  \dist(x,z)}  \| D_y^+ f\|_\infty  \|D_z^+ g\|_\infty \,dx\,dy\,dz \,.
		\end{align*}
		We then note that, for $C_{\alpha, d} \coloneqq \int_{\R^d} e^{-\frac{\alpha}{4}  \norm{x}} \diff x < \infty$,
		\begin{equation}
			\int_{\R^d} e^{-\frac{\alpha}{2}  \dist(x,y)} e^{-\frac{\alpha}{2}  \dist(x,z)} \,dx \leq C_{\alpha, d}  e^{-\frac{\alpha}{4}  \dist(y,z)} 
		\end{equation} 
		which follows from bounding $\dist(x,y) + \dist(x,z) \geq \dist(y,z)$ so $\dist(x,y) + \dist(x,z) \geq \frac{\dist(x,y)}{2} + \frac{\dist(y,z)}{2}$. 
		It remains to bound $$\iint_{\region^2} e^{-\frac{\alpha}{4}  \dist(y,z)} \| D_y^+ f\|_\infty \|D_z^+ g\|_\infty \,dy\,dz.$$
		To this end, we split the integral into two cases: either $y$ is close to $\region_f$ and $z$ is close to $\region_g$, or at least one of them is far away from their respective region.
		Formally, set $A = \{x \in \R^d \mid \dist(x, \region_f) \le \dist(\region_f, \region_g)/3\}$ and $B = \{x \in \R^d \mid \dist(x, \region_g) \le \dist(\region_f, \region_g)/3\}$.
		We note that $\R^d \times \R^d = (A \times B) \cup (\R^d \times B^c) \cup (A^c \times \R^d)$. 
		For $y \in A$ and $z \in B$, we have $\dist(y, z) \ge \dist(\region_f, \region_g)/3$. 
		\begin{align}
			\int_B \int_A e^{-\frac{\alpha}{4} \dist(y, z)} \| D_y^+ f\|_\infty  \|D_z^+ g\|_\infty \,dy\,dz 
			&\leq e^{- \frac{\alpha}{12}  \dist(\region_f, \region_g)}  \left(\int_{\R^d}  \norm{D_y^+ f}_{\infty} \diff y\right)  \left(\int_{\R^d} \norm{D_z^+ g}_{\infty} \diff z\right) \notag \\
			&= e^{- \frac{\alpha}{12}  \dist(\region_f, \region_g)}  \tnorm{f} \cdot   \tnorm{g} . \label{eq:bound_close}
		\end{align}
		For $z \notin B$, using that $g$ is $(\region_g, \kappa, K)$-local yields
		\begin{align}
			\int_{B^c} \int_{\R^d} e^{-\frac{\alpha}{4}  \dist(y,z)} \| D_y^+ f\|_\infty \|D_z^+ g\|_\infty \,dy\,dz  \notag
			&\leq K e^{- \frac{\kappa}{3} \dist(\region_f, \region_g)} \int_{\R^d} \norm{D_y^+ f}_{\infty} \int_{\R^d} e^{-\frac{\alpha}{4}  \dist(y,z)} \diff z \diff y \\
			&\le C_{\alpha, d} K e^{- \frac{\kappa}{3} \dist(\region_f, \region_g)} \tnorm{f} . \label{eq:bound_far_z}
		\end{align}
		Analogously 
		\begin{align}
			\int_{\R^d} \int_{A^c} e^{-\frac{\alpha}{4}  \dist(y,z)} \| D_y^+ f\|_\infty  \|D_z^+ g\|_\infty \,dy\,dz
			\le C_{\alpha, d} K e^{- \frac{\kappa}{3} \dist(\region_f, \region_g)} \tnorm{g} \label{eq:bound_far_y} .
		\end{align}
		Summing \eqref{eq:bound_close}, \eqref{eq:bound_far_z} and \eqref{eq:bound_far_y} yields
		\[
		\iint_{\region^2} e^{-\frac{\alpha}{4} \dist(y,z)} \| D_y^+ f\|_\infty  \|D_z^+ g\|_\infty \,dy\,dz \leq C_{\alpha, d}  (\tnorm{f}  \cdot \tnorm{g} + K  (\tnorm{f} + \tnorm{g}))  e^{- m_1  \dist(\region_f, \region_g)}
		\]
		for $m_1 = \min(\alpha/12, \kappa/3)$.
		Consequently,
		\[
		\int_0^t \int_\Lambda F_s(x) G_s(x) \,dx \,ds \le C_{\alpha, d}^2 e^{2(1+b)t}  (\tnorm{f} \cdot  \tnorm{g} + K  (\tnorm{f} + \tnorm{g}))  e^{- m_1  \dist(\region_f, \region_g)}
		\]
		and substituting this bound back into \eqref{eq:pure_magic} proves the lemma for $C_1 = 2 \activity C_{\alpha, d}^2$ and $M=2(b+1)$.
	\end{proof}
	
	\subsection{Correlation functions and strong spatial mixing} \label{sec:cor-SSM}

	In this subsection we will prove Lemma \ref{lem:upgrade-to-small-vol}, and also show that SSM holds.  In particular, we will show: 
	
	\begin{proposition}\label{prop:SSM}
		Let $\lambda \geq 0$ and $\phi$ be a repulsive potential that decays exponentially with constants $B$ and $\alpha$ (see Definition \ref{def:potential_decay}).  Suppose $\phi$ satisfies single-site SSM for some constants $C_0$ and $m_0$.  Then $\phi$ satisfies SSM up to $\lambda$ for constants $C = C(\lambda,B,\alpha,C_0,m_0,d) > 0$ and $m = m(m_0,\alpha) > 0$.
	\end{proposition}

	We will first prove Lemma \ref{lem:upgrade-to-small-vol} and then prove Proposition \ref{prop:SSM} in a similar manner.  Both proofs will follow the strategy in \cite{michelen2022strong} and make use of correlation functions.  Each proof will begin by showing that one-point correlation functions are close (see Lemma \ref{lem:one-point-SSM}).    
	We begin by truncating the sum over $\bX$ using exponential decay of $\phi$.

	\begin{lemma}\label{lem:truncate-energy-distance}
		Let $\blambda \leq \lambda$ have compact support $\Lambda$.  Let $\phi$ be a repulsive potential that decays exponentially with constants $B$ and $\alpha$ (see Definition \ref{def:potential_decay}).  Then there is a constant $C := C(d,B,\alpha)$ so that for all $T \geq 1$ and $v \in \R^d$ we have \begin{equation*}
			\left|\E_{\bX \sim \Lambda,\blambda} \exp\left(- \sum_{x \in \bX} \phi(v,x)\right) - \E_{\bX \sim \Lambda,\blambda} \exp\left(- \sum_{x \in \bX \cap B_T(v)} \phi(v,x)\right)\right| \leq C \lambda e^{-\alpha T/2}\,.
		\end{equation*}
	\end{lemma}
	\begin{proof}
		Note that by Lemma \ref{lem:complexproductdifference} and Poisson domination (\ref{lem:Poisson-domination}) we have \begin{align*}
			\left|\E_{\bX \sim \Lambda,\blambda} \exp\left(- \sum_{x \in \bX} \phi(v,x)\right) - \E_{\bX \sim \Lambda,\blambda} \exp\left(- \sum_{x \in \bX \cap B_T(v)} \phi(v,x)\right)\right| &\leq \E_{\bX \sim \Lambda,\blambda} \left[\sum_{x \in \bX \cap B_T(v)^c}(1 - e^{-\phi(v,x)})\right] \\
			&\leq \lambda \int_{x \in B_T(v)^c} (1 - e^{-\phi(v,x)})\,\diff x\\
			&\leq \lambda B \int_{x \in B_T(v)^c} e^{-\alpha\cdot\dist(v,x)}\,\diff x \,.
		\end{align*} 
		Bounding the integral by $C \lambda e^{-\alpha T/2}$ for some $C$ depending on $B,\alpha,d$ completes the proof.
	\end{proof}
	
	We now show that one can prove decay of one-point correlation functions:
	
	\begin{lemma}\label{lem:1-pt-SSSSM}
		Let $\blambda \leq \lambda$ have compact support contained in $\Lambda$.   Let $\phi$ be a repulsive potential that decays exponentially with constants $B$ and $\alpha$ (see Definition \ref{def:potential_decay}) and that satisfies single-site SSM (see Definition \ref{def:spatial_mixing}) with constants $C_0,m_0 > 0$ for all bounded measurable $\Lambda'$ with $\Vol(\Lambda') \geq 1$.  Then there are constants $C := C(C_0,m_0,d,B,\alpha,\lambda)$ and $m = m(m_0,\alpha,d)$ so that for all $v$ and $x$ we have \begin{equation*}
			\rho_{\Lambda,\blambda_x}(v)\left( 1 - C e^{-m \cdot\dist(v,x)}\right) \leq \rho_{\Lambda,\blambda}(v) e^{-\phi(v,x)} \leq \rho_{\Lambda,\blambda_x}(v)\left( 1 + C e^{-m \cdot\dist(v,x)}\right)\,.
		\end{equation*}
	\end{lemma}
	\begin{proof}
		Recall that by \eqref{eq:1-point-identity} we have $$\rho_{\Lambda,\blambda}(v) = \blambda(v) \E_{\bX \sim \Lambda, \blambda} \exp\left(-\sum_{x \in \bX} \phi(v,x) \right)$$
		and note $\blambda_x(v) = \blambda(v) e^{-\phi(v,x)}.$
		By \eqref{eq:lower-bound-exponential-added-energy} we have $\rho_{\Lambda,\blambda}(v) \geq c\blambda(v)$ for some $c = c(d,B,\alpha,\lambda) > 0$ so it is sufficient to bound \begin{equation*}
			\left|\E_{\bX \sim \Lambda, \blambda} \exp\left(-\sum_{x \in \bX} \phi(v,x) \right) - \E_{\bX \sim \Lambda, \blambda_x} \exp\left(-\sum_{x \in \bX} \phi(v,x) \right) \right|\,.
		\end{equation*}
		Set $T = \dist(v,x)$ and note that by adjusting $C$ we may assume that $T$ is large enough so that $\Vol(B_{T/2}(v)) \geq 1$.  Letting $C_1$ be the constant from Lemma \ref{lem:truncate-energy-distance} we may bound \begin{align*}
			\Bigg|\E_{\bX \sim \Lambda, \blambda}& \exp\left(-\sum_{x \in \bX} \phi(v,x) \right) - \E_{\bX \sim \Lambda, \blambda_x} \exp\left(-\sum_{x \in \bX} \phi(v,x) \right) \Bigg| \\
			&\leq 2 C_1 e^{-\alpha T/4} +         \Bigg|\E_{\bX \sim \Lambda, \blambda} \exp\left(-\sum_{x \in \bX \cap B_{T/2}(v)} \phi(v,x) \right) - \E_{\bX \sim \Lambda,\blambda_x} \exp\left(-\sum_{x \in \bX \cap B_{T/2}(v)} \phi(v,x) \right) \Bigg| \\
			&\leq 2 C_1 e^{-\alpha T/4} + \|\mu_{\blambda} - \mu_{\blambda_x}\|_{B_{T/2}(v)}\\
			&\leq 2 C_1 e^{-\alpha T/4} + C_0 e^{-m T/2}
		\end{align*} 
		where the last inequality is by single-site SSM.  
	\end{proof}
	
	It will now be easy to upgrade to decay of $k$-point correlation functions.  
	\begin{corollary}\label{cor:multipoint-SSSSM}
		Let $\blambda \leq \lambda$ have compact support contained in $\Lambda$.   Let $\phi$ be a repulsive potential that decays exponentially with constants $B$ and $\alpha$ (see Definition \ref{def:potential_decay}) and that satisfies single-site SSM (see Definition \ref{def:spatial_mixing}) with constants $C_0,m_0 > 0$ for all bounded measurable $\Lambda'$ with $\Vol(\Lambda') \geq 1$.  Then there are constants $C := C(C_0,m,d,B,\alpha,\lambda)$ and $m = m(m_0,\alpha,d)$ so that for all $\mathbf{v} = (v_1,\ldots,v_k)$ and $x$ we have 
		\begin{equation*}
			\rho_{\Lambda,\blambda_x}(\mathbf{v}) \prod_{j = 1}^k \left( 1 - C e^{-m \cdot\dist(v_j,x)}\right) \leq  \rho_{
				\Lambda,\blambda}(\mathbf{v}) e^{-\sum_{j =1}^k \phi(v_j,x)} \leq \rho_{\Lambda,\blambda_x}(\mathbf{v}) \prod_{j = 1}^k \left( 1 + C e^{-m \cdot\dist(v_j,x)}\right)\,.
		\end{equation*}
	\end{corollary}
	\begin{proof}
		This follows immediately from combining \eqref{eq:multipoint-correlation-identity} with Lemma \ref{lem:1-pt-SSSSM}.
	\end{proof}

	We obtain strong spatial mixing by writing the projected measure of an event as a sum of non-negative terms involving correlations.  For an event $A \in \events_{\Lambda'}$ we  define 
	$$A_{\Lambda'} = \{\bX : \bX \cap \Lambda' \in A\}\,.$$
	
	\begin{lemma}[Lemma 12, \cite{michelen2022strong}] \label{lem:ssm-via-densities}
		Let $\Lambda' \subset \Lambda$ be compact sets, $x_0 \in \Lambda'$ and $\blambda$ be a bounded activity function.  For a tempered repulsive potential $\phi$ we have $$\mu_{\Lambda,\blambda}(A_{\Lambda'}) = \int_{k \geq 0} \frac{1}{k!} \int_{(\Lambda')^k} \one_{ \{x_1,\ldots,x_k\} \in A} \rho_{\Lambda,\blambda}(x_1,\ldots,x_k) \exp\left(-\int_{\Lambda'} \rho_{\Lambda,\widehat{\blambda}_{x,x_1,\ldots,x_k}}(x) \,\diff x \right) \diff x_1 \, \ldots \diff x_k$$
		where \begin{equation*}
			\widehat{\blambda}_{x,x_1,\ldots,x_k}(y) = \begin{cases}
				0 & \text{ if } |y - x_0| < |x - x_0| \text{ and }y \in \Lambda' \\
				\blambda(y) \prod_{i = 1}^k e^{-\phi(y,x_i)} & \text{ otherwise }
			\end{cases} \,.
		\end{equation*}
	\end{lemma}
	
	We now complete the proof of Lemma \ref{lem:upgrade-to-small-vol}.

	\begin{proof}[Proof of Lemma \ref{lem:upgrade-to-small-vol}]
		We wish to show that $$\|\mu_{\Lambda,\blambda} - \mu_{\Lambda,\blambda_w}\|_{\Lambda'} \leq C \Vol(\Lambda') e^{-m \cdot\dist(\Lambda',w))}\,.$$
		By Corollary \ref{cor:multipoint-SSSSM},  we have constants $C_0, \alpha_0$ so that for all $v_1,\ldots,v_k$ we have $$ \rho_{\Lambda,\blambda}(v_1,\ldots,v_k) \leq \rho_{\Lambda,\blambda_x}(v_1,\ldots,v_k)\left( 1 + C \Vol(\Lambda') e^{-m \cdot\dist(\Lambda',w)}\right)^k\,.$$
		Set $\blambda' = \blambda_w$.   Define $\eps = C_0 e^{-\alpha_0 \cdot\dist(\Lambda',w)}$ and assume $\delta := (\lambda+1) \vol(\Lambda') \eps \leq 1$ as otherwise there is nothing to show.   Defining
		$$h_{\blambda}(\mathbf{x}) = \exp\left(-\int_{\Lambda'} \rho_{\Lambda,\widehat{\blambda}_{x,x_1,\ldots,x_k}}(x) \,\diff x \right)\,, $$
		if we assume that $\delta \leq 1$ then we may use Lemma \ref{lem:1-pt-SSSSM} to bound \begin{equation}
			\exp\left(-\int_{\Lambda'} \rho_{\Lambda,\widehat{\blambda}_{x,x_1,\ldots,x_k}}(x) \,\diff x \right) \leq \exp\left(-\int_{\Lambda'} \rho_{\Lambda,\widehat{\blambda}'_{x,x_1,\ldots,x_k}}(x) \,\diff x \right)(1 + 2 \delta)
		\end{equation}
		(see \cite[Corollary 16]{michelen2022strong} for more details).  
		As such, we see \begin{align*}
			\|\mu_{\Lambda,\blambda} - \mu_{\Lambda,\blambda'}\|_{\Lambda'} &\leq \sum_{k \geq 1} \frac{1}{k!} \int_{(\Lambda')^k} |\rho_{\Lambda,\blambda}(\mathbf{x}) h_{\blambda}(\mathbf{x}) - \rho_{\Lambda,\blambda'}(\mathbf{x}) h_{\blambda'}(\mathbf{x})| \diff \mathbf{x}  \\
			&\leq \sum_{k \geq 1} \frac{1}{k!} \int_{(\Lambda')^k} \rho_{\Lambda,\blambda}(x) h_{\blambda}(x) \left[(1 + \eps)^k(1 + 2\delta) - 1 \right] \diff \mathbf{x} \,.
		\end{align*}
		Lemma \ref{lem:ssm-via-densities} shows $$\sum_{k \geq 0} \frac{1}{k!} \int_{(\Lambda')^k} \rho_{\Lambda,\blambda}(x) h_{\blambda}(x) \, \diff \mathbf{x} = 1$$ 
		and $$\sum_{k \geq 0} \frac{1}{k!} \int_{(\Lambda')^k} \rho_{\Lambda,\blambda}(x) h_{\blambda}(x) (1 + \eps)^k \, \diff \mathbf{x} = \E_{\bX \sim \Lambda,\blambda} (1 + \eps)^{|\bX \cap \Lambda'|} \leq \exp\left( \eps \lambda \Vol(\Lambda') \right) \leq e^{\delta}$$
		where the inequality is by Poisson domination (Lemma \ref{lem:Poisson-domination} and Lemma \ref{lemma:poisson_moments}). 
		Combining the last three displayed equations shows 
		$$\|\mu_{\Lambda,\blambda} - \mu_{\Lambda,\blambda'}\|_{\Lambda'} \leq (1 + 2\delta) e^{\delta} - 1 \leq 10 \delta$$
		where in the last bound we used $\delta \leq 1$. 
	\end{proof}

	The proof of Proposition \ref{prop:SSM} follows the same trajectory, where the main work is to show the analogue of Lemma \ref{lem:1-pt-SSSSM}.  We first show that one can truncate correlation functions and capture the total variation distance on a ball.
	
	\begin{lemma}\label{lem:SSSSM-to-SSM}
		Let $\blambda \leq \lambda$ have compact support $\Lambda$.   Let $\phi$ be a repulsive potential that satisfies single-site SSM (see Definition \ref{def:spatial_mixing}) with constants $C_0 >0 , m > 0$.  Then there is a constant $C := C(d,C_0,m,\lambda)$ so that for all $T \geq 1$ and $v \in \R^d$ if we let $\blambda' = \blambda|_{B_{2T}(v)}$ then $$\| \mu_{\Lambda,\blambda} - \mu_{\Lambda,\blambda'}\|_{B_T(v)} \leq C e^{-m T/2}\,.$$
	\end{lemma}
	\begin{proof}
		First note that by the DLR equations \eqref{eq:finite_DLR} we have 
		\begin{equation}\label{eq:DLR-for-SSM}
			\E_{\bY \sim \Lambda,\blambda} [\mu_{\Lambda,\blambda'_{\bY \setminus B_{2T}(v)}}] = \mu_{\Lambda,\blambda}\,.
		\end{equation}
		Let $Y \in \pointsets$ and enumerate  $y_1, y_2, \dots$ be the points in $Y \setminus B_{2T}(v)$.  Set $Z_j = \{y_1, \dots, y_j\}$, and $Z_0 = \emptyset$.
		Note that $Z_j$ and $Z_{j-1}$ differ by only a single point.
		By the triangle inequality and single-site SSM, we then have
		\begin{align*}
			\left\|{\mu_{\region, \activities'_{Y \setminus B_{2T}(v)}} - \mu_{\region, \activities'}}\right\|_{B_T(v)} 
			\le \sum_{j = 1}^{\infty}  \left\|{\mu_{\region, \activities_{Z_{j}}'} - \mu_{\region, \activities_{Z_{j-1}}'}}\right\|_{B_T(v)}  \le C_0 \vol(B_T(v))  \sum_{y \in Y \setminus B_{2T}(v)} e^{- m  \cdot\dist(y, B_T(v))}.
		\end{align*}
		
		Combining the above displayed equation with \eqref{eq:DLR-for-SSM} we see
		\begin{align*}
			\| \mu_{\Lambda,\blambda} - \mu_{\Lambda,\blambda'}\|_{B_T(v)} &\leq C_0 \Vol(B_T(v)) \E_{\bY \sim \Lambda,\blambda}\left[  \sum_{y \in \bY \setminus B_{2T}(v)} e^{- m \cdot\dist(y, B_T(v))} \right] \\
			&\leq C_0 \lambda \Vol(B_T(v)) \int_{y \notin B_{2T}(v)} e^{-m \cdot\dist(y,B_T(v))} \diff y \\
			&\leq C e^{-m T/2}
		\end{align*}
		where the second inequality is due to Poisson domination (Lemma \ref{lem:Poisson-domination}) and the third is by taking $C$ large enough depending on $m,C_0,\lambda$ and $d$.
	\end{proof}
	
	We are now ready to show that we have exponential decay of influence on $1$-point correlation functions:
	
	\begin{lemma}\label{lem:one-point-SSM}
		Let $\blambda,\blambda' \leq \lambda$ have compact support contained in $\Lambda$.   Let $\phi$ be a repulsive potential that decays exponentially with constants $B$ and $\alpha$ (see Definition \ref{def:potential_decay}) and satisfies single-site SSM (see Definition \ref{def:spatial_mixing}) with constants $C_0,m_0 > 0$.  Then there are constants $C := C(C_0,m,d,B,\alpha,\lambda)$ and $m = m(m_0,\alpha,d)$ so that for all $v$ with $\blambda(v) = \blambda'(v)$ we have \begin{equation*}
			\rho_{\Lambda,\blambda}(v) \leq \rho_{\Lambda,\blambda'}(v)\left( 1 + C e^{-m \cdot\dist(v,\mathrm{supp}(\blambda \neq \blambda'))}\right)\,.
		\end{equation*}
	\end{lemma}
	\begin{proof}
		Recall that by \eqref{eq:1-point-identity} we have $$\rho_{\Lambda,\blambda}(v) = \blambda(v) \E_{\bX \sim \Lambda, \blambda} \exp\left(-\sum_{x \in \bX} \phi(v,x) \right)\,.$$
		By \eqref{eq:lower-bound-exponential-added-energy} we have $\rho_{\Lambda,\blambda}(v) \geq c\blambda(v)$ for some $c = c(d,B,\alpha,\lambda) > 0$ so it is sufficient to bound \begin{equation*}
			\left|\E_{\bX \sim \Lambda, \blambda} \exp\left(-\sum_{x \in \bX} \phi(v,x) \right) - \E_{\bX \sim \Lambda, \blambda'} \exp\left(-\sum_{x \in \bX} \phi(v,x) \right) \right|\,.
		\end{equation*}
		Set $T = \dist(v,\mathrm{supp}(\blambda \neq \blambda'))$ and note that by adjusting $C$ we may assume that $T \geq 2$.  Letting $C_1$ be the constant from Lemma \ref{lem:truncate-energy-distance} we may bound \begin{align*}
			\Bigg|\E_{\bX \sim \Lambda, \blambda}& \exp\left(-\sum_{x \in \bX} \phi(v,x) \right) - \E_{\bX \sim \Lambda, \blambda'} \exp\left(-\sum_{x \in \bX} \phi(v,x) \right) \Bigg| \\
			&\leq 2 C_1 e^{-\alpha T/4} +         \Bigg|\E_{\bX \sim \Lambda, \blambda} \exp\left(-\sum_{x \in \bX \cap B_{T/2}(v)} \phi(v,x) \right) - \E_{\bX \sim \Lambda,\blambda'} \exp\left(-\sum_{x \in \bX \cap B_{T/2}(v)} \phi(v,x) \right) \Bigg| \\
			&\leq 2 C_1 e^{-\alpha T/4} + \|\mu_{\blambda} - \mu_{\blambda'}\|_{B_{T/2}(v)}\,.
		\end{align*} 
		Letting $C_2$ denote the constant from Lemma \ref{lem:SSSSM-to-SSM} we see that $$\|\mu_{\blambda} - \mu_{\blambda'}\|_{B_{T/2}(v)} \leq 2 C_2 e^{-m_0 T/2}$$
		thus completing the proof.
	\end{proof}
	
	We now simply indicate how one completes the proof in the same manner as the proof of Lemma \ref{lem:upgrade-to-small-vol}:
	
	\begin{proof}[Proof of Proposition~\ref{prop:SSM}]
		In the same manner as Corollary \ref{cor:multipoint-SSSSM} we see that for any tuple $\mathbf{v}$ one may iterate Lemma \ref{lem:one-point-SSM} to show $$\rho_{\Lambda,\blambda}(\mathbf{v}) \leq \rho_{\Lambda,\blambda'}(\mathbf{v})\left( 1 + C e^{-m \cdot\dist(v,\mathrm{supp}(\blambda \neq \blambda'))}\right)^k\,.$$
		Using Lemma \ref{lem:ssm-via-densities} as in the proof of Lemma \ref{lem:upgrade-to-small-vol} completes the proof.
	\end{proof}

	\begin{proof}[Proof of Theorem~\ref{thm:spatial-mixing}]
		The Poincar\'e inequality for all $\blambda \leq \lambda < \lambda_{\spec}$ follows from Theorem \ref{thm:Kondratiev}.  Single-site SSM then follows from Theorem \ref{thm:ssssm} while SSM follows from Proposition \ref{prop:SSM}. 
	\end{proof}

	\section{Infinite-volume Gibbs measures and thermodynamic limits} \label{sec:uniqueness}
	
	We prove thermodynamic properties along with uniqueness of infinite-volume Gibbs measure. 
	
	\subsection{Uniqueness}
	
	In this section, we show that single-site SSM implies that there is a unique infinite-volume Gibbs measure on $\R^d$.  
	The main result of this section is the following uniqueness theorem.
	\begin{theorem} \label{thm:ssssm_uniqueness}
		Let $\activity \in \R_{\ge 0}$ and let $\potential$ be a repulsive potential that satisfies single-site SSM (see Definition \ref{def:spatial_mixing}) up to $\activity$.
		Then, for every activity function $\activities \le \activity$, there is a unique infinite-volume Gibbs measure on $\R^d$ that is compatible with $\potential$ and $\activities$.
	\end{theorem}
	
	We note that by a standard compactness argument there is always at least one Gibbs measure (see, e.g., \cite[Theorem~B.1]{jansen2019cluster}) and so the content of Theorem \ref{thm:ssssm_uniqueness} is that there is precisely one Gibbs measure.

	To prove Theorem \ref{thm:ssssm_uniqueness}, we show that single-site SSM implies that the following uniqueness criterion applies.
	
	\begin{lemma}\label{lemma:uniqueness_expectation}
		Fix an activity function $\activities$ and a repulsive pair potential $\potential$, and write $\region_n = [-n, n]^d$.
		Suppose for all $k \in \N$ and all infinite-volume Gibbs measures $\mu$ compatible with $\potential$ and $\activities$ it holds that
		\[
		\lim_{n \to \infty} \E_{\bX \sim \mu} \left[\left\|\mu_{\region_{n}, \activities_{\bX \setminus \region_n}} - \mu_{\region_{n}, \activities} \right\|_{\region_k}\right] = 0,
		\]
		then there is at most one infinite-volume Gibbs measure compatible with $\potential$ and $\activities$.
	\end{lemma}
	
	\begin{proof}
		We will use the following standard fact that distributions of point processes are characterized by their finite-volume projections: 
		If, for two probability measures $P, Q$ on $(\pointsets, \events)$ it holds, for every bounded measurable $\region \subset \R^d$, the pushforwards of $P$ and $Q$ under $X \mapsto X \cap \region$ are identical, then $P$ and $Q$ must be identical (this follows from, e.g. \cite[Lemma~6.1.III]{daley2003introduction}).

		Suppose $\mu_1, \mu_2$ are two infinite-volume Gibbs measures. 
		Our goal is now to show that, for all bounded measurable $\region \subset \R^d$ and all events $A \in \events$, it holds that $\E_{\bX \sim \mu_1}[\ind{\bX \cap \region \in A}] = \E_{\bX \sim \mu_2}[\ind{\bX \cap \region \in A}]$. 
		It follows that $\mu_1 = \mu_2$, and consequently, there is at most one Gibbs measure.
		
		To this end, choose $k$ large enough such that $\region \subseteq \region_k$.
		Since $\mu_1, \mu_2$ satisfy the DLR equations, it holds that, for all $n > k$,
		\begin{align*}
			&\absolute{\E_{\bX \sim \mu_1}[\ind{\bX \cap \region \in A}] - \E_{\bX \sim \mu_2}[\ind{\bX \cap \region \in A}]} \\
			&\hspace{2em}= \absolute{\E_{\bX \sim \mu_1}\big[\E_{\bY \sim \mu_{\region_n, \activities_{\bX \setminus \region_n}}}[\ind{\bY \cap \region \in A}]\big] - \E_{\bX \sim \mu_2}\big[\E_{\bY \sim \mu_{\region_n, \activities_{\bX \setminus \region_n}}}[\ind{\bY \cap \region \in A}]\big]} \\
			&\hspace{2em} \le \E_{\bX \sim \mu_1}\left[\absolute{\E_{\bY \sim \mu_{\region_n, \activities_{\bX \setminus \region_n}}}[\ind{\bY \cap \region \in A}] - \E_{\bY \sim \mu_{\region_n, \activities}}[\ind{\bY \cap \region \in A}]}\right] \\  
			&\hspace{4em}+ \E_{\bX \sim \mu_2}\left[\absolute{\E_{\bY \sim \mu_{\region_n, \activities_{\bX \setminus \region_n}}}[\ind{\bY \cap \region \in A}] - \E_{\bY \sim \mu_{\region_n, \activities}}[\ind{\bY \cap \region \in A}]}\right]
		\end{align*}
		Since $\region \subseteq \region_k$, we further have 
		\[
		\absolute{\E_{\bY \sim \mu_{\region_n, \activities_{X \setminus \region_n}}}[\ind{\bY \cap \region \in A}] - \E_{\bY \sim \mu_{\region_n, \activities}}[\ind{\bY \cap \region \in A}]} \le \left\|{\mu_{\region_n, \activities_{X \setminus \region_n}} - \mu_{\region_n, \activities}}\right\|_{\region_k}
		\]
		for all $X \in \pointsets$.
		Thus, we have 
		\begin{align*}
			&\absolute{\E_{\bX \sim \mu_1}[\ind{\bX \cap \region \in A}] - \E_{\bX \sim \mu_2}[\ind{\bX \cap \region \in A}]} \\
			&\hspace{2em}\le \E_{\bX \sim \mu_1}\left[\left\|{\mu_{\region_n, \activities_{\bX \setminus \region_n}} - \mu_{\region_n, \activities}}\right\|_{\region_k}\right] + \E_{\bX \sim \mu_2}\left[\left\|{\mu_{\region_n, \activities_{\bX \setminus \region_n}} - \mu_{\region_n, \activities}}\right\|_{\region_k}\right]
		\end{align*}
		and taking $n \to \infty$ concludes the proof.
	\end{proof}
	
	We now use Lemma \ref{lemma:uniqueness_expectation} to show Theorem \ref{thm:ssssm_uniqueness}.
	\begin{proof}[Proof of Theorem \ref{thm:ssssm_uniqueness}]
		Let $k \in \N$, $n > k$ and $Y \in \pointsets$.
		Since $Y$ is a locally finite point set, it must be countable.
		For any enumeration $y_1, y_2, \dots$ of the points in $Y \setminus \region_n$, let $Z_j = \{y_1, \dots, y_j\}$, and set $Z_0 = \emptyset$.
		Note that $Z_j$ and $Z_{j-1}$ differ by only a single point.
		By the triangle inequality and single-site SSM, we then have
		\begin{align*}
			\left\|\mu_{\region_n, \activities_{Y \setminus \region_n}} - \mu_{\region_n, \activities}\right\|_{\region_k} 
			\le \sum_{j = 1}^{\infty}  \left\|{\mu_{\region_n, \activities_{Z_{j}}} - \mu_{\region_n, \activities_{Z_{j-1}}}}\right\|_{\region_k}  \le C  \vol(\region_k) \sum_{y \in Y \setminus \region_n} e^{- m \kappa\cdot\dist(y, \region_k)}.
		\end{align*}
		Note that by the DLR equations \eqref{eq:finite_DLR} and Poisson domination (Lemma \ref{lem:Poisson-domination}), for every compact set $K$ we have $$\E_{\bY \sim \mu} \sum_{y \in \bY \cap K \setminus \Lambda_n} e^{-m\cdot\dist(y,\Lambda_k)} \leq \lambda \int_{K \setminus \Lambda_n} e^{-m \cdot\dist(y,\Lambda_k)} \diff y\,.$$
		By monotone convergence we obtain    
		\begin{equation} \label{eq:infinite-volume-to-finite}
			\E_{\bY \sim \mu}\left[\left\|{\mu_{\region_n, \activities_{\bY \setminus \region_n}} - \mu_{\region_n, \activities}}\right\|_{\region_k}\right] 
			\le C \vol(\region_k)   \lambda \int_{\R^d \setminus \region_n} e^{- m \kappa\cdot\dist(y, \region_k)} \diff y.
		\end{equation}
		For fixed $k$, the right-hand side goes to $0$ as $n \to \infty$, which concludes the proof.  
	\end{proof}

	\subsection{Decay of correlations}
	
	We will show that the unique infinite-volume Gibbs measure enjoys exponential decay of correlations.  For this, we recall that for an infinite-volume Gibbs measure $\mu$ compatible with $\phi$ and $\blambda$ the $k$-point correlation function satisfies \eqref{eq:k-pt-correlation-definition}; when the infinite-volume Gibbs measure is unique, we will simply write $\rho_{\blambda} = \rho_{\mu,\blambda}$.  Under the assumption of uniqueness, we will see that one has convergence of finite volume correlation functions to their infinite-volume analogue.
	\begin{fact}\label{fact:convergence-corr-fcns}
		Suppose that $\phi$ is a repulsive potential and $\lambda \in \R_{\geq 0}$ so that for all $\blambda \leq \lambda$ there is a unique infinite-volume Gibbs measure compatible with $\phi$ and $\blambda$.  Then for $\Lambda_n = [-n,n]^d$ and all $v_1,\ldots,v_k$ we have \begin{equation*}
			\lim_{n \to \infty} \rho_{\Lambda_n,\blambda}(v_1,\ldots,v_k) = \rho_{\blambda}(v_1,\ldots,v_k)\,.
		\end{equation*}
	\end{fact}
	\begin{proof}
		Let $\mu_n$ denote the finite-volume Gibbs measures associated to $(\Lambda_n,\blambda)$ and $\mu$ the infinite-volume Gibbs measure.  The sequence $\mu_n$ is tight and so has a convergence subsequence; each convergent subsequence must converge to a Gibbs measure (see e.g.\ \cite[Theorem~5.7]{jansen2018gibbsian}) and so $\mu_n$ converges to $\mu$.  Equation \eqref{eq:multipoint-correlation-identity} implies $$ \rho_{\Lambda_n,\blambda}(v_1,\ldots,v_k) = \prod_{j = 1}^k \blambda(v_j)  e^{- \sum_{1 \leq i < j \leq k} \phi(v_i,v_j)} \E_{\bY \sim \Lambda_n,\blambda} \left[\exp\left(- \sum_{j = 1}^k \sum_{y \in \bY} \phi(v_j,y) \right) \right]$$ and so taking $n \to \infty$ and using \eqref{eq:k-pt-correlation-definition} completes the proof.
	\end{proof}
	
	Approximate factorization of correlation functions follows quickly.
	
	\begin{proposition}\label{prop:exponential-decay-k-point}
		Let $\activity \in \R_{\ge 0}$ and let $\potential$ be a repulsive potential that decays exponentially with constants $B$ and $\alpha$ (see Definition \ref{def:potential_decay}) and satisfies single-site SSM (see Definition \ref{def:spatial_mixing}) with constants $C_0,m_0 > 0$ up to $\lambda$.  
		Then for each $k,\ell$ there are constants $C = (C_0,m_0,B,\alpha,d,\lambda) > 0$ and $m = (m_0,\alpha,d)$ so that for all $\mathbf{v} = (v_1,\ldots,v_k)$ and $\mathbf{w} = (w_1,\ldots,w_\ell)$ we have \begin{equation*}
			\left|\rho_{\blambda}(\mathbf{v},\mathbf{w}) - \rho_{\blambda}(\mathbf{v})\rho_{\blambda}(\mathbf{w}) \right| \leq C^{k + \ell} e^{-m \cdot\dist(\mathbf{v},\mathbf{w})}
		\end{equation*}
		where we write $\dist(\mathbf{v},\mathbf{w}) = \min_{i,j} \dist(v_i,w_j)$ and $\rho_{\blambda}$ refers to the correlation functions of the unique infinite-volume Gibbs measure compatible with $\phi$ and $\blambda$.
	\end{proposition}
	\begin{proof}
		By Fact \ref{fact:convergence-corr-fcns} it is sufficient to prove \begin{equation}\label{eq:corr-fcns-need}
			\left|\rho_{\Lambda,\blambda}(\mathbf{v},\mathbf{w}) - \rho_{\Lambda,\blambda}(\mathbf{v})\rho_{\Lambda,\blambda}(\mathbf{w}) \right| \leq C^{k + \ell} e^{-m \cdot\dist(\mathbf{v},\mathbf{w})}
		\end{equation}
		for all bounded and measurable $\Lambda$.   We note that $$\rho_{\Lambda,\blambda}(\mathbf{v},\mathbf{w}) = \rho_{\Lambda,\blambda}(\mathbf{v}) e^{-\sum_{i,j} \phi(v_i,w_j)} \rho_{\Lambda,\blambda_{\mathbf{v}}}(\mathbf{w})\,.$$
		
		Letting  $\blambda_i: = \blambda_{v_1,\ldots,v_i}$ we see from the previous display equation along with Lemma \ref{lem:complexproductdifference} that \begin{equation}\label{eq:k-point-decay-break-up}
			\left|\rho_{\Lambda,\blambda}(\mathbf{v},\mathbf{w}) - \rho_{\Lambda,\blambda}(\mathbf{v})\rho_{\Lambda,\blambda}(\mathbf{w}) \right| \leq \left(\lambda+1\right)^{k +\ell}\left( 1 - e^{-\sum_{i,j} \phi(v_i,w_j)} + \sum_{i = 1}^k\left| \rho_{\Lambda,\blambda_i}(\mathbf{w}) - \rho_{\Lambda,\blambda_{i-1}}(\mathbf{w}) \right|\right)\,. 
		\end{equation}
		
		Using exponential decay of the potential and Lemma \ref{lem:complexproductdifference} bounds \begin{equation}\label{eq:k-point-potential-decay}
			1 - e^{-\sum_{i,j} \phi(v_i,w_j)} \leq k\ell B e^{- \alpha\cdot\dist(\mathbf{v},\mathbf{w})} \,.
		\end{equation}
		
		For the latter term in \eqref{eq:k-point-decay-break-up}, if we let $C_1$ and $m$ be from Corollary \ref{cor:multipoint-SSSSM} then we have \begin{align}\label{eq:one-point-difference}
			\left|\rho_{\Lambda,\blambda_i}(\mathbf{w}) - \rho_{\Lambda,\blambda_{i-1}}(\mathbf{w}) \right| \leq (\lambda+1)^\ell\left[ \left(1 - e^{-\sum_{a,b} \phi(v_a,w_b) } \right) + \left|1 -  \prod_{j = 1}^\ell\left(1 + C_1 e^{-m \cdot\dist(w_j,v_i)} \right)  \right| \right] \,.
		\end{align}
		Recall that we have $\rho_{\Lambda,\blambda}(x_1,\ldots,x_r) \leq \lambda^r$ for all $\blambda,\Lambda$ and $x_1,\ldots,x_r$.  Thus, by taking $C$ large enough as a function of $\lambda$, we may assume $e^{-m \cdot\dist(\mathbf{v},\mathbf{w})} \leq \frac{1}{10(k + \ell) C_1}.$  This implies \begin{equation*}
			\left|1 -  \prod_{j = 1}^\ell\left(1 + C_1 e^{-m \cdot\dist(w_j,v_i)} \right)  \right| \leq 2 C_1 e^{-m \cdot\dist(\mathbf{v},\mathbf{w})} \,.
		\end{equation*}
		
		Combining with \eqref{eq:one-point-difference} and arguing as in \eqref{eq:k-point-potential-decay} shows \begin{equation} \label{eq:one-point-corr-difference}
			\left|\rho_{\Lambda,\blambda_i}(\mathbf{w}) - \rho_{\Lambda,\blambda_{i-1}}(\mathbf{w}) \right| \leq (\lambda +1 )^\ell \left[ k \ell B e^{-\alpha\cdot\dist(\mathbf{v},\mathbf{w})} + 2 \ell  C_1 e^{-m \cdot\dist(\mathbf{v},\mathbf{w})} \right]\,.
		\end{equation} 
		
		Combining \eqref{eq:k-point-decay-break-up} with \eqref{eq:one-point-corr-difference} shows a bound of \begin{align*}
			\left|\rho_{\Lambda,\blambda}(\mathbf{v},\mathbf{w}) - \rho_{\Lambda,\blambda}(\mathbf{v})\rho_{\Lambda,\blambda}(\mathbf{w}) \right| \leq \left(\lambda+1\right)^{2(k + \ell)}\left[k\ell B e^{-\alpha\cdot\dist(\mathbf{v},\mathbf{w})} + 2 k^2 \ell  B e^{-\alpha\cdot\dist(\mathbf{v},\mathbf{w})} + 2 k \ell C_1 e^{-m \cdot\dist(\mathbf{v},\mathbf{w})} \right] \,.
		\end{align*}
		Taking $C$ large enough as a function of $\lambda,B,C_1,\alpha,m$ completes the proof.
	\end{proof}

	\begin{proof}[Proof of Theorem \ref{thm:decay-of-k-point-correlations}]
		Single-site SSM follows from Theorem \ref{thm:spatial-mixing}.  The theorem then follows from Proposition~\ref{prop:exponential-decay-k-point}.
	\end{proof}

	\section{Zero-freeness from spatial mixing}\label{sec:zero-free}
	In this section, we show that single-site SSM implies zero-freeness and analyticity for repulsive interactions with exponential decay.
	Given a repulsive potential $\potential$ and a bounded measurable region $\region \subset \R^d$ with $\vol(\region) > 0$, recall that the finite-volume pressure at activity $\activity \in \R_{\ge 0}$ is given by
	\[
	p_{\region}(\activity) = \frac{1}{\vol(\region)} \cdot \log(Z_{\region}(\activity)) .
	\]
	We are interested in identifying activity regimes such that the limit $p = \lim_{\region \nearrow \R^d} p_{\region}$, called the infinite-volume pressure, is an analytic function. 
	For studying this question, it is beneficial to study the partition function $Z_{\region}$ not only for real activities, but to extend its domain to complex inputs.
	In particular, analyticity of $p$ on some interval in $\R_{\ge 0}$ is closely related to the question of whether $Z_{\region}$ is non-vanishing in some uniform neighborhood of the interval as $\region \nearrow \R^d$.
	Our main result in this section is the following relationship between single-site SSM and analyticity and zero-freeness.
	\begin{theorem}\label{thm:zero_freeness}
		Let $\activity \in \R_{\ge 0}$, and let $\potential$ be a repulsive potential that decays exponentially (see Definition \ref{def:potential_decay}) and that satisfies single-site SSM (see Definition \ref{def:spatial_mixing}) up to $\activity$.
		Then there exists some $\delta > 0$ such that the following holds:
		\begin{enumerate}[(1)]
			\item\label{thm:zero_freeness:roots} For every bounded measurable region $\region \subset \R^d$, we have $Z_{\region} \neq 0$ for all complex activities in the $\delta$-neighborhood of $[0, \activity]$.
			
			\item\label{thm:zero_freeness:log-bound} There is a constant $M$ so that for $\Lambda_n = [-n,n]^d$ there is an analytic branch of $\log Z_{\Lambda_n}(w)$ for $w$ in the $\delta$-neighborhood of $[0,\lambda]$ so that $\frac{1}{\Vol(\Lambda_n)}|\log Z_{\Lambda_n}(w)| \leq M$.

			\item\label{thm:zero_freeness:analyticity} Given a sequence of bounded measurable regions $(\region_n)_{n \in \N}$ such that $p = \lim_{n \to \infty} p_{\region_n}$ exists pointwise on $\R_{\ge 0}$, it holds that $p$ has an analytic continuation in the $\delta$-neighborhood of $[0, \activity]$.
		\end{enumerate}
	\end{theorem}

	Theorem \ref{thm:main} follows quickly now:
	\begin{proof}[Proof of Theorem~\ref{thm:main}]
		Single-site SSM holds for all $\lambda < \lambda_\spec$ by Theorem \ref{thm:spatial-mixing}.  Theorem~\ref{thm:ssssm_uniqueness} implies uniqueness of infinite-volume Gibbs measure, and Theorem \ref{thm:zero_freeness}\ref{thm:zero_freeness:analyticity} shows analyticity of the pressure.
	\end{proof}
	
	We will prove Theorem \ref{thm:zero_freeness} by a fairly involved inductive step involving ratios of partition functions. 
	For a repulsive potential $\phi$, activity function $\activities$, bounded measurable $\region \subset \R^d$, and $\eps\in \mathbb{C}$, we let
	
	\[
	F_{\region}^{\activities}(\varepsilon)
	\coloneqq
	\E_{\bX\sim\region,\activities}
	\big[(1+\varepsilon)^{\size{\bX}}\big]
	=
	\frac{Z_{\region}((1+\varepsilon)\activities)}{Z_{\region}(\activities)} .
	\]

	The main work of proving Theorem \ref{thm:zero_freeness} will be the following technical proposition.
	\begin{proposition}\label{prop:expectations}
		Let $\potential$ be a repulsive potential that decays exponentially (see Definition \ref{def:potential_decay}) with constants $B < \infty$ and $\alpha > 0$, and that satisfies single-site SSM (Definition \ref{def:spatial_mixing}) with constants $C$ and $m$ up to some activity $\activity \in \R_{\ge 0}$.
		For all $R > 0$ and $\kappa \in (0, 1/2)$, there are $\beta = \beta(\alpha) > 0$, $D = D(\activity, \alpha, B,d) < \infty$, $\delta = \delta(\activity, \alpha, B, d) > 0$ and $\epsbound = \epsbound(d, \activity, \kappa, R, \alpha, B, m, C) > 0$ such that, for all $\varepsilon \in \C$ with $\absolute{\varepsilon} \le \epsbound$, all activity functions $\activities \le \activity$ and all $\region \in \borel_{b}$ it holds that:
		\begin{enumerate}[(i)]
			\item\label{prop:expectations:zerofree}  $\absolute{F_{\region}^{\activities}(\varepsilon)} \ge (1-\kappa)^{c_{R/2}(\region)}$, where $c_{R/2}(\region)$ is the minimum number of balls of radius $R/2$ required to cover $\region$.
			In particular $F_{\region}^{\activities}(\varepsilon) \neq 0$. 
			
			\item\label{prop:expectations:region} If $\region' = \region \setminus B_{R}(x)$ for some $x \in \region$, then $\displaystyle \absolute{\frac{F_{\region}^{\activities}(\varepsilon)}{F_{\region'}^{\activities}(\varepsilon)} - 1} \le \kappa$.
			\item\label{prop:expectations:decay} For all $x \in \R^d$ with $\dist(x, \region) \ge \delta$, it holds that
			$\displaystyle \absolute{\frac{F_{\region}^{\activities_x}(\varepsilon)}{F_{\region}^{\activities}(\varepsilon)} - 1} \le D  e^{-\beta\cdot\dist(x, \region)}$.
			\item\label{prop:expectations:boundary} For all $x \in \R^d$, it holds that
			$\displaystyle \absolute{\frac{F_{\region}^{\activities_x}(\varepsilon)}{F_{\region}^{\activities}(\varepsilon)} - 1} \le \sqrt{\epsbound}$.
		\end{enumerate}
	\end{proposition}
	
	We prove Proposition \ref{prop:expectations} via an induction over the number of balls of radius $\ell \coloneqq R/2$ required to cover $\region$.
	Provided \ref{prop:expectations:zerofree}-\ref{prop:expectations:boundary} are true for all bounded measurable regions that can be covered by $<n$ balls, the induction step has four parts.
	\begin{description}
		\item[Part 1] We start by showing that \ref{prop:expectations:region} holds for any region that can be covered by $n$ balls, assuming that \ref{prop:expectations:zerofree} and \ref{prop:expectations:boundary} hold for all regions that can be covered by $< n$ balls.
		\item[Part 2] We then prove \ref{prop:expectations:zerofree} for all regions covered by $n$ balls given that \ref{prop:expectations:region} holds for all regions covered by $n$ balls, and \ref{prop:expectations:zerofree} holds for all regions covered by $< n$ balls.
		\item[Part 3] Next, we show that \ref{prop:expectations:decay} holds for all regions covered by $n$ balls given that \ref{prop:expectations:zerofree} and \ref{prop:expectations:region} hold for all regions covered by $n$ balls and that \ref{prop:expectations:boundary} holds for all regions covered by $< n$ balls.
		\item[Part 4] Finally, we prove \ref{prop:expectations:boundary} for all regions covered by $n$ balls, assuming \ref{prop:expectations:zerofree} and \ref{prop:expectations:decay} hold for all regions covered by $n$ balls, and \ref{prop:expectations:boundary} holds for regions covered by $<n$ balls.
	\end{description}
	See Figure \ref{img:proofstruct} for an overview of the structure of the induction step.  We note that Parts 1 and 2 do not make use of the assumption of single-site SSM, while both Parts 3 and 4 do.
	\begin{figure}[h!]
		\centering
		\includegraphics[scale=1]{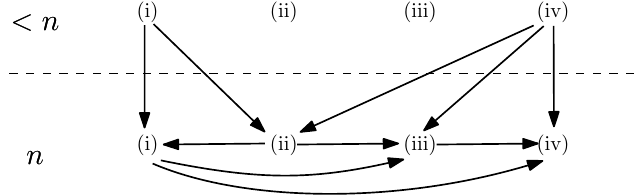}
		\caption{Structure of the induction step.} \label{img:proofstruct}
	\end{figure}
	
	Moreover, for Part 3 of the induction step, we use the following elementary lemma. 
	\begin{lemma}
		\label{lemma:ftc} Let $t > 0$, and let
		$f\colon [0, t] \to \C \setminus \{0\}$ be absolutely continuous.  It holds
		that
		\begin{equation*}
			\frac{f(t)}{f(0)} = \exp\left(\int_{0}^{t} \frac{f'(s)}{f(s)} \diff s \right),
		\end{equation*}
		where the integral should be understood as a Lebesgue integral and $f'$ is almost everywhere a derivative of $f$.
	\end{lemma}
	\begin{proof}
		Since $f$ is continuous and non-zero, we have that $1/f$ is bounded and continuous.  In particular,  we have that $f'/f \in L^1([0,t])$ and so $I(u) = \int_0^u\frac{f'(t)}{f(t)}\,\diff t$ is absolutely continuous and $I'(u) = \frac{f'(u)}{f(u)}$ almost-everywhere.  Compute $$\frac{\diff}{\diff u}\left(f(u) e^{-I(u)} \right) = f'(u) e^{-I(u)} - I'(u) f(u) e^{-I(u)} = 0$$
		almost everywhere.  Hence $f(u) e^{-I(u)}$ is constant almost-everywhere on $[0,t]$.  Since $f(u) e^{-I(u)}$ is absolutely continuous, it is in fact constant, implying $f(t) e^{-I(t)} = f(0)$ and completing the proof.
	\end{proof}
	
	The following technical lemma will also be relevant.
	
	\begin{lemma} \label{lemma:integral_identity}
		Let $\region \subset \R^d$ be bounded and measurable, let $\activities$ be a real-valued activity function on $\Lambda$, fix $x\in \Lambda$ and let $\activities_t=(1-t)\activities+t \activities_x$ for $t\in[0,1]$. Suppose that $F_{\region}^{\activities_t}(\varepsilon)\neq0$ for all $t\in [0,1]$, then
		\[
		\frac{\diff}{\diff t}\log F_\region^{\activities_t}(\varepsilon)
		=
		\int_\region \frac{\diff\activities_t}{\diff t}(u)
		\left[
		(1+\varepsilon)
		\frac{Z_\region((1+\varepsilon)(\activities_t)_u)}
		{Z_\region((1+\varepsilon)\activities_t)}
		-
		\frac{Z_\region((\activities_t)_u)}
		{Z_\region(\activities_t)}
		\right]
		\diff u\, .
		\]
	\end{lemma}

	We prove Lemma \ref{lemma:integral_identity} in Section \ref{sec:integral-identity-proof}, and proceed by first proving Proposition \ref{prop:expectations}.  
	\begin{proof}[Proof of Proposition \ref{prop:expectations}]
		Set $\ell \coloneqq R/2$. 
		We prove our statement by induction over $c_{\ell}(\region)$, the number of balls of radius $\ell$ required to cover $\region$.
		If $c_{\ell}(\region) = 0$, then $\region = \emptyset$ and \ref{prop:expectations:zerofree}-\ref{prop:expectations:boundary} hold trivially.
		
		Now, for some $n \in \N$, suppose \ref{prop:expectations:zerofree}-\ref{prop:expectations:boundary} hold for all regions that can be covered by $<n$ balls of radius $\ell$, we proceed to show that \ref{prop:expectations:zerofree}-\ref{prop:expectations:boundary} carry over to all regions that can be covered by $n$ balls.
		
		\medskip 
		\paragraph{\texorpdfstring{\textbf{Part 1: (\ref{prop:expectations:zerofree} and \ref{prop:expectations:boundary} for $< n$ imply \ref{prop:expectations:region} for $n$):}}{Part  1: ((i) and (iv) for < n imply (ii) for n)}}
		Fix some $\activities \le \activity$ and some $\region \in \borel_b$ that can be covered by $n$ balls of radius $\ell$, and note that, for any $x \in \region$, $\region' \coloneqq \region \setminus B_{R}(x)$ can be covered by fewer than $n$ balls of radius $\ell = R/2$. 
		In particular, we know by \ref{prop:expectations:zerofree} that $F_{\region'}^{\activities}(\varepsilon) \neq 0$, and so the ratio $F_{\region}^{\activities}(\varepsilon) / F_{\region'}^{\activities}(\varepsilon)$ is well-defined. Defining $\Delta \coloneqq \region \setminus \region'$ and 
		applying the finite-volume DLR equation \eqref{eq:finite_DLR}, we obtain that
		\[
		F_{\region}^{\activities}(\varepsilon)
		= \E_{\bX \sim \region, \activities} [ F_{\region'}^{\activities_{\bX \cap \Delta}}(\varepsilon) \cdot (1 + \varepsilon)^{\size{\bX \cap \Delta}}].
		\]
		Thus, we have 
		\begin{align}\label{eq:twoterms}
			\absolute{\frac{F_{\region}^{\activities}(\varepsilon)}{F_{\region'}^{\activities}(\varepsilon)} - 1} 
			&= \absolute{\E_{\bX \sim \region, \activities} \left[\frac{F_{\region'}^{\activities_{\bX \cap\Delta}}(\varepsilon)}{F_{\region'}^{\activities}(\varepsilon)}\  (1 + \varepsilon)^{\size{\bX \cap \Delta}} \right] - 1} \nonumber \\
			&\le \absolute{\E_{\bX \sim \region, \activities} \left[\left(\frac{F_{\region'}^{\activities_{\bX \cap\Delta}}(\varepsilon)}{F_{\region'}^{\activities}(\varepsilon)} - 1\right)  (1 + \varepsilon)^{\size{\bX \cap \Delta}} \right]}  
			+ \absolute{\E_{\bX \sim \region, \activities}[(1 + \varepsilon)^{\size{\bX \cap \Delta}} - 1]} .
		\end{align}
		We proceed by showing that, for $\epsbound > 0$ sufficiently small as a function of $\activity$, $R$, and $\kappa$, each of the two terms is bounded by $\kappa/2$.
		
		For the first term in \eqref{eq:twoterms} recall that $\region'$ can be covered with $< n$ balls of radius $\ell$.
		For $X \in \pointsets_f$, applying \ref{prop:expectations:boundary}  $\size{X}$ times along a telescoping product and using Lemma \ref{lem:complexproductdifference} yields
		\[
		\absolute{\frac{F_{\region'}^{\activities_{X}}(\varepsilon)}{F_{\region'}^{\activities}(\varepsilon)} - 1} \le \sqrt{\epsbound} \size{X}  e^{\sqrt{\epsbound}  \size{X}} .
		\]
		Moreover, we have $|(1+\eps)^{|X|}|\leq e^{\epsbound |X|}$ and therefore
		\[
		\absolute{\E_{\bX \sim \region, \activities} \left[\left(\frac{F_{\region'}^{\activities_{\bX \cap\Delta}}(\varepsilon)}{F_{\region'}^{\activities}(\varepsilon)} - 1\right)  (1 + \varepsilon)^{\size{\bX \cap \Delta}} \right]}
		\le
		\sqrt{\epsbound} \E_{\bX \sim \region, \activities}\big[\size{\bX \cap \Delta}  e^{ (\sqrt{\epsbound}+\epsbound) \size{\bX \cap \Delta}}\big] \le   O_{R,\lambda}(\sqrt{\eps^\ast})
		\]
		where the final inequality is  by Poisson domination (Lemma \ref{lem:Poisson-domination}) and Lemma \ref{lemma:poisson_moments} \ref{lemma:poisson_moments:linexp}.
		Choosing $\epsbound > 0$ sufficiently small depending on $\activity$, $R$ and $\kappa$, the right-hand side is at most $\kappa/2$.

		It remains to bound the second term in \eqref{eq:twoterms}, i.e., to show that for $\epsbound$ sufficiently small, we have
		\[
		\absolute{\E_{\bX \sim \region, \activities}[(1 + \varepsilon)^{\size{\bX \cap \Delta}} - 1]} \le \kappa/2.
		\]
		To this end, we use the crude bound
		$
		\absolute{(1 + \varepsilon)^{n} - 1} \le \absolute{\varepsilon} n e^{\absolute{\varepsilon} n} \le \epsbound n e^{\epsbound n} 
		$ for any $n \in \mathbb{R}$. 
		We then have
		\[
		\absolute{\E_{\bX \sim \region, \activities}\left[(1 + \varepsilon)^{\size{\bX \cap \Delta}} - 1\right]} \le \E_{\bX \sim \region, \activities}\left[\absolute{(1 + \varepsilon)^{\size{\bX \cap \Delta}} - 1 } \right]\le \epsbound \E_{\bX \sim \region, \activities}\left[ \size{\bX \cap \Delta} e^{\epsbound \size{\bX \cap \Delta}} \right] \leq O_{R,\lambda}(\epsbound)
		\]
		where in the last inequality we used 
		Poisson domination (Lemma \ref{lem:Poisson-domination}) and Lemma \ref{lemma:poisson_moments} \ref{lemma:poisson_moments:linexp}.
		Thus, choosing $\epsbound$ small enough as a function of $R$, $\activity$ and $\kappa$ concludes Part 1.
		
		\medskip 
		
		\paragraph{\texorpdfstring{\textbf{Part 2: (\ref{prop:expectations:zerofree} for $< n$ and \ref{prop:expectations:region} for $n$ imply \ref{prop:expectations:zerofree} for $n$)}:}{Part 2: ((i) for < n and (ii) for n imply (i) for n)}}
		Let $\region \in \borel_b$ be a region that can be covered by $n$ balls of radius $\ell = R/2$.
		If $\region$ can also be covered with $< n$ such balls (i.e., $c_{\ell}(\region) < n$), then the statement follows from the inductive hypothesis.
		Thus, assume the smallest possible cover of $\region$ requires exactly $n$ ball (i.e., $c_{\ell}(\region) = n$).
		Choose a point $x \in \region$ and set $\region' = \region \setminus B_{R}(x)$.
		Assuming that \ref{prop:expectations:region} applies to $\region$, we have
		\[
		\absolute{F_{\region}^{\activities}(\varepsilon)} \ge (1 - \kappa)  \absolute{F_{\region'}^{\activities}(\varepsilon)} .
		\]
		Since $c_{\ell}(\region') \le n-1$, the inductive hypothesis \ref{prop:expectations:zerofree} further yields $\absolute{F_{\region'}^{\activities}(\varepsilon)} \ge (1-\kappa)^{n-1},$ which concludes Part 2.
		
		\medskip 
		
		\paragraph{\texorpdfstring{\textbf{Part 3: (\ref{prop:expectations:zerofree} and \ref{prop:expectations:region} for $n$, and \ref{prop:expectations:boundary} for $<n$  imply \ref{prop:expectations:decay} for $n$)}:}{Part 3: ((i) and (ii) for n, and (iv) for < n imply (iii) for n)}}
		Fix some $\region \in \borel_b$ that can be covered by $n$ balls of radius $\ell$, let $\activities \le \activity$ be some activity function, and let $x \in \R^d$. For $t\in[0,1]$, let $\activities_t=(1-t)\activities+t \activities_x$  and
		\[
		F_t\coloneqq F_{\region}^{\activities_t}(\eps)\, .
		\]
		Note that, assuming \ref{prop:expectations:zerofree} holds on $\region$ for all activity functions bounded by $\activity$, we have $F_t\neq 0$ for all $t\in [0,1]$. In particular, we may apply Lemma~\ref{lemma:ftc} to conclude that 
		\begin{align}\label{eq:ftcapp}
			\frac{F_{\region}^{\activities_x}(\eps)}{F_{\region}^{\activities}(\eps)}= \frac{F_1}{F_0} = \exp \left(\int_{0}^1 \frac{F'_t}{F_t} \diff t \right)\, .
		\end{align}
		It will suffice to show that for all $t\in [0,1]$,
		\begin{align}\label{eq:logderivbd}
			\absolute{\frac{F_t'}{F_t}}\leq \frac{D}{2} e^{-\beta\cdot\dist(x,\region)}
		\end{align}
		for some $\beta=\beta(\alpha), D=D(B, \alpha, \lambda, d)$. Indeed choosing $\delta=\delta(\beta,D)$ sufficiently large so that the right-hand side of~\eqref{eq:logderivbd} is at most 1, 
		it follows from~\eqref{eq:ftcapp} that
		\[
		\absolute{\frac{F_{\region}^{\activities_x}(\eps)}{F_{\region}^{\activities}(\eps)}-1}\leq 2 \absolute{\int_{0}^1 \frac{F'_t}{F_t} \diff t}
		\leq 
		D e^{-\beta\cdot\dist(x,\region)}\, .
		\]
		To establish~\eqref{eq:logderivbd}  we will appeal to Lemma~\ref{lemma:integral_identity} which implies that 
		\begin{align}\label{eq:ftcagain}
			{\frac{F_t'}{F_t}}= 
			\int_\region \frac{d\activities_t}{dt}(u)
			\left[
			(1+\varepsilon)
			\frac{Z_\region((1+\varepsilon)(\activities_t)_u)}
			{Z_\region((1+\varepsilon)\activities_t)}
			-
			\frac{Z_\region((\activities_t)_u)}
			{Z_\region(\activities_t)}
			\right]
			\diff u\, .
		\end{align}
		
		First we observe that
		\[
		\frac{Z_\region((1+\varepsilon)(\activities_t)_u)}
		{Z_\region((1+\varepsilon)\activities_t)}
		=
		\frac{F_\region^{(\activities_t)_u}(\varepsilon)}
		{F_\region^{\activities_t}(\varepsilon)}
		\frac{Z_\region((\activities_t)_u)}{Z_\region(\activities_t)}
		\]
		and $Z_\region((\activities_t)_u)\leq Z_\region(\activities_t)$ since $(\activities_t)_u\leq \activities_t$. Second,
		using the fact that \ref{prop:expectations:region} holds on $\region$, we have 
		\[
		\absolute{\frac{F_\region^{(\activities_t)_u}(\varepsilon)}{F_\region^{\activities_t}(\varepsilon)}}
		\le 
		\frac{1 + \kappa}{1 - \kappa} \cdot \absolute{\frac{F_{\region'}^{(\activities_t)_u}(\varepsilon)}
			{F_{\region'}^{\activities_t}(\varepsilon)}} 
		\leq 3\absolute{\frac{F_{\region'}^{(\activities_t)_u}(\varepsilon)}
			{F_{\region'}^{\activities_t}(\varepsilon)}}\, ,
		\]
		where $\region' = \region \setminus B_R(v)$ for some $v \in \region$.
		Since $\region'$ can be covered by $< n$ balls of radius $\ell$, \ref{prop:expectations:boundary} applies for the region $\region'$ and the activity function $\activities_{t}$, which yields
		\[
		\absolute{\frac{F_{\region'}^{(\activities_t)_u}(\varepsilon)}
			{F_{\region'}^{\activities_t}(\varepsilon)}}
		\le 1 + \sqrt{\epsbound}<2 
		\]
		(taking $\epsbound<1$).
		Returning to~\eqref{eq:ftcagain} and noting
		\[
		\frac{d\activities_t}{dt}(u)
		=
		\activities_x(u)-\activities(u)
		=
		-\activities(u)\big(1-e^{-\potential(x,u)}\big),
		\]
		we conclude, using exponential decay of $\potential$, that
		\[
		\absolute{\frac{F_t'}{F_t}}\leq A \int_\region \activities(u)\big(1-e^{-\potential(x,u)}\big) \diff u
		\leq 
		\lambda AB \int_{\{z:\norm{z}\ge \dist(x,\region)\}}e^{-\alpha\norm{z}}\,\diff z \leq (D/2)e^{-\beta\cdot\dist(x,\region)}
		\]
		where $A$ is an absolute constant and $\beta=\alpha/2, D=D(B, \alpha, \lambda,d)$
		as desired.
		
		\medskip 
		
		\paragraph{\texorpdfstring{\textbf{Part 4: (\ref{prop:expectations:zerofree} and \ref{prop:expectations:decay} for $n$, and \ref{prop:expectations:boundary} for $< n$ imply \ref{prop:expectations:boundary} for $n$):}}{Part 4: ((i) and (iii) for n, and (iv) for < n imply (iv) for n)}}
		Fix a region \(\region\in\borel_b\) that can be covered by \(n\) balls of
		radius \(\ell=R/2\), fix an activity function \(\activities\le\activity\), and
		fix \(x\in\R^d\). By \ref{prop:expectations:zerofree}, applied to
		\(\region\) and all activity functions bounded by \(\activity\), both
		\(F_{\region}^{\activities}(\varepsilon)\) and
		\(F_{\region}^{\activities_x}(\varepsilon)\) are nonzero. Thus the ratio in
		\ref{prop:expectations:boundary} is well-defined.
		We may assume
		\[
		\dist(x,\region)
		\le
		\max\left\{\delta,\beta^{-1}\log(D/\sqrt{\epsbound})\right\}\, ,
		\]
		otherwise the desired bound follows immediately from
		\ref{prop:expectations:decay}. 
		Choose a constant \(A\) sufficiently large, to be fixed below,
		and set
		\[
		S\coloneqq A\log(1/\epsbound).
		\]
		Indeed (assuming $\epsbound<1/2$) we choose $A=A(\delta, R, \beta, D, C,m)$ sufficiently large so that
		\begin{align}\label{eq:S-bound}
			S\ge \dist(x,\region)+R+\delta, \quad \quad De^{-\beta S}\leq \epsbound  \quad \text{and} \quad C S^d e^{-mS}\le \sqrt{\epsbound}/32 \, .
		\end{align} 
		Define
		\[
		\Delta_1\coloneqq \region\cap B_S(x),
		\qquad
		\Delta_2\coloneqq \region\cap B_{2S}(x),
		\qquad
		\Delta_a\coloneqq \Delta_2\setminus \Delta_1,
		\]
		and
		\[
		\region_2\coloneqq \region\setminus \Delta_2.
		\]
		
		\begin{figure}
			\centering
			\includegraphics[width=0.5\linewidth]{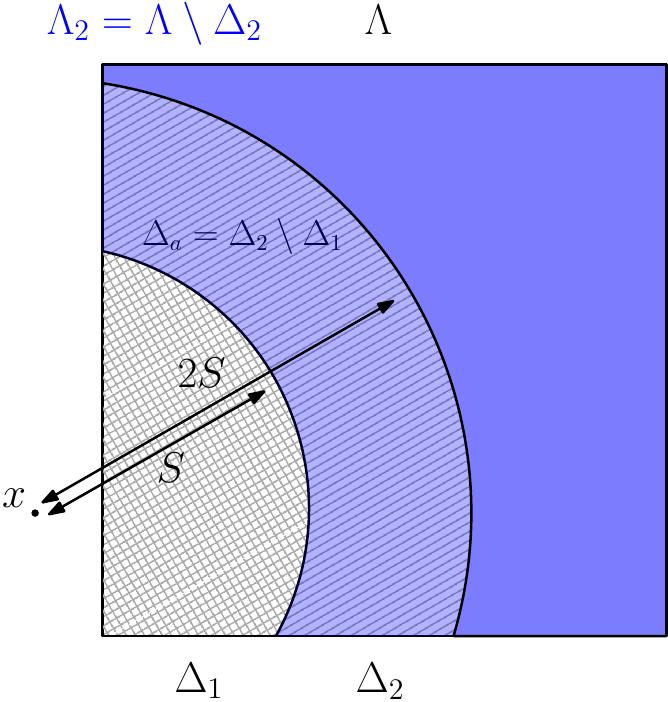} \hspace{1.7cm}
			\caption{Illustration of the regions considered in step 4. The region $\Delta_1 = \Lambda \cap B_S(x)$ is represented by the checkered pattern, the region $\Delta_2 = \Lambda \cap B_{2S}(x)$ is represented by the lined pattern. In the proof, we handle the expected influence of points in the ball $\Delta_1$ and the annulus $\Delta_a$ on the region $\Lambda_2$ separately. }
			\label{fig:regions}
		\end{figure}
		See  \Cref{fig:regions} for an illustration of these regions.
		
		We note that since \(S\ge \dist(x,\region)+R\), the set \(\region_2\) can be covered by \(<n\) balls of radius \(\ell\). Therefore the
		induction hypothesis applies to \(\region_2\).

		The high-level idea for deriving \ref{prop:expectations:boundary} is as follows. To bound $\absolute{\frac{F_{\region}^{\activities_x}(\varepsilon)}{F_{\region}^{\activities}(\varepsilon)} - 1}$, it is sufficient to show that  $$\frac{F_{\region}^{\activities_x}(\varepsilon)}{F_{\region_2}^{\activities}(\varepsilon)} - \frac{F_{\region}^{\activities}(\varepsilon)}{F_{\region_2}^{\activities}(\varepsilon)}$$ is small in absolute value because $\Lambda_2$ covers the majority of $\Lambda$.

		The key object is the annulus observable
		\[
		\Gamma(W)
		\coloneqq
		(1+\varepsilon)^{\size W}
		\frac{F_{\region_2}^{\activities_W}(\varepsilon)}
		{F_{\region_2}^{\activities}(\varepsilon)},
		\qquad
		W\subseteq \Delta_a.
		\]
		We will show that
		\begin{align}\label{eq:expect-approx}
			\frac{F_{\region}^{\activities}(\varepsilon)}
			{F_{\region_2}^{\activities}(\varepsilon)}
			\approx
			\E_{\bX\sim\region,\activities}
			\Gamma(\bX\cap\Delta_a)
			\quad
			\text{and}
			\quad
			\frac{F_{\region}^{\activities_x}(\varepsilon)}
			{F_{\region_2}^{\activities}(\varepsilon)}
			\approx
			\E_{\bX\sim\region,\activities_x}
			\Gamma(\bX\cap\Delta_a)\, .
		\end{align}
		We have approximations rather than equalities here since we are ignoring the influence of points in $\Delta_1$ on $\region_2$. However $\dist(\Delta_1, \Lambda_2)\geq S$ and so by \ref{prop:expectations:decay} their effect in $\region_2$ is small. 
		It remains to compare the two expectations in~\eqref{eq:expect-approx} and this is where spatial mixing enters. The observable \(\Gamma(\bX\cap\Delta_a)\)
		depends only on the configuration in the annulus \(\Delta_a\) and since
		\(\dist(x,\Delta_a)\ge S\), single-site SSM will allow us to bound $\left\|
		\mu_{\region,\activities_x}
		-
		\mu_{\region,\activities}
		\right\|_{\Delta_a}\,.$
		This in turn will allow us to bound the difference in expectations at~\eqref{eq:expect-approx} since $\Gamma$ is typically small in absolute value.

		We now continue with the formal proof and proceed by writing
		\begin{align}\label{eq:numerator-denominator}
			\absolute{\frac{F_{\region}^{\activities_x}(\varepsilon)}{F_{\region}^{\activities}(\varepsilon)} - 1}
			= 
			\frac{\absolute{\frac{F_{\region}^{\activities_x}(\varepsilon)}{F_{\region_2}^{\activities}(\varepsilon)} - \frac{F_{\region}^{\activities}(\varepsilon)}{F_{\region_2}^{\activities}(\varepsilon)}}}{\absolute{\frac{F_{\region}^{\activities}(\varepsilon)}{F_{\region_2}^{\activities}(\varepsilon)}}} \, .
		\end{align}
		Our goal is to upper-bound the numerator by $\frac{\sqrt{\epsbound}}{2}$ and lower-bound the denominator by $\frac{1}{2}$.
		
		We start with bounding the numerator.
		First we collect two estimates for \(\Gamma\). For every finite $W\subseteq \Delta_a$,
		\begin{align}\label{eq:Gamma-absolute-bound}
			|\Gamma(W)|
			\le
			e^{2\sqrt{\epsbound}\size W}\, ,\quad \text{and}
		\end{align}
		\begin{align}\label{eq:Gamma-close-to-one}
			|\Gamma(W)-1|
			\leq 
			2\sqrt{\epsbound}\,\size W\,e^{2\sqrt{\epsbound}\size W}\, .
		\end{align}
		Indeed, by the inductive hypothesis \ref{prop:expectations:boundary} applied to $\region_2$, and using a telescoping product together with Lemma \ref{lem:complexproductdifference} we have,
		\begin{align} \label{eq:annulus-ratio-close}
			\absolute{
				\frac{
					F_{\region_2}^{\activities_W}(\varepsilon)}
				{F_{\region_2}^{\activities}(\varepsilon)}
				-1
			}
			\le
			\sqrt{\epsbound}\absolute{W}e^{\sqrt{\epsbound}\absolute{W}}.
		\end{align}
		Moreover, the same telescoping product
		gives the bound
		\begin{align}\label{eq:annulus-ratio-bd}
			\left|
			\frac{
				F_{\region_2}^{\activities_W}(\varepsilon)}
			{F_{\region_2}^{\activities}(\varepsilon)}
			\right|
			\le
			(1+\sqrt{\epsbound})^{\size W}
			\le
			e^{\sqrt{\epsbound}\size W}.
		\end{align}
		Since \(|1+\varepsilon|\le e^{\epsbound}\le e^{\sqrt{\epsbound}}\) (assuming \(\epsbound<1\)) we have
		~\eqref{eq:Gamma-absolute-bound}.
		Finally, applying the inequality $|zw-1|\leq |z||w-1|+|z-1|$
		to the product defining \(\Gamma\), we get
		\begin{align*}
			|\Gamma(W)-1|
			\le
			e^{\sqrt{\epsbound}\absolute{W}}
			\absolute{
				\frac{
					F_{\region_2}^{\activities_W}(\varepsilon)}
				{F_{\region_2}^{\activities}(\varepsilon)}
				-1
			}
			+
			\absolute{(1+\eps)^{|W|}-1}
			\leq 
			2\sqrt{\epsbound}\,\size W\,e^{2\sqrt{\epsbound}\size W}\, ,
		\end{align*}
		where we used~\eqref{eq:annulus-ratio-close} and  Lemma~\ref{lem:complexproductdifference} to bound $\absolute{(1+\eps)^{|W|}-1}\leq \epsbound\absolute{W}e^{\epsbound\absolute{W}}$.
		
		Suppose now that \(U\subseteq \Delta_1\cup\{x\}\) is finite and note that every point of
		\(U\) has distance at least \(S\) from \(\region_2\). Since \(S\ge\delta\), the
		inductive hypothesis \ref{prop:expectations:decay} 
		applied to \(\region_2\), gives
		\[
		\left|
		\frac{
			F_{\region_2}^{\pmb{\xi}_u}(\varepsilon)}
		{F_{\region_2}^{\pmb{\xi}}(\varepsilon)}
		-1
		\right|
		\le
		D e^{-\beta S}\leq \epsbound
		\]
		for every activity function \(\pmb{\xi}\le\activity\) and every \(u\in U\).
		Therefore, telescoping over the points of \(U\) and applying
		Lemma~\ref{lem:complexproductdifference}, we obtain, uniformly in
		\(W\subseteq\Delta_a\),
		\begin{align}\label{eq:inner-to-outer-small}
			\absolute{
				\frac{F_{\region_2}^{\activities_{W\cup U}}(\varepsilon)}
				{F_{\region_2}^{\activities_W}(\varepsilon)}
				-1
			}
			\le
			\epsbound\absolute{U}
			e^{\epsbound\absolute{U}}\, .
		\end{align}
		We now make~\eqref{eq:expect-approx} formal.
		We first treat the activity \(\activities\) (the argument with \(\activities_x\) will be identical). By the finite-volume DLR equation~\eqref{eq:finite_DLR},
		conditioning on the configuration in
		\(\Delta_1\cup\Delta_a=\region\setminus\region_2\), we have
		\begin{align}\label{eq:DLR-screen-no-x}
			\frac{F_{\region}^{\activities}(\varepsilon)}
			{F_{\region_2}^{\activities}(\varepsilon)}
			&=
			\E_{\bX\sim\region,\activities}
			\left[
			(1+\varepsilon)^{\size{\bX\cap\Delta_1}
				+\absolute{\bX\cap\Delta_a}}
			\frac{F_{\region_2}^{\activities_{\bX\cap(\Delta_1\cup\Delta_a)}}(\varepsilon)}
			{F_{\region_2}^{\activities}(\varepsilon)}
			\right].
		\end{align}
		Fix a configuration \(X\) and write \(W=X\cap\Delta_a\) and \(U=X\cap\Delta_1\).
		Then the integrand in \eqref{eq:DLR-screen-no-x} is
		\[
		\Gamma(W)\,
		(1+\varepsilon)^{\absolute{U}}
		\frac{F_{\region_2}^{\activities_{W\cup U}}(\varepsilon)}
		{F_{\region_2}^{\activities_W}(\varepsilon)}\,  ,
		\]
		and so its difference from \(\Gamma(W)\) is bounded by
		\[
		|\Gamma(W)|
		\absolute{
			(1+\eps)^{\absolute{U}}
			\frac{F_{\region_2}^{\activities_{W\cup U}}(\varepsilon)}
			{F_{\region_2}^{\activities_W}(\varepsilon)}
			-1
		}\, .
		\]
		Again using the inequality $|zw-1|\leq |z||w-1|+|z-1|$
		together with \eqref{eq:inner-to-outer-small} and the bound
		\(
		|(1+\eps)^{\absolute{U}}-1|
		\le
		\epsbound\absolute{U}e^{\epsbound \absolute{U}}
		\), we obtain
		\begin{align*}
			\absolute{
				(1+\varepsilon)^{\absolute{U}}
				\frac{F_{\region_2}^{\activities_{W\cup U}}(\varepsilon)}
				{F_{\region_2}^{\activities_W}(\varepsilon)}
				-1
			}
			\le
			2\epsbound \absolute{U}
			e^{2\epsbound \absolute{U}}\, ,
		\end{align*}
		Combining this with \eqref{eq:Gamma-absolute-bound}, and using Poisson
		domination on the set \(\Delta_2\), whose volume is bounded by \(C_dS^d\), gives
		\begin{align}\label{eq:screen-approx-no-x}
			\absolute{
				\frac{F_{\region}^{\activities}(\varepsilon)}
				{F_{\region_2}^{\activities}(\varepsilon)}
				-
				\E_{\bX\sim\region,\activities}\Gamma(\bX\cap\Delta_a)
			}
			\le
			\E_{\bY \sim \Pois(\activity \vol(\Delta_2))}
			\left[
			2\epsbound \bY e^{4\sqrt{\epsbound}\bY}
			\right]=O_{\lambda,d}(\epsbound S^d)
		\end{align}
		where for the final inequality we used Lemma~\ref{lemma:poisson_moments} part \ref{lemma:poisson_moments:linexp}.
		
		By an identical argument
		\begin{align}\label{eq:screen-approx-with-x}
			\absolute{
				\frac{F_{\region}^{\activities_x}(\varepsilon)}
				{F_{\region_2}^{\activities}(\varepsilon)}
				-
				\E_{\bX\sim\region,\activities_x}\Gamma(\bX\cap\Delta_a)
			}
			=O_{\lambda,d}(\epsbound S^d)\, .
		\end{align}

		We next compare the two expectations in~\eqref{eq:screen-approx-no-x} and \eqref{eq:screen-approx-with-x}. By \eqref{eq:Gamma-absolute-bound},
		\[
		|\Gamma(X\cap\Delta_a)|\le e^{2\sqrt{\epsbound} \absolute{X\cap\Delta_a}}.
		\]
		Choose
		\[
		K=L\big(S^d+\log(1/\epsbound)\big),
		\]
		with \(L=L(d,\lambda)\) sufficiently large. By Poisson domination, Lemma~\ref{lemma:poisson_moments} part~\ref{lemma:poisson_moments:exp}, and the bound
		\(
		\P(\Pois(\rho)>K)\leq (e\rho/K)^K
		\)
		we then have
		\begin{align}\label{eq:screen-tail}
			\E_{\bX\sim\region,\activities}
			\left[
			|\Gamma(\bX\cap\Delta_a)|
			\ind{\absolute{\bX\cap\Delta_a}>K}
			\right]
			+
			\E_{\bX\sim\region,\activities_x}
			\left[
			|\Gamma(\bX\cap\Delta_a)|
			\ind{\absolute{\bX\cap\Delta_a}>K}
			\right]
			\le
			\sqrt{\epsbound}/8 .
		\end{align}
		On the other hand \eqref{eq:Gamma-absolute-bound} gives
		\[
		\absolute{\Gamma(\bX\cap\Delta_a)\ind{\absolute{\bX\cap\Delta_a}\leq K}}
		\le
		e^{2\sqrt{\epsbound} K}
		\le 2
		\]
		after shrinking \(\epsbound\) if necessary. 
		We conclude that
		\begin{align*}
			\absolute{
				\E_{\bX\sim\region,\activities_x}
				\Gamma(\bX\cap\Delta_a)
				-
				\E_{\bX\sim\region,\activities}
				\Gamma(\bX\cap\Delta_a)
			}
			\le
			4
			\left\|{
				\mu_{\region,\activities_x}
				-
				\mu_{\region,\activities}
			}\right\|_{\Delta_a}
			+
			\sqrt{\epsbound}/8\, .
		\end{align*}
		Now, since \(\dist(x,\Delta_a)\ge S\), we have by single-site SSM with constants $C < \infty$ and $m > 0$,
		\[
		\left\|{
			\mu_{\region,\activities_x}
			-
			\mu_{\region,\activities}}\right\|_{\Delta_a}
		\le
		C\vol(\Delta_a)e^{-mS}
		\le
		C S^d e^{-mS} \leq \sqrt{\epsbound}/32\, .
		\]
		We thus obtain
		\begin{align}\label{eq:screen-expectation-comparison}
			\left|
			\E_{\bX\sim\region,\activities_x}
			\Gamma(\bX\cap\Delta_a)
			-
			\E_{\bX\sim\region,\activities}
			\Gamma(\bX\cap\Delta_a)
			\right|
			\le
			\sqrt{\epsbound}/4 .
		\end{align}

		Combining \eqref{eq:screen-approx-no-x},
		\eqref{eq:screen-approx-with-x}, and
		\eqref{eq:screen-expectation-comparison} and choosing
		\(\epsbound>0\) sufficiently small we get
		\begin{align*}
			\absolute{
				\frac{F_{\region}^{\activities_x}(\varepsilon)}{F_{\region_2}^{\activities}(\varepsilon)}
				-
				\frac{F_{\region}^{\activities}(\varepsilon)}{F_{\region_2}^{\activities}(\varepsilon)}
			}
			\le\sqrt{\epsbound}/2\, .
		\end{align*}
		It remains to lower-bound the denominator in~\eqref{eq:numerator-denominator}.
		By \eqref{eq:screen-approx-no-x}, it suffices to show that
		\(
		\E_{\bX\sim\region,\activities}
		\Gamma(\bX\cap\Delta_a)
		\)
		is close to \(1\). Using \eqref{eq:Gamma-close-to-one}, Poisson domination, and Lemma~\ref{lemma:poisson_moments}
		\begin{align*}
			\absolute{
				\E_{\bX\sim\region,\activities}
				\Gamma(\bX\cap\Delta_a)
				-1
			}
			\le
			2\sqrt{\epsbound}\,
			\E_{\bX\sim\region,\activities}
			\left[
			\absolute{\bX\cap\Delta_a}e^{2\sqrt{\epsbound} \absolute{\bX\cap\Delta_a}}
			\right]
			=O_{d,\lambda}(\sqrt{\epsbound}S^d)
		\end{align*}
		Hence, after shrinking \(\epsbound\) if necessary,
		\[
		\absolute{
			\E_{\bX\sim\region,\activities}
			\Gamma(\bX\cap\Delta_a)
			-1
		}
		\le
		1/4.
		\]
		Together with \eqref{eq:screen-approx-no-x}, and shrinking \(\epsbound\) once
		more if necessary, this yields
		\begin{align*}
			\absolute{
				\frac{F_{\region}^{\activities}(\varepsilon)}
				{F_{\region_2}^{\activities}(\varepsilon)}
				-1
			}
			\le
			1/2
		\end{align*}
		as desired. Returning to~\eqref{eq:numerator-denominator} we conclude that 
		\[
		\absolute{
			\frac{F_{\region}^{\activities_x}(\varepsilon)}
			{F_{\region}^{\activities}(\varepsilon)}
			-1
		}
		\le
		\sqrt{\epsbound}
		\, .
		\]
		This proves \ref{prop:expectations:boundary} for \(\region\), and closes the
		induction. \end{proof}

	\subsection{Proof of Lemma \texorpdfstring{\ref{lemma:integral_identity}}{5.4}} \label{sec:integral-identity-proof}
	Before proving Lemma \ref{lemma:integral_identity}, we record the following basic identity.
	\begin{lemma}
		\label{lem:log-derivative-F}
		Let \(\region\subset\R^d\) be bounded and measurable and let $\pmb{\xi}, \pmb{h}$ be bounded complex-valued activity functions on $\Lambda$. Then
		\[
		\frac{d}{ds}Z_\region(\pmb{\xi}+s\pmb{h})\bigg|_{s=0}
		=
		\int_\region \pmb{h}(u)Z_\region(\pmb{\xi}_u)\,\diff u\, . 
		\]
	\end{lemma}
	\begin{proof}
		Recall that
		\[
		Z_\region(\pmb{\xi})
		=
		\sum_{k\ge0}\frac1{k!}
		\int_{\region^k}
		\prod_{i=1}^k \pmb{\xi}(x_i)
		e^{-H(x_1,\ldots,x_k)}
		\,\diff x_1\cdots \diff x_k\, .
		\]
		We differentiate term by term. First note that
		\[
		\frac{d}{ds}
		\prod_{i=1}^k(\pmb{\xi}(x_i)+s\pmb{h}(x_i))
		\bigg|_{s=0}
		=
		\sum_{j=1}^k
		\pmb{h}(x_j)\prod_{i\ne j}\pmb{\xi}(x_i).
		\]
		By symmetry, it follows that
		\begin{align}
			\frac{d}{ds}Z_\region(\pmb{\xi}+s\pmb{h})\bigg|_{s=0}
			&=
			\sum_{k\ge1}\frac{k}{k!}
			\int_{\region^k}
			\pmb{h}(x_1)
			\prod_{i=2}^k \pmb{\xi}(x_i)
			e^{-H(x_1,\ldots,x_k)}
			\,\diff x_1\cdots \diff x_k \nonumber \\
			&=
			\int_\region \pmb{h}(u)
			\sum_{m\ge0}\frac1{m!}
			\int_{\region^m}
			\prod_{i=1}^m \pmb{\xi}(y_i)
			e^{-H(u,y_1,\ldots,y_m)}
			\,\diff y_1\cdots \diff y_m
			\,\diff u , \label{eq:direct-deriv}
		\end{align}
		where dominated convergence justifies differentiation under the sum and integral in the first equality, and the interchange of the sum and integral in the second. 
		Since
		\[
		H(u,y_1,\ldots,y_m)
		=
		H(y_1,\ldots,y_m)
		+
		\sum_{i=1}^m \potential(u,y_i),
		\]
		the inner sum in~\eqref{eq:direct-deriv} is precisely $Z_\region(\pmb{\xi}_u)$.
	\end{proof}

	We are now ready to show Lemma \ref{lemma:integral_identity}
	\begin{proof}[Proof of Lemma \ref{lemma:integral_identity}]
		Recall that
		\begin{align}\label{eq:Frecall}
			F_\region^{\activities_t}(\varepsilon)
			=
			\frac{Z_\region((1+\varepsilon)\activities_t)}
			{Z_\region(\activities_t)} ,
		\end{align}
		We apply Lemma~\ref{lem:log-derivative-F} twice. First take $\pmb{\xi}=(1+\varepsilon)\activities_t$ and $\pmb{h}=\frac{d}{dt} \pmb{\xi}$.
		Then
		\[
		\frac{d}{dt}Z_\region((1+\varepsilon)\activities_t)
		=
		(1+\varepsilon)
		\int_\region
		\frac{d\activities_t}{dt}(u)
		Z_\region((1+\varepsilon)(\activities_t)_u)
		\,\diff u .
		\]
		Dividing by \(Z_\region((1+\varepsilon)\activities_t)\), which is nonzero
		because \(F_\region^{\activities_t}(\varepsilon)\ne0\), gives
		\begin{align}\label{eq:logderiv1}
			\frac{d}{dt}
			\log
			Z_\region((1+\varepsilon)\activities_t)
			=
			(1+\varepsilon)
			\int_\region
			\frac{d\activities_t}{dt}(u)
			\frac{
				Z_\region((1+\varepsilon)(\activities_t)_u)}{
				Z_\region((1+\varepsilon)\activities_t)}
			\,\diff u .
		\end{align}
		
		Second, take $\pmb{\xi}=\activities_t$ and $\pmb{h}=\frac{d}{dt} \pmb{\xi}$.
		We have \(Z_\region(\activities_t)>0\) since $\activities_t$ is real and non-negative, thus
		\begin{align}\label{eq:logderiv2}
			\frac{d}{dt}\log Z_\region(\activities_t)
			=
			\int_\region
			\frac{d\activities_t}{dt}(u)
			\frac{Z_\region((\activities_t)_u)}
			{Z_\region(\activities_t)}
			\,\diff u .
		\end{align}
		The result follows by taking the derivative of the logarithm of~\eqref{eq:Frecall} and applying~\eqref{eq:logderiv1} and~\eqref{eq:logderiv2}.
	\end{proof}

	\subsection{Proof of Theorem \texorpdfstring{\ref{thm:zero_freeness}}{5.1}}
	Before we can prove Theorem \ref{thm:zero_freeness}, we first establish a connection between analyticity of the infinite-volume pressure and zero-freeness of partition functions in a uniform neighborhood of the real axis.
	To this end, we show the following version of the Yang--Lee theorem \cite{yang1952statistical}, following the proof of the analogous statement for the Ising model in the book by Friedli and Velenik \cite[Theorem 3.42]{friedli2017statistical}. 
	\begin{proposition}[Yang--Lee Theorem] \label{prop:leeyang}
		Let $\potential$ be a repulsive potential, and let $(\region_n)_{n \in \N}$ be a sequence of bounded measurable regions with $\region_n \nearrow \R^{d}$ such that the infinite-volume limit of the pressure $p(\activity) = \lim_{n \to \infty} p_{\region_n}(\activity)$ exists for all $\activity \in \R_{\ge 0}$.
		Let $D \subseteq \C$ be open, simply-connected, and such that $D \cap \R$ is some open interval and $D \cap \R_{\ge 0} \neq \emptyset$.
		If $Z_{\region_n}(z) \neq 0$ for all $n \in \N$ and $z \in D$, then the infinite-volume pressure $p$ has an analytic continuation to $D$.
	\end{proposition}
	\begin{proof}
		First note that $Z_{\region_n}(\activity) > 0$ for $\activity \in D \cap \R_{\ge 0}$. 
		Thus, if $D \cap \R$ is an open interval with $D \cap \R_{\ge 0} \neq \emptyset$ and $Z_{\region_n} \neq 0$ on $D$, the intermediate value theorem yields $Z_{\region_n} > 0$ on $D \cap \R$.
		Consequently, $\log(Z_{\region_n})$  has an analytic continuation to $D$ such that $\Re(\log(Z_{\region_n})) = \log(\absolute{Z_{\region_n}})$ (see \cite[Theorem B.23 \& Remark B.24]{friedli2017statistical}).
		Using this continuation, we define the analytic continuation of the finite-volume pressure $p_{\region_n}$ to $D$ for every $n \in \N$.
		Now, using the analytic continuation of $p_{\region_n}$, define $h_n (z) = e^{p_{\region_n}(z)}$, which is analytic on $D$.
		Our first step is to apply Vitali's convergence theorem (see \cite[Theorem B.25]{friedli2017statistical}) to the sequence $(h_n)_{n \in \N}$.
		To this end, note that the limit $h(\lambda) = \lim_{n \to \infty} h_n (\lambda) = e^{\lim_{n \to \infty} p_{\region_n}(\lambda)} = e^{p(\lambda)}$ exists for all $\lambda \in D \cap \R_{\ge 0}$ and, since $D \cap \R$ is an open interval and $D \cap \R_{\ge 0} \neq \emptyset$, $D \cap \R_{\ge 0}$ must contain an accumulation point.
		Moreover, since $\absolute{h_n(z)} = e^{\Re(p_{\region_n}(z))} = e^{\log(\absolute{Z_{\region_n}(z)})/\vol(\region_n)}$ and $\absolute{Z_{\region_n}(z)} \le e^{\absolute{z} \cdot \vol(\region_n)}$, we see that the sequence $(h_n)_{n \in \N}$ is locally uniformly bounded.
		Thus, by Vitali's theorem, the limit $h = \lim_{n \to \infty} h_n$ is analytic on $D$, and the convergence $h_n \to h$ is locally uniform.
		Using the fact that $h(0) = e^{\lim_{n \to \infty} 1/\vol(\region_n)} = 1$, we further know by Hurwitz's theorem (see \cite[Theorem B.26]{friedli2017statistical}) that $h \neq 0$ on all of $D$.
		Moreover, since $Z_{\region_n}$ is real on $\R$ for all $n$, $h$ is strictly positive on $D \cap \R$. 
		In particular, $\log(h)$ has an analytic continuation to $D$, and since $\log(h(\lambda)) = p(\lambda)$ for $\lambda \in D \cap \R_{\ge 0}$, this yields an analytic continuation of $p$ to $D$.
	\end{proof}
	
	Using Proposition \ref{prop:expectations} and Proposition \ref{prop:leeyang}, we can now prove Theorem \ref{thm:zero_freeness}.
	\begin{proof}[Proof of Theorem \ref{thm:zero_freeness}]
		By Proposition \ref{prop:leeyang}, part \ref{thm:zero_freeness:analyticity} of Theorem \ref{thm:zero_freeness} follows immediately from part \ref{thm:zero_freeness:roots}.
		Next, let $\epsbound > 0$ be as in Proposition \ref{prop:expectations}.
		By classical techniques such as the cluster expansion, there is some $r>0$ such that $Z_{\region_n}(z) \neq 0$ for all $z \in \C$ with $\absolute{z} \le r$ (see, e.g., \cite[Theorem~4.5.2]{ruelle1969statistical}).
		If $r > \activity$, there is nothing to prove.
		Otherwise, note that every $z$ in the complex $\epsbound r/2$-neighborhood of $[r/2, \activity]$ can be written as $z = \activity' \cdot (1 + \varepsilon)$ for some $\activity' \in [r/2, \activity]$ and $\absolute{\varepsilon} \le \epsbound$.
		Thus, applying part \ref{prop:expectations:zerofree} of Proposition \ref{prop:expectations} implies 
		\[
		Z_{\region_n}(z) = F_{\region_n}^{\activity'}(\varepsilon)  \cdot Z_{\region_n}(\activity') \neq 0,
		\]
		which concludes the proof of part \ref{thm:zero_freeness:roots}.

		For the proof of \ref{thm:zero_freeness:log-bound}, we again note that by, e.g., \cite[Theorem~4.5.2]{ruelle1969statistical}) the conclusion holds for all $|w| \leq r$ for some $r$.  It is thus sufficient to show the conclusion holds for some $\delta$-neighborhood of $[r,\lambda]$.  There is some $\eps_\ast > 0$ so that for $R = 1$ and $\kappa = 1/4$, the conclusion of Proposition \ref{prop:expectations} holds for all $\lambda' \leq \lambda$.    
		Find a sequence of sets $\Lambda_n = \Lambda^{(0)}, \Lambda^{(1)},\ldots, \Lambda^{(N)} = \emptyset$ with $N = \Theta_{d}(\Vol(\Lambda_n))$ so that for each $j \in \{1,\ldots,N\}$ there is some point $x_j$ with $\Lambda^{(j)} = \Lambda^{(j-1)} \setminus B_R(x_j)$.  For all $|\eps| \leq \eps_\ast$ define \begin{equation}
			f_j(\eps) = F_{\Lambda^{(j)}}^{\activity'}(\eps), \qquad g_j(\eps) = \frac{f_{j-1}(\eps)}{f_{j}(\eps)}\,.
		\end{equation}
		
		Note that $f_j$ is analytic in $\eps$.  By Proposition \ref{prop:expectations} \ref{prop:expectations:zerofree}, we have $f_{j}(\eps) \neq 0$ for all $|\eps| \leq \eps_\ast$ and so $g_j$ is analytic.  Since $f_N \equiv 1$ we have \begin{equation}\label{eq:Z-telescope}   \prod_{j = 1}^N g_j (\eps)=f_0(\eps)=\frac{Z_{\Lambda_n}((1 + \eps)\lambda')}{Z_{\Lambda_n}(\lambda')}\,.\end{equation}
		
		Since $|g_j(\eps) - 1| \leq 1/4$, if we let $\mathrm{Log}$ be the principal branch of the complex logarithm then $\mathrm{Log}(g_j)$ is analytic. Define $h_j = \mathrm{Log}(g_j)$ and $H = \sum_{j = 1}^N h_j$.  We note that $H$ is analytic in $|\eps| \leq \eps_\ast$ since it is a sum of analytic functions.  Further, \eqref{eq:Z-telescope} implies $e^{H(\eps)}  =Z_{\Lambda_n}((1 + \eps)\lambda')/Z_{\Lambda_n}(\lambda')$ and so $H + \log Z_{\Lambda_n}(\lambda')$ is an analytic branch of $\log Z_{\Lambda_n}((1 + \eps)\lambda')$.  Since $|g_j(\eps) - 1| \leq 1/4$, we have that $|h_j(\eps)| \leq 1$; noting that $0 \leq \log Z_{\Lambda_n}(\lambda) \leq \lambda \Vol(\Lambda_n)$ implies $|H(\eps) + \log Z_{\Lambda_n}(\lambda')| \leq N + \lambda'\Vol(\Lambda_n)$ for all $|\eps| \leq \eps_\ast$.  Since $N = \Theta(\Vol(\Lambda_n))$, this completes the proof.
	\end{proof}
	
	\section{Optimal mixing of block dynamics from spatial mixing}\label{sec:block-dynamics}
	We prove that single-site SSM implies that the mixing time of the block dynamics with a sufficiently large update radius is $O(\vol(\region) \log(\vol(\region)))$.  
	Recall that, for Borel $\region \subseteq \R^d$, $\pointsets_{\region}$ denotes the set of all locally-finite point configurations $X \in \pointsets$ such that $X \subseteq \Lambda$, and that we write $\events_{\region}$ for the trace of $\pointsets_{\region}$ in $\events$.
	In particular, if $\region$ is bounded then all configurations in $\pointsets_{\region}$ are finite.

	\begin{definition}[Block dynamics] \label{def:block_dynamics}
		Let $\region \subset \R^d$ be bounded and measurable, and let $L > 0$.
		A Markov chain $(\bX_t)_{t \ge 0}$ on $(\pointsets_{\region}, \events_{\region})$ is called block dynamics on $\region$ with update radius $L$ if it evolves according to the following update rule for every time step $t \in \N_0$:
		\begin{enumerate}
			\item Choose $\randpoint \in \region$ uniformly at random.
			\item Choose $\bY \sim \mu_{\region \cap B_L(\randpoint), \activities_{\bX_t \setminus B_L(\randpoint)}}$.
			\item Set $\bX_{t+1} = (\bX_t \setminus B_L(\randpoint)) \cup \bY$.
		\end{enumerate}
	\end{definition}
	It is easily checked that for every $A \subseteq \events_{\region}$, the probability that $\bX_{t+1} \in A$ is a measurable function of $\bX_t$, and thus the construction above describes the transition kernel of a Markov chain.
	
	Our goal is to bound the speed of convergence of the law of $\bX_t$ as given above to its stationary distribution as $t \to \infty$.  
	Recall that for a pair of Borel probability measures $\mu_1,\mu_2$ on $\pointsets_{\region}$, we define the total variation distance via $$\| \mu_1 - \mu_2\|_{\mathrm{TV}} := \sup_{A \in \events_{\region}} |\mu_1(A) - \mu_2(A)|\,.$$
	
	We also note that  \begin{equation}\label{eq:TV-coupling}
		\|\mu_1 - \mu_2\|_{\mathrm{TV}} = \inf_{(\bX,\bY)} \P(\bX \neq \bY)
	\end{equation}
	where the infimum is over all couplings $(\bX,\bY)$ of $(\mu,\nu)$ (see, e.g., \cite[Chapter 1.5]{lindvall1992lectures} for a statement in this level of generality).
	
	For each $\eps > 0$, the mixing time of a Markov kernel $P$ with unique stationary distribution $\pi$ is defined by $\tau_{\mathrm{mix}}(\eps) = \min\{t : \sup_{x}\| \delta_x P^t - \pi\|_{\mathrm{TV}} \leq \eps\}.$  The main theorem of this section shows that we have optimal mixing of block dynamics.
	
	\begin{theorem}\label{th:mixing-block-dynamics}
		Let $\activity \ge 0$, and suppose that $\potential$ is repulsive and satisfies single-site SSM up to $\activity$ (see Definition \ref{def:spatial_mixing}) with constants $m$ and $C$. 
		Then, there exists some $L_0 \coloneqq L_0(\activity, C, m, d) > 0$ such that, for all activity functions $\activities \le \activity$ and all $L \ge L_0$, the block dynamics with update radius $L$ have mixing time $O(\vol(\region) \cdot \log(\vol(\region)))$ on boxes $\region = [-n, n]^d \subset \R^d$.
	\end{theorem}
	
	Of course Theorem \ref{thm:block-dynamics} is now immediate.
	
	\begin{proof}[Proof of Theorem~\ref{thm:block-dynamics}]
		Single-site SSM follows from Theorem \ref{thm:spatial-mixing} after which Theorem \ref{th:mixing-block-dynamics} completes the proof.
	\end{proof}
	
	The proof of Theorem \ref{th:mixing-block-dynamics} will proceed through the machinery of \emph{path coupling} introduced by Bubley and Dyer \cite{bubley1997path}.  We directly follow the proof in \cite[Chapter~14]{levin2017markov}.  
	For $X, Y \in \pointsets_{\region}$, let $\Delta(X,Y) := |X \triangle Y|$ where $\triangle$ is the symmetric difference.  We will also make use of the Wasserstein metric \begin{equation}\label{eq:wasserstein-definition}
		W_\Delta(\mu,\nu) := \inf_{(\bX,\bY)} \E \Delta(\bX,\bY)
	\end{equation}
	where the infimum is over all couplings $(\bX,\bY)$ of $(\mu,\nu)$.  
	If we define $L^1$ to be the set of Borel probability measures $\mu$ on $\pointsets_{\Lambda,f}$ so that $\int \Delta(X,\emptyset)\,\diff\mu(X) < \infty$, then by Kantorovich-Rubinstein duality \begin{equation}\label{eq:wasserstein-duality}
		W_\Delta(\mu,\nu) = \sup_{F:\mathrm{Lip}(F) \leq 1}| \mu F - \nu F| \quad \text{ for all } \mu,\nu \in L^1
	\end{equation}
	where the supremum is over functions $F: \pointsets_{\region} \to \R$ that are at most $1$-Lipschitz in the $\Delta$-metric.  
	For a version of Kantorovich-Rubinstein duality in this level of generality, see \cite[Theorem 5.10]{villani2009optimal}\footnote{To apply \cite[Theorem 5.10]{villani2009optimal}, we note that $\pointsets_{\Lambda}$ is a Borel subset of the Polish space of finite counting measures on the compact set $\region$; the cost function $\Delta(X,Y)$ is measurable and lower semicontinuous in the vague topology.}.

	To this end, we will use the following general version of a path coupling lemma.
	
	\begin{lemma} \label{lemma:path_coupling}
		Let $\region \subset \R^d$ be compact. 
		Consider a Markov kernel $P$ on $(\pointsets_{\Lambda},\events_{\region})$ with stationary distribution $\pi$.  Suppose there is $\delta > 0$ so that the following holds: for each pair $X,Y \in \pointsets_{\Lambda}$ with $\Delta(X,Y) = 1$ there is a coupling $(\bX,\bY)$ of $P(X,\cdot)$ and $P(Y,\cdot)$ so that $\E[\Delta(\bX,\bY)] \leq 1 - \delta\,. $
		Then, for all $X \in \pointsets_{\Lambda}$ and $t \in \mathbb{N}_0$ it holds that 
		\[
		\|P^t(X, \cdot) - \pi\|_{\mathrm{TV}} \le e^{- \delta t} \cdot \E_{\bY \sim \pi} [\Delta(X, \bY)] \,.
		\]
	\end{lemma}
	\begin{proof}
		For each fixed $X,Y \in \pointsets_{\Lambda}$, we wish to find a coupling of $(\bX,\bY)$ so that $$\E[\Delta(\bX,\bY)] \leq (1 - \delta) \cdot \Delta(X,Y)\,.$$
		By \eqref{eq:wasserstein-duality}, the triangle inequality implies \begin{align*}
			W_\Delta(\delta_XP,\delta_Y P) \leq \sum_{j = 0}^{K - 1} W_\Delta(\delta_{Z_j} P,\delta_{Z_{j+1}}P) \leq  (1 - \delta) K = (1 - \delta) \Delta(X, Y)
		\end{align*}
		for a sequence of configurations $X=Z_0, \dots, Z_K = Y$ with $K=\Delta(X, Y)$ and $\Delta(Z_j, Z_{j+1}) = 1$, where, in the last inequality, we used the contractive assumption and the definition of $W_\Delta$ in \eqref{eq:wasserstein-definition}.    
		
		In particular, for any measures $\mu_1,\mu_2 \in L^1$, we may apply the duality statement \eqref{eq:wasserstein-duality} to see 
		\begin{align*}
			W_{\Delta}(\mu_1 P,\mu_2 P) 
			&= \sup_{F : \mathrm{Lip}(F) \leq 1} | \mu_1 P F - \mu_2 P F| 
			\le  \sup_{F : \mathrm{Lip}(F) \leq 1} \inf_{(\bX, \bY)} \E |\delta_{\bX} P F - \delta_{\bY} P F| \\ 
			&\le (1-\delta) \inf_{(\bX, \bY)} \E \Delta(\bX, \bY) = (1 - \delta) W_{\Delta}(\mu_1, \mu_2) \,,
		\end{align*}
		where the infimum is over couplings $(\bX,\bY)$ of $(\mu_1,\mu_2)$.

		By stationarity of $\pi$ with respect to $P$, this shows that $$W_\Delta(\delta_X P^t,\pi) \leq (1 - \delta)^t W_{\Delta}(\delta_X,\pi) \leq e^{-\delta t} \E_{\bY \sim \pi}[\Delta(X,\bY)]$$ where we assumed $\pi \in L^1$ as otherwise the lemma holds trivially.
		Since $\Delta(X,Y) \geq \one\{X \neq Y\}$ we have $\|P^t(X,\cdot) - \pi\|_{\mathrm{TV}} \leq W_\Delta(\delta_X P^t,\pi)$ by \eqref{eq:TV-coupling}  and \eqref{eq:wasserstein-definition} which completes the proof.   
	\end{proof}
	
	Lemma \ref{lemma:path_coupling} is particularly convenient to apply if the diameter of $\Omega$ with respect to $\Delta$ is bounded.
	However, this will not necessarily be the case in our application.
	Instead, we argue that after some short ``burn-in'' period, the block dynamics from any initial configuration with sufficiently high probability reach a state within some bounded (expected) distance from a random state drawn from the stationary distribution.
	This is done by combining the triangle inequality with the following lemma.
	\begin{lemma}\label{lemma:burn_in}
		Let $(\bX_t)_{t \in \N_0}$ be a block dynamics chain with some activity function $\activities \le \activity$, repulsive potential $\potential$ and update radius $L > 0$ on a box $\region = [-n, n]^d$.
		There are $t_0 \in \Theta(\vol(\region) \cdot \log(\vol(\region)/\varepsilon))$ and $s \in \Theta(\vol(\region) \cdot \log(1/\varepsilon))$ such that for all $t \ge t_0$ it holds that $\P[\size{\bX_t} \ge s] \le \varepsilon$, where the implicit constants only depend on $L$, $\activity$ and $d$.
	\end{lemma}
	\begin{proof}
		Let $T \in \mathbb{N}$ be fixed.  First sample $y_1,\ldots,y_{T}$ which will be the centers of the block updates.  Define $B_j = B_L(y_j)$ and $S = \bigcup_{j = 1}^{T} B_j$.  Since $\phi$ is repulsive, we may use Poisson domination (Lemma \ref{lem:Poisson-domination}) and induct on $T$ to note that $\bX_T \cap S$ is stochastically dominated by $\bY \cap S$ where $\bY$ is a Poisson point process of intensity $\lambda$.  Then we see for each $T$ and $s > 0$ that \begin{align*}
			\P(|\bX_T| \geq s) \leq \P(\Lambda \not\subset S) + \P(|\bY \cap S| \geq s)\,.
		\end{align*}
		
		Note that by standard bounds on the coupon collector problem, there is a constant $C = C_{L,\lambda,d} > 0$ so that for $T \geq C \vol(\Lambda) \cdot \log(\vol(\Lambda)/\eps)$ we have $\P(\Lambda \not\subset S)  \leq \eps/2$.  Similarly, by a standard tail bound on Poisson random variables, we can find  $C' = C_{L,\lambda,d}' > 0$ so that for $s \geq C'\vol(\Lambda) \cdot \log(1/\eps)$ we have $\P(|\bY \cap S| \geq s) \leq \eps/2$, completing the proof.
	\end{proof}
	
	Using Lemma \ref{lemma:path_coupling} and Lemma \ref{lemma:burn_in}, we prove the following bound on the mixing time of block dynamics.
	\begin{proof}[Proof of Theorem \ref{th:mixing-block-dynamics}]
		For $L$ to be chosen later, let $P$ be the transition kernel of the block dynamics on some box $\region = [-n, n]^d$, for $n$ large enough as a function of $d$, $C$ and $m$.
		Our goal is to apply Lemma \ref{lemma:path_coupling} to $P$.
		Thus, it suffices to exhibit a contractive coupling of $P(X, \cdot)$ and $P(Y, \cdot)$ for $X, Y \in \pointsets_{\region}$ with $\Delta(X,Y) = 1$.  It is thus sufficient to show that $W_\Delta(P(X,\cdot), P(Y,\cdot)) \leq 1 - \delta/\vol(\Lambda)$. 
		
		For block dynamics, recall that we choose a uniform random point $\randpoint \in \Lambda$.  
		For every fixed $y \in \region$, let $P^y$ denote the transition kernel of block dynamics given $\randpoint = y$. 
		Note that by the triangle inequality, along with \eqref{eq:wasserstein-duality}, we have
		\begin{equation}\label{eq:wasserstein-triangle-inequality}
			W_\Delta(P(X,\cdot),P(Y,\cdot)) \leq \E_{\randpoint} W_\Delta(P^{\randpoint}(X,\cdot),P^{\randpoint}(Y,\cdot))\,,
		\end{equation}
		where measurability of $y \mapsto W_\Delta(P^{\randpoint}(X,\cdot),P^{\randpoint}(Y,\cdot))$ follows from the measurability of $y \mapsto P^{y}(X,\cdot)$, $y \mapsto P^{y}(Y,\cdot)$ and $(\mu_1, \mu_2) \mapsto W_{\Delta}(\mu_1, \mu_2)$.
		
		In order to upper bound the right-hand side of \eqref{eq:wasserstein-triangle-inequality}, we will specify couplings for each fixed triple $(X,Y,y)$.  
		Since $\Delta(X,Y) = 1,$ assume without loss of generality that $Y = X \cup \{x\}$ for some $x \in \Lambda \setminus X.$  
		
		For $0 < \ell < L$ to be determined later, we construct a coupling of $P^{y}(X, \cdot)$ and $P^{y}(Y, \cdot)$ as follows:
		\begin{enumerate}[(a)] 
			\item \label{it:PC-case-1}If $\dist(x, y) \le L$, draw $\bZ \sim \mu_{B_L(y) \cap \region, \activities_{X \setminus B_L(y)}}$, and set 
			\[
			\bX = (X \setminus B_L(y)) \cup \bZ \text{ and }\bY = (Y \setminus B_L(y)) \cup \bZ.
			\]
			\item \label{it:PC-case-2} If $\dist(x, y) \ge L + \ell$, draw $(\bZ_1, \bZ_2)$ from a coupling of $\mu_{B_L(y) \cap \region, \activities_{X \setminus B_L(y)}}$ and $\mu_{B_L(y) \cap \region, \activities_{Y \setminus B_L(y)}}$ that minimizes the probability that $\bZ_1 \neq \bZ_2$, and set 
			\[
			\bX = (X \setminus B_L(y)) \cup \bZ_1 \text{ and } \bY = (Y \setminus B_L(y)) \cup \bZ_2.
			\]
			\item \label{it:PC-case-3} If $L < \dist(x, y) < L+\ell$, draw $(\bZ_1, \bZ_2)$ from a coupling of $\mu_{B_L(y) \cap \region, \activities_{X \setminus B_L(y)}}$ and $\mu_{B_L(y) \cap \region, \activities_{Y \setminus B_L(y)}}$ that minimizes the probability that $\bZ_1 \setminus B_{\ell}(x) \neq \bZ_2 \setminus B_{\ell}(x)$, draw 
			\begin{align*}
				\bZ'_1 &\sim \mu_{B_L(y) \cap B_{\ell}(x) \cap \region, \activities_{(X \setminus B_L(y)) \cup (\bZ_1 \setminus B_{\ell}(x))}} \text{ and }\\
				\bZ'_2 &\sim \mu_{B_L(y) \cap B_{\ell}(x) \cap \region, \activities_{(Y \setminus B_L(y)) \cup (\bZ_2 \setminus B_{\ell}(x))}} 
			\end{align*}
			independently, and set
			\[
			\bX = (X \setminus B_L(y)) \cup (\bZ_1 \setminus B_{\ell}(x)) \cup \bZ'_1 \text{ and } \bY = (Y \setminus B_L(y)) \cup (\bZ_2 \setminus B_{\ell}(x)) \cup \bZ'_2.
			\]
		\end{enumerate}
		It is easily checked that in each of the three cases the above defines a coupling $(\bX,\bY)$ of $P^{y}(X,\cdot),P^{y}(Y,\cdot).$  
		As such, using \eqref{eq:wasserstein-definition}, we have the bound 
		\begin{equation} \label{eq:mix-over-y-to-Delta}
			W_\Delta(P^{y}(X,\cdot),P^{y}(Y,\cdot)) \leq \E[\Delta(\bX,\bY)] \,.
		\end{equation}
		We proceed by bounding the right-hand side of \eqref{eq:mix-over-y-to-Delta} for every choice of $(X, Y, y)$ with $\Delta(X, Y) = 1$ along the same case distinction used for constructing the coupling.
		
		\noindent    \underline{Case \ref{it:PC-case-1}: $\dist(y,x) \leq L$}. 
		If $\dist(y, x) \le L$ then $\Delta(\bX, \bY) = 0$ since $X$ and $Y$ only differ on $B_L(y)$, and therefore
		\begin{align}
			\E[\Delta(\bX, \bY)] = 0. \label{eq:coupling_close}
		\end{align}

		\noindent \underline{Case \ref{it:PC-case-2}: $\dist(y,x) \geq L + \ell$}.
		Our goal is to show that
		\begin{align}
			\E[\Delta(\bX, \bY)] \le 1 + q(L) \cdot e^{-m \cdot (\dist(x, y) - L) / 2 } \label{eq:coupling_far}
		\end{align}
		for some polynomial $q$ that only depends on $\activity$, $C$, $m$ and $d$.
		We first note that, for $\dist(x, y) \ge L + \ell$, $\Delta(\bX, \bY) = 1 + \Delta(\bZ_1, \bZ_2)$.
		We then use that $\Delta(\bZ_1, \bZ_2) \le \ind{\bZ_1 \neq \bZ_2} \cdot (\size{\bZ_1} + \size{\bZ_2})$ to bound
		\begin{align*}
			\E[\Delta(\bZ_1, \bZ_2)]
			\le \E[\ind{\bZ_1 \neq \bZ_2} \cdot \size{\bZ_1}] + \E[\ind{\bZ_1 \neq \bZ_2} \cdot \size{\bZ_2}] \,.
		\end{align*}
		Using the Cauchy--Schwarz inequality, we further get
		\[
		\E[\ind{\bZ_1 \neq \bZ_2} \cdot \size{\bZ_1}] \le \sqrt{\P(\bZ_1 \neq \bZ_2)} \cdot \sqrt{\E[\size{\bZ_1}^2 ]} \,.
		\]
		By Poisson domination, we  have $\sqrt{\E[\size{\bZ_1}^2]} \le q_1(L)$ for some polynomial $q_1$.
		Applying the same argument to $\bZ_2$, we get
		\[
		\E[\Delta(\bZ_1, \bZ_2)] \le 2 \cdot q_1(L) \cdot \sqrt{\P(\bZ_1 \neq \bZ_2)} \,.
		\]
		For $\dist(y, x) \ge L + \ell$, it holds that $X \setminus B_L(y)$ and $Y \setminus B_L(y)$ only differ in $x$, and thus, by single-site SSM (\Cref{def:spatial_mixing}),
		\[
		\sqrt{\P(\bZ_1 \neq \bZ_2)}
		\le \sqrt{C \cdot \vol(B_L(y) \cap \region) \cdot e^{-m \cdot (\dist(x, y) - L)}} \le q_2(L) \cdot e^{-m \cdot (\dist(x, y) - L)/2}  \,
		\] 
		for some polynomial $q_2$, which proves \eqref{eq:coupling_far}.
		
		\noindent \underline{Case \ref{it:PC-case-3}: $L < \dist(y,x) < L + \ell$}. 
		We will show that
		\begin{align}
			\E[\Delta(\bX, \bY)] \le 1 + p(L) \cdot e^{-m \cdot \ell/2} + r(\ell) \label{eq:coupling_mid}
		\end{align}
		for polynomials $p$ and $r$ only depending on $\activity$, $C$, $m$ and $d$.
		We first note that, if $L < \dist(x, y) < L + \ell$, then 
		\[
		\Delta(\bX, \bY) = 1 + \Delta(\bZ_1 \setminus B_{\ell}(x), \bZ_2 \setminus B_{\ell}(x)) + \Delta(\bZ'_1, \bZ'_2) .
		\]
		Using Poisson domination (Lemma \ref{lem:Poisson-domination}), the expectation of the latter term is bounded by
		\begin{align*}
			\E[\Delta(\bZ'_1, \bZ'_2)]
			\le \E[\size{\bZ'_1} ] + \E[\size{\bZ'_2} ]
			\le r(\ell) , 
		\end{align*}
		for a polynomial $r$.
		For the expectation of the remaining term, using similar computations as in the case $\dist(x, y) > L + \ell$, we obtain
		\begin{align*}
			\E[\Delta(\bZ_1 \setminus B_{\ell}(x), \bZ_2 \setminus B_{\ell}(x))] 
			&\le \sqrt{\P(\bZ_1 \setminus B_{\ell}(x) \neq \bZ_2 \setminus B_{\ell}(x))} \cdot \left(\sqrt{\E[\size{\bZ_1}^2]} + \sqrt{\E[\size{\bZ_2}^2]} \right) \\
			&\le p_1(L) \cdot \sqrt{\P(\bZ_1 \setminus B_{\ell}(x) \neq \bZ_2 \setminus B_{\ell}(x))}
		\end{align*}
		for some polynomial $p_1$.
		Using single-site SSM, we have
		\[
		\sqrt{\P(\bZ_1 \setminus B_{\ell}(x) \neq \bZ_2 \setminus B_{\ell}(x))} \le  p_2(L) \cdot e^{-m \cdot \ell/2} 
		\]
		for some polynomial $p_2$, proving  \eqref{eq:coupling_mid}.
		
		Combining \eqref{eq:wasserstein-triangle-inequality} and \eqref{eq:mix-over-y-to-Delta} with the pointwise bounds in  \eqref{eq:coupling_close}, \eqref{eq:coupling_far} and \eqref{eq:coupling_mid}, we get, for $\randpoint$ uniformly random from $\region$,
		\begin{align*}
			W_{\Delta}(P(X, \cdot), P(Y, \cdot)) 
			&\le \P(L < \dist(x, \randpoint) \le L + \ell) \cdot \left(1 + r(L) + p(L) \cdot e^{-m \cdot \ell/2} \right) \\
			&\hspace{2em} + \E\left[\ind{\dist(x, \randpoint) > L + \ell} \cdot \left(1 + q(L) \cdot e^{-m \cdot (\dist(x, \randpoint) - L) / 2}\right)\right] \\
			&= \P(\dist(x, \randpoint) > L) + \P(L < \dist(x, \randpoint) \le L + \ell) \cdot \left(r(L) + p(L) \cdot e^{-m \cdot \ell/2} \right) \\
			&\hspace{2em} + q(L) \cdot \E\left[\ind{\dist(x, \randpoint) > L + \ell} \cdot e^{-m \cdot (\dist(x, \randpoint) - L) / 2} \right] \\
			&\le 1 - \frac{c \cdot L^d}{\vol(\region)} + \frac{(\hat{q}(L) + \hat{p}(L)) \cdot e^{-m\cdot \ell /2} + L^{d-1} \cdot \hat{r}(\ell)}{\vol(\region)} 
		\end{align*}
		for some $c > 0$ depending on $d$ and polynomials $\hat{q}, \hat{p}$ and $\hat{r}$ depending on $\activity$, $C$, $m$ and $d$.
		Choosing $\ell \in \Theta(\log(L))$ and $L$ large enough such that
		\[
		c \cdot L^d > (\hat{q}(L) + \hat{p}(L)) \cdot e^{-m\cdot \ell /2} + L^{d-1} \cdot \hat{r}(\ell),
		\]
		we then get $W_{\Delta}(P(X, \cdot), P(Y, \cdot)) \le 1 - \frac{\delta}{\vol(\region)}$ for some $\delta > 0$ and all $X, Y \in \pointsets_{\region}$ with $\Delta(X, Y) = 1$.
		
		Applying Lemma \ref{lemma:path_coupling} and recalling the definition of the Wasserstein distance in \eqref{eq:wasserstein-definition} yields that 
		\begin{align}
			\absolute{P^t(X, \cdot) - \mu_{\region, \activities}}_{tv} \le e^{-\delta \cdot t/\vol(\region)} \cdot \E_{\bY \sim \mu_{\region, \activities}}[\Delta(X, \bY)] \le e^{-\delta \cdot t/\vol(\region)} \cdot (\activity \cdot \vol(\region) + \size{X}) \label{eq:mixing}
		\end{align}
		for all $X \in \pointsets_{\region}$.
		In particular, the desired mixing time bound follows if we restrict to initial configurations $X$ with size bounded by some fixed polynomial in $\vol(\region)$.
		
		To extend the result to arbitrary starting configurations, suppose that $(\bX_t)_{t \in \N_0}$ is a version of the block dynamics for a sufficiently large update radius and any initial configuration $\bX_0 = X$.
		By Lemma \ref{lemma:burn_in} we can choose $t_0 \in \Theta(\vol(\region) \cdot \log(\vol(\region)/\varepsilon))$ and $s \in \Theta(\vol(\region) \cdot \log(1/\varepsilon))$ such that $\Pr[\size{\bX_{t_0}} \ge s ] \le \varepsilon/2$.
		Further, choose $t_1 \ge \frac{\vol(\region)}{\delta} \cdot \log\left(\frac{\activity \cdot \vol(\region) + s}{\varepsilon/2}\right) \in \Theta(\vol(\region) \cdot \log(\vol(\region)/\varepsilon))$.
		By \eqref{eq:mixing}, we can couple $\bX_{t_0 + t_1}$ with some $\bZ \sim \mu_{\region, \activities}$ such that $\P[\bX_{t_0 + t_1} \neq \bZ \mid \size{\bX_{t_0}} \le s] \le \varepsilon/2$.
		Using this coupling, we have
		\[
		\P[\bX_{t_0 + t_1} \neq \bZ] \le \P[\size{\bX_{t_0}} > s] + \P[\bX_{t_0 + t_1} \neq \bZ \mid \size{\bX_{t_0}} \le s] \le \varepsilon,
		\]
		and thus $\absolute{P^t(X, \cdot) - \mu_{\region, \activities}}_{tv} \le \varepsilon$ for $t \ge t_0 + t_1 \in \Theta(\vol(\region) \cdot \log(\vol(\region)/\varepsilon))$.
	\end{proof}

	\section{The canonical ensemble}\label{sec:canonical-ensemble}
	In this section we prove Theorem \ref{thm:analyticity-canonical} and Corollaries~\ref{cor:canon-hs} and~\ref{cor:canon-pd}. 
	Recall that for $\alpha > 0$ we define   
	$$\psi(\alpha) := \lim_{n \to \infty} \frac{1}{\Vol(\Lambda_n)}\log \widehat{Z}_{\Lambda_n}(\lfloor \alpha \Vol(\Lambda_n)\rfloor ) 
	$$
	where 
	$$\widehat{Z}_{\Lambda}(k) = \frac{1}{k!} \int_{\Lambda^{k}} e^{-H(\mathbf{x})}\,\diff \mathbf{x}\,.$$
	
	Our main technical result is to show the limit $\psi$ is given in terms of the grand canonical pressure $p$ (Proposition \ref{prop:psi-exists}). We note that equivalence of ensembles has been established in various settings.  For instance see \cite[Theorem~3.4]{ruelle1969statistical} for various statements about equivalence of ensembles under various differentiability assumptions. 
	
	Recall that for $\lambda>0$ we define the limiting density
	$$\alpha(\lambda)= \lim_{n \to \infty} \alpha_{\Lambda_n}(\lambda)
	=\lim_{n\to\infty}
	\frac{\E_{\bX \sim \Lambda_n,\lambda}[|\bX|]}{\Vol(\Lambda_n)}\, .
	$$
	We will show that for $\lambda\in (0,\lambda_\spec)$, this limit exists and moreover, $\alpha(\lambda)$ is strictly increasing on this interval. We may therefore define the limit
	\[
	\alpha_{\spec}
	\coloneqq
	\lim_{\lambda\uparrow\lambda_{\spec}}\alpha(\lambda)\, .
	\]

	\begin{proposition}\label{prop:psi-exists}
		Let $\phi$ be a repulsive potential that decays exponentially (see Definition \ref{def:potential_decay})   Then for all $\lambda < \lambda_\spec$, the limit defining $\alpha(\lambda)$ exists and $\alpha$ is nondecreasing on the interval $(0,\lambda_\spec)$. Moreover, we have the identity
		$$\psi(\alpha(\lambda)) = p(\lambda) - \alpha(\lambda) \log \lambda.$$  
	\end{proposition}

	We start by showing a basic bound of approximate flatness of the canonical partition functions $\widehat{Z}_{\Lambda}(k)$. \begin{lemma} \label{lem:approximately-smooth}
		For a repulsive potential $\phi$ and all $k,\ell \geq 0$ we have \begin{equation*}
			\widehat{Z}_\Lambda(k + \ell) \leq \Vol(\Lambda)^\ell \frac{k!}{(k+\ell)!} \widehat{Z}_{\Lambda}(k)\,.
		\end{equation*}
	\end{lemma}
	\begin{proof}
		Note that \begin{align*}
			\widehat{Z}_\Lambda(k + \ell) &= \frac{k!}{(k+\ell)!} \cdot \frac{1}{k!} \int_{\Lambda^{k+\ell}} e^{-H(x_1,\ldots,x_k,x_{k+1},\ldots,x_{k+\ell})} \,\diff\mathbf{x} \leq   \frac{k!}{(k+\ell)!} \cdot \frac{\Vol(\Lambda)^\ell}{k!} \int_{\Lambda^{k}} e^{-H(x_1,\ldots,x_k)} \,\diff\mathbf{x} \\
			&= \Vol(\Lambda)^\ell \frac{k!}{(k+\ell)!} \widehat{Z}_{\Lambda}(k)\,. \qedhere 
		\end{align*}
	\end{proof}

	We need the fact that the variance for the number of points in a sample  grows at least linearly.  This was proven for repulsive potentials by Ginibre \cite{ginibre1967rigorous}  and more generally by Dereudre--Flimmel \cite{dereudre2024nonuniformity}.

	\begin{lemma}\label{lem:non-hyper-uniform}
		For fixed $\lambda > 0$ and repulsive, tempered potential $\phi$ there is a constant $c = c(\lambda,\phi)$ so that $$\Var_{\mathbf{X} \sim \Lambda,\lambda}(|\mathbf{X}|) \geq c \Vol(\Lambda)\,.$$
	\end{lemma}
	
	We now prove Proposition \ref{prop:psi-exists}.
	\begin{proof}[Proof of Proposition \ref{prop:psi-exists}]
		
		Note that for $\lambda > 0$ and $\Lambda$ that $$\frac{\E_{\bX \sim \Lambda,\lambda} |\bX|}{\Vol(\Lambda)} = \lambda p_{\Lambda}'(\lambda) = \alpha_{\Lambda}(\lambda)\,.$$
		
		By Theorem~\ref{thm:spatial-mixing} and Theorem \ref{thm:zero_freeness}\ref{thm:zero_freeness:log-bound}, we have that for $\Lambda_n = [-n,n]^d$  and each $\lambda_0 \in [0,\lambda_\spec)$, the sequence of analytic functions $\lambda \mapsto \alpha_{\Lambda_n}(\lambda)$ converge uniformly on a ball containing $\lambda_0$ in the complex plane.  In particular, this implies that we may swap limits and derivatives \begin{equation}\label{eq:finite-volume-density}
			\lim_{n \to \infty} \frac{\E_{\bX \sim \Lambda_n,\lambda} |\bX|}{\Vol(\Lambda_n)} = \lambda p'(\lambda) = \alpha(\lambda) \,.
		\end{equation}
		
		We note that $\alpha(\lambda)$ is strictly increasing since for each $\Lambda$ we have 
		\begin{align}\label{eq:finitep-convex}
			\frac{\diff^2}{\diff t^2}( p_{\Lambda}(\lambda e^t)) \big|_{t = 0} = \frac{\Var_{X \sim \Lambda,\lambda}(|\bX|)}{\vol(\Lambda)} \geq c > 0
		\end{align}
		where the latter is by Lemma \ref{lem:non-hyper-uniform}.

		We now turn to proving the identity in the statement of the proposition. 
		Fix $\lambda\in (0,\lambda_\spec)$ and choose $\lambda_n$ approaching $\lambda$.  First  
		let $\bX_n$ be a sample from $\mu_{\Lambda_n,\lambda_n}$ and set $X_n = |\bX_n|$. Define $m_n = \E X_n$ and $\sigma_n^2 = \Var(X_n)$.   Choose $\lambda_n$ so that $m_n = \lfloor \alpha(\lambda) \Vol(\Lambda_n) \rfloor$ and note that by \eqref{eq:finite-volume-density} we have $\lambda_n \to \lambda$ as $n \to \infty$.  Further, we note that $\lambda_n$ is unique due to Lemma \ref{lem:non-hyper-uniform}.  
		
		Using the uniform bound on the complex logarithm of the partition function guaranteed by Theorem \ref{thm:zero_freeness}\ref{thm:zero_freeness:log-bound}, a central limit theorem of Michelen-Sahasrabudhe \cite[Theorem 1.8]{michelen2025central} shows that, e.g.\, \begin{equation} \label{eq:CLT}
			\lim_{n \to \infty} \P\left( \frac{X_n - m_n}{\sigma_n} \in (0,1) \right) = \int_0^1 \frac{1}{\sqrt{2\pi}} e^{-x^2 / 2}\,\diff x > \frac{1}{3}\,.
		\end{equation}
		Lemma \ref{lem:non-hyper-uniform} shows that $\sigma_n = \Omega(\Vol(\Lambda_n)^{1/2})$ while Cauchy's integral formula along with Theorem \ref{thm:zero_freeness}\ref{thm:zero_freeness:log-bound} shows \begin{equation}\label{eq:variance-upper-bound}
			\frac{\sigma_n^2}{\Vol(\Lambda_n)} = \frac{\diff^2}{\diff t^2} \log Z_{\Lambda_n} (\lambda_n e^{t}) \bigg|_{t = 0} \leq C
		\end{equation}
		for some constant $C$ depending only on $M$ and $\delta$ from Theorem \ref{thm:zero_freeness}.  For $n$ sufficiently large we have \begin{align*}
			\frac{1}{3} \leq \P\left(\frac{X_n - m_n}{\sigma_n} \in (0,1) \right) \leq \frac{\lambda_n^{m_n}}{Z_{\Lambda_n}(\lambda_n)} \sum_{\ell = 0}^{\sigma_n} \lambda_n^\ell \widehat{Z}_{\Lambda_n}(m_n + \ell) \leq \left( \frac{\lambda_n^{m_n} \widehat{Z}_{\Lambda_n}(m_n)}{Z_{\Lambda_n} (\lambda_n)} \right) \cdot \left( \sum_{\ell = 0}^{\sigma_n} \frac{\Vol(\Lambda_n)^{\ell}}{(m_n + \ell)_\ell} \right)
		\end{align*}
		where $(a)_b=a!/(a-b)!$ denotes the falling factorial and
		in the first inequality we used \eqref{eq:CLT} and in the latter we used Lemma \ref{lem:approximately-smooth}.  If we write $V_n = \Vol(\Lambda_n)$ then note that \begin{equation*}
			\sum_{\ell = 0}^{\sigma_n} \frac{V_n^\ell}{(m_n + \ell)_\ell} \leq \sum_{\ell = 0}^{\sigma_n} \left(\frac{V_n}{m_n} \right)^{\ell}  \leq \sum_{\ell = 0}^{\sigma_n} \left(\alpha(\lambda_\ast)^{-1} + o(1) \right)^{\ell} \leq \exp(O(\sqrt{V_n}))
		\end{equation*}
		where the latter bound is using \eqref{eq:variance-upper-bound}.  Since we have the trivial inequality $\frac{\lambda_n^{m_n} \widehat{Z}_{\Lambda_n}(m_n)}{Z_{\Lambda} (\lambda_n)} \leq 1$ we see that \begin{equation*}
			\lim_{n \to \infty} \frac{1}{V_n}\left| \log \widehat{Z}_{\Lambda_n}(m_n) + m_n \log \lambda_n - \log Z_{\Lambda_n}(\lambda_n) \right| = 0
		\end{equation*}    
		completing the proof.
	\end{proof}

	\begin{proof}[Proof of Theorem \ref{thm:analyticity-canonical}]
		For $\alpha_0 \in (0,\alpha_\spec)$, let $\lambda_0 \in (0,\lambda_\spec)$ be the unique $\lambda_0$ so that $\alpha(\lambda_0) = \alpha_0$.  By~\eqref{eq:finite-volume-density}, Theorem~\ref{thm:spatial-mixing} and Theorem \ref{thm:zero_freeness}\ref{thm:zero_freeness:analyticity} we have that $\lambda \mapsto \alpha(\lambda) = \lambda p'(\lambda)$ in analytic in a neighborhood of $\lambda_0$.  By ~\eqref{eq:finitep-convex} and Lemma \ref{lem:non-hyper-uniform} we have that $\alpha'(\lambda_0) \neq 0$.  The analytic inverse function theorem thus guarantees an open set $U \subset \C$ with $\alpha_0 \in U$ on which we have $\alpha \mapsto \lambda(\alpha)$ is an analytic inverse of the function $\lambda \mapsto \alpha(\lambda)$.  Using Proposition \ref{prop:psi-exists} completes the proof.
	\end{proof}
	
	Corollary~\ref{cor:canon-hs} follows immediately.
	
	\begin{proof}[Proof of Corollary~\ref{cor:canon-hs}]
		For the case of the hard-sphere potential, Proposition~\ref{prop:hard-spheres} shows that $\lambda_\spec>3^{-d/2}$ for $d$ sufficiently large and so combining Theorem \ref{thm:analyticity-canonical} with \eqref{eq:LB-density} shows that $\psi$ is analytic up to density $c d 2^{-d}$ for each fixed $c < \log(2/\sqrt{3})$ for $d$ sufficiently large.
	\end{proof}
	
	Corollary~\ref{cor:canon-pd} takes a little more work. By Theorem \ref{thm:analyticity-canonical}, our goal is to show that for potentials with $\lambda_\spec=+\infty$, the density $\alpha(\lambda)$ tends to $\infty$ as $\lambda\to\infty$. To this end we need two short lemmas.
	
	The first lemma says that if $\phi$ satisfies the conditions Corollary~\ref{cor:canon-pd}, then $\phi(x,0)$ is uniformly bounded as soon as $x$ is bounded away from $0$. In particular, $\phi$ cannot have a hard core.   
	
	\begin{lemma}\label{lem:impossible-core}
		Let $\phi$ be a repulsive, translation invariant potential that decays exponentially. If $\lambda_\spec=+\infty$ then for all $r>0$,
		\[
		\esssup_{\|x\|\ge r}\phi(x,0)<\infty\, .
		\]
	\end{lemma}
	\begin{proof}
		Let $g(x)=1-e^{-\phi(x,0)}$.
		Since $\lambda_\spec=+\infty$, $g$ is positive definite and measurable and therefore continuous (modifying $g$ on a null set if necessary, see e.g. \cite[Theorem 3.10.20, page 221]{bogachev2007measure}). Thus Bochner's theorem implies that there is a measure $\mu$ so that $$g(x) = \int_{\R^d} e^{i \langle x, \xi\rangle} \diff\mu(\xi)\,.$$ Since $g(0) = 1$, we in fact have that $\mu$ is a probability measure. Note that
		\[
		g(h)\leq g(0)=1
		\]
		for all $h\in \R^d$. If $g(h)=1$ then we must have that $e^{-i\langle h, \xi \rangle}  = 1$ for $\mu$-almost every $\xi$.  We then see that for all $x \in \R^d$  we have $$g(x) = \int_{\R^d} e^{i\langle x, \xi\rangle} \diff\mu(\xi) =  \int_{\R^d} e^{i\langle x - h, \xi\rangle} e^{i\langle h, \xi\rangle}\diff\mu(\xi) =  \int_{\R^d} e^{i\langle x - h, \xi\rangle}\diff\mu(\xi) = g(x - h)\, .$$
		Since $g(x)\to0$ as $\norm{x}\to\infty$ (since $\phi$ decays exponentially) we conclude that $h=0$. Since $g$ is continuous and $g(x)\to0$ as $\norm{x}\to\infty$ we have
		\[
		\esssup_{\|x\|\ge r}g(x)<1\, .
		\]
		The result follows. 
	\end{proof}
	
	We now show that the canonical pressure cannot equal $-\infty$ when $\lambda_{\spec} = +\infty$.  Note that in the case when $\phi$ has a hard core, one in fact does have $\psi(\alpha) = -\infty$ for all $\alpha$ sufficiently large.
	
	\begin{lemma}\label{lem:psi-finite}
		Let $\phi$ be a repulsive, translation invariant potential that decays exponentially. If $\lambda_\spec=+\infty$ then $\psi(\alpha)>-\infty$ for all $\alpha\in (0,\infty)$.
	\end{lemma}
	\begin{proof}
		Fix a density $\alpha>0$.  Choose a lattice $\mathcal L\subset\R^d$ of
		density $\eta>\alpha$ i.e.\ so that 
		\[
		M_n:=|\mathcal L \cap \Lambda_n|= (\eta +o(1)) \vol(\Lambda_n)\, .
		\]
		Let $s>0$ be the minimum distance between two distinct elements in $\mathcal L$. Let $\rho=s/4$ so that the balls $B(z,\rho)$, $z\in\mathcal L$, are disjoint and have mutual distance at least $s/2$. 
		Let
		\[
		k_n=\lfloor \alpha\vol(\Lambda_n)\rfloor 
		\]
		and note that $M_n\ge k_n$ for large $n$ since $\eta>\alpha$.
		
		Suppose now that $X\subseteq \Lambda_n$ is a configuration of size $k_n$ such that $|X\cap B(z,\rho)|=1$ for exactly $k_n$ elements $z\in \mathcal L$. Note that $\norm{x-y}\geq s/2$ for all distinct $x,y\in X$. 
		By Lemma~\ref{lem:impossible-core} and the fact the $\phi$ decays exponentially we have 
		\[
		H(X)\leq C k_n\, ,
		\]
		for some constant $C$ depending only on $s,\phi$.
		Consequently,
		\[
		\widehat Z_{\Lambda_n}(k_n)
		\ge
		\binom{M_n}{k_n}
		\vol(B(0,\rho))^{k_n}
		e^{-Ck_n}.
		\]
		Taking logarithms, dividing by $\vol(\Lambda_n)$, and sending $n\to\infty$
		gives a finite lower bound for $\psi(\alpha)$.  
	\end{proof}
	
	\begin{proof}[Proof of Corollary~\ref{cor:canon-pd}]
		Since $\lambda_{\spec}=\infty$, Theorem~\ref{thm:zero_freeness} gives
		analyticity of the infinite-volume pressure $p(\lambda)$ for every
		$\lambda>0$. Thus
		\[
		\alpha(\lambda)
		=
		\lambda p'(\lambda)
		\]
		is analytic for every $\lambda>0$.
		By Theorem~\ref{thm:analyticity-canonical} it suffices to show that 
		$\alpha(\lambda)\to\infty$ as $\lambda\to\infty$.
		
		Let $P(t):=p(e^t)$ and note that $P'(t)=\alpha(e^t)$. Moreover, by~\eqref{eq:finitep-convex} we see that $P$ is a limit of convex functions and is therefore convex.

		Fix $a>0$.  For $k_n=\lfloor a\vol(\Lambda_n)\rfloor$, the
		grand-canonical partition function satisfies
		\[
		Z_{\Lambda_n}(e^t)
		\ge
		e^{t k_n}\widehat Z_{\Lambda_n}(k_n).
		\]
		Taking logarithms, dividing by $\vol(\Lambda_n)$ and sending $n\to\infty$, we obtain
		\begin{align}\label{eq:pressure-lb}
			P(t)
			\ge
			at+\psi(a)\, ,
		\end{align}
		where $\psi(a)>-\infty$
		by Lemma~\ref{lem:psi-finite}.
		
		Suppose, for contradiction, that $\alpha(e^t)=P'(t)$ does not tend to
		$\infty$.  Since $P$ is convex, this would imply that $P'(t)\le A$ for all sufficiently large $t$, for
		some finite $A$.  Consequently,
		\[
		P(t)\le P(t_0)+A(t-t_0)
		\qquad
		\text{for all }t\ge t_0.
		\]
		Choosing $a>A$, this contradicts
		\eqref{eq:pressure-lb} for large $t$.
	\end{proof}

	\section{Computing \texorpdfstring{$\lambda_\spec$}{lambda-spec}} \label{sec:lambda-spec}
	
	We begin with a proof that in the case of a translation invariant potential, one can write $\lambda_\spec$ in terms of the Fourier transform of the associated function $g$:
	
	\begin{lemma}\label{lem:fourier-criterion}
		Suppose that $\phi$ is translation invariant, i.e.\ $\phi(x,y) = h(x - y)$ for some $h$.  Define $g(x) = 1 - \exp(-h(x))$ and assume that $\phi$ is tempered and repulsive.  Then $$\sup_{f \in L^2 : \|f \|_2 = 1} \iint f(x) f(y) (e^{-\phi(x,y)} - 1)\,\diff x\,\diff y = \esssup_{\xi} (-\widehat{g}(\xi))\,.$$
	\end{lemma}
	\begin{proof}
		We will show each inequality.  Let $f \in L^2$ with $\|f\|_2 = 1$.  Since $g \in L^1$ and $f \in L^2$ we may apply Parseval's theorem to see \begin{equation} \label{eq:parseval-application}
			\iint f(x) f(y)  g(x - y) \,dx\,dy =  \int f(x) (g * f)(x) \,dx =  \int \widehat{g}(\xi) |\widehat{f}(\xi)|^2  \,d\xi \geq  \essinf_\xi \widehat{g}(\xi) 
		\end{equation}
		where in the last bound we used that $\|\widehat{f}\|_2 = \|f \|_2 = 1$ by Plancherel's theorem. Multiplying by $-1$ shows the left hand side is at most the right-hand side in the lemma.  For the reverse, for each $\eps > 0$ we want to find $f \in L^2$ with $\|f \|_2 = 1$ so that $$\iint f(x) f(y) (- g(x - y))\,\diff x\,\diff y \geq \esssup_{\xi}(-\widehat{g}(\xi)) - \eps\,.$$
		Let $A \subset \{\xi : \widehat{g}(\xi) \leq \essinf \widehat{g} + \eps\}$ be a finite volume measurable set which we may assume to be symmetric since $g$ is symmetric.  Define $f$ via $\widehat{f}(\xi) = \ind{A} / \Vol(A)^{1/2}$.  Note that by Parseval we have that $\|f\|_2 = 1$.  Apply \eqref{eq:parseval-application} to see $$\iint f(x) f(y)  g(x - y) \,dx\,dy =  \int \widehat{g}(\xi) |\widehat{f}(\xi)|^2  \,d\xi = \frac{1}{\Vol(A)} \int_A \widehat{g}(\xi)\,\diff \xi \leq \essinf_\xi \widehat{g}(\xi) + \eps$$
		completing the proof.
	\end{proof}

	We now deduce Corollary \ref{cor:positive-definite}:
	
	\begin{proof}[Proof of Corollary \ref{cor:positive-definite}]
		This follows from combining Theorem \ref{thm:main} with Lemma \ref{lem:fourier-criterion}.
	\end{proof}

	We will use the criteria Lemma \ref{lem:fourier-criterion} in order to identify $\lambda_\spec$ for the hard sphere potential and to bound $\lambda_\spec$ for the Gaussian core potential.
	
	\subsection{The hard sphere potential} \label{sec:lambda-spec-HS}
	
	Given the hard sphere potential defined by \eqref{eq:hard-sphere-potential} we first note that one can write $\lambda_\spec$ in terms of the fourier transform of the unit ball.  In particular, if we let $h_d$ denote the indicator function of the ball of volume $1$ in $\R^d$ and set $a_d = - \min_{\xi} \widehat{h}_d(\xi)$ then we obtain the identity \begin{equation} \label{eq:hard-sphere-in-terms-unit-ball}
		2^d\lambda_\spec = a_d^{-1}
	\end{equation}
	using Lemma \ref{lem:fourier-criterion}.  We first show that one can write $a_d$ in terms of Bessel functions:
	
	\begin{lemma}\label{lem:unit-ball-fourier}
		Let $h_d$ denote the indicator function of the ball of volume $1$ in $\R^d$ and set $a_d = - \min_{\xi} \widehat{h}_d(\xi)$.  Then writing $\nu = d/2$ we have
		\begin{equation}
			a_d = \Vol(B_{\R^d}(1))^{-1} (2\pi)^{\nu} ( j_{\nu + 1,1})^{-\nu} (-J_{\nu}(j_{\nu + 1,1}))
		\end{equation}
		where $J_\nu$ is the Bessel function of the first kind of order $\nu$ and $j_{\nu+1,1}$ is the first positive zero of $J_{\nu+1}$.  
	\end{lemma}
	
	We then find the asymptotic behavior as $d$ tends to $\infty$:
	
	\begin{lemma}\label{lem:asymptotics-fourier-ball}
		As $d \to \infty$ we have
		\begin{equation*}
			a_d  = (1 + O(d^{-1/3})) 2^{4/3}e^{-1}\sqrt{\pi} (\mathrm{Ai}'(-\alpha_1)) d^{-1/6}\left(\frac{2}{e} \right)^{d/2} \exp(-2^{2/3} \alpha_1 d^{1/3})
		\end{equation*}
		where $\mathrm{Ai}$ is the Airy function and $\alpha_1 = \min\{\alpha  > 0: \mathrm{Ai}(-\alpha) = 0\}$ is the magnitude of its first negative root.
	\end{lemma}
	
	We will prove both Lemma \ref{lem:unit-ball-fourier} and \ref{lem:asymptotics-fourier-ball} in Appendix \ref{sec:bessel}.
	\begin{proof}[Proof of Proposition \ref{prop:hard-spheres}]
		Combining \eqref{eq:hard-sphere-in-terms-unit-ball} with Lemma \ref{lem:asymptotics-fourier-ball} shows \begin{equation}\label{eq:hard-sphere-lambda-spec-asympt}
			2^d \lambda_\spec = (1 + O(d^{-1/3}))  \frac{e}{2^{4/3} \sqrt{\pi} (\mathrm{Ai}'(-\alpha_1))} d^{1/6} \exp(2^{2/3} \alpha_1 d^{1/3}) \left(\frac{e}{2} \right)^{d/2} \,. \qedhere
		\end{equation}
	\end{proof}

	\subsection{Gaussian core model}
	
	Recall that the Gaussian core model is given by the potential \begin{equation}\label{eq:gaussian-core}
		\phi(x,y) = \beta e^{- \left\| x - y \right\|_2^2 / 2}
	\end{equation}
	where $\beta > 0$ is an inverse temperature parameter.  We will be  interested in the regime where $\beta$ is fixed and $d \to \infty$.  In this regime, we observe that the temperedness constant may be computed asymptotically: \begin{equation}\label{eq:temperedness-GCM}
		C_\phi = \int_{\R^d} (1 - \exp\left(-\beta \exp(- \|x\|_2^2 / 2 \right))\,\diff x = (1 + o(1)) \beta \int_{\R^d} e^{-\|x\|_2^2 / 2}\,\diff x = (1 + o(1))\beta (2\pi)^{d/2}\,.\end{equation}
	
	We will show that $\lambda_\spec$ is exponentially large as $d \to \infty$.
	\begin{lemma}\label{lem:gaussian-core-bound}
		For fixed $\beta > 0$, there is a constant $C_\beta > 0$ so that $\lambda_{\spec} \geq  C_\beta 2^d / C_\phi$ for the Gaussian core model \eqref{eq:gaussian-core}.
	\end{lemma}
	\begin{proof}
		Letting $\psi(z) = \beta e^{-\|z\|^2 / 2}$, we may write $g(z) = 1 - e^{-\psi(z)}$ as \begin{align*}
			g(z) = -\sum_{n \geq 1} \frac{(-\beta)^n}{n!} e^{-n \|z\|^2 / 2} 
		\end{align*}
		and so \begin{align*}
			-\widehat{g}(\xi) = (2\pi)^{d/2} \sum_{n \geq 1} \frac{(-\beta)^n}{n!} n^{-d/2} e^{-2\pi^2 \|\xi\|_2^2 / n} \,.
		\end{align*}
		Relabeling $2\pi^2 \|\xi\|_2^2 = x$, we will throw away all odd terms for $n \geq 3$ and bound \begin{align*}
			\sum_{n \geq 1} \frac{(-\beta)^n}{n!} n^{-d/2} e^{-x / n} \leq -\beta e^{-x} + \frac{\beta^2}{2} 2^{-d/2} e^{-x/2} + \sum_{k\geq 2} \frac{\beta^{2k}}{(2k)!} (2k)^{-d/2} e^{-x/2k}\,.
		\end{align*}
		We note that \begin{equation*}
			\max_{x \geq 0}\left( -\beta e^{-x} + \frac{\beta^2}{2} 2^{-d/2} e^{-x/2} \right) = \frac{\beta^3}{16} 2^{-d}
		\end{equation*}
		while \begin{equation*}
			\sum_{k\geq 2} \frac{\beta^{2k}}{(2k)!} (2k)^{-d/2} e^{-x/2k} \leq 2^{-d} \sum_{k \geq 2} \frac{\beta^{2k}}{(2k)!} \leq 2^{-d} \cosh(\beta) \,.
		\end{equation*}
		Recalling $C_\phi \sim \beta (2\pi)^{d/2}$, applying Lemma \ref{lem:fourier-criterion} completes the proof.
	\end{proof}

	\section*{Acknowledgments}
	A.G. is funded by the Postdoc Network Brandenburg.
	M.J.\ is supported by a UK Research and Innovation Future Leaders Fellowship MR/W007320/2.  M.M.\ is supported in part by NSF grants DMS-2336788 and DMS-2246624.
	M.P.\ is funded by the Deutsche Forschungsgemeinschaft (DFG, German Research Foundation) -- project number 390859508.
	W.P.\ supported in part by NSF grant DMS-2348743.

	\bibliography{references.bib}

\begin{thebibliography}{10}

\bibitem{alder1957phase}
B.~J. Alder and T.~E. Wainwright.
\newblock Phase transition for a hard sphere system.
\newblock {\em The Journal of chemical physics}, 27(5):1208, 1957.

\bibitem{anand2023perfect}
K.~Anand, A.~G{\"o}bel, M.~Pappik, and W.~Perkins.
\newblock Perfect sampling for hard spheres from strong spatial mixing.
\newblock In {\em 27th International Conference on Randomization and
  Computation and the 26th International Conference on Approximation Algorithms
  for Combinatorial Optimization Problems}, pages 1--18. Schloss
  Dagstuhl-Leibniz-Zentrum f{\"u}r Informatik, 2023.

\bibitem{bakry2013analysis}
D.~Bakry, I.~Gentil, and M.~Ledoux.
\newblock {\em Analysis and Geometry of {M}arkov Diffusion Operators}.
\newblock Grundlehren der mathematischen Wissenschaften. Springer International
  Publishing, 2013.

\bibitem{balister2005continuum}
P.~Balister, B.~Bollob{\'a}s, and M.~Walters.
\newblock Continuum percolation with steps in the square or the disc.
\newblock {\em Random Structures \& Algorithms}, 26(4):392--403, 2005.

\bibitem{bernard2011two}
E.~P. Bernard and W.~Krauth.
\newblock Two-step melting in two dimensions: First-order liquid-hexatic
  transition.
\newblock {\em Physical review letters}, 107(15):155704, 2011.

\bibitem{bertini2002spectral}
L.~Bertini, N.~Cancrini, and F.~Cesi.
\newblock The spectral gap for a {G}lauber-type dynamics in a continuous gas.
\newblock {\em Annales de l'IHP Probabilit{\'e}s et statistiques},
  38(1):91--108, 2002.

\bibitem{betsch2023uniqueness}
S.~Betsch and G.~Last.
\newblock On the uniqueness of {G}ibbs distributions with a non-negative and
  subcritical pair potential.
\newblock In {\em Annales de l'Institut Henri Poincare (B) Probabilites et
  statistiques}, volume~59, pages 706--725. Institut Henri Poincar{\'e}, 2023.

\bibitem{bogachev2007measure}
V.~I. Bogachev.
\newblock {\em Measure theory}, volume~1.
\newblock Springer, 2007.

\bibitem{bubley1997path}
R.~Bubley and M.~Dyer.
\newblock Path coupling: A technique for proving rapid mixing in {M}arkov
  chains.
\newblock In {\em Proceedings 38th Annual Symposium on Foundations of Computer
  Science}, pages 223--231. IEEE, 1997.

\bibitem{chen2025rapid}
X.~Chen, Z.~Chen, Z.~Chen, Y.~Yin, and X.~Zhang.
\newblock Rapid mixing on random regular graphs beyond uniqueness.
\newblock In {\em 2025 {IEEE} 66th {A}nnual {S}ymposium on {F}oundations of
  {C}omputer {S}cience---{FOCS} 2025}, pages 2170--2193. IEEE Comput. Soc.
  Press, Los Alamitos, CA, 2025.

\bibitem{cohn2018gaussian}
H.~Cohn and M.~de~Courcy-Ireland.
\newblock The {G}aussian core model in high dimensions.
\newblock {\em Duke Math. J.}, 167(13):2417--2455, 2018.

\bibitem{cohn2017sphere}
H.~Cohn, A.~Kumar, S.~D. Miller, D.~Radchenko, and M.~Viazovska.
\newblock The sphere packing problem in dimension 24.
\newblock {\em Ann. of Math. (2)}, 185(3):1017--1033, 2017.

\bibitem{cohn2022universal}
H.~Cohn, A.~Kumar, S.~D. Miller, D.~Radchenko, and M.~Viazovska.
\newblock Universal optimality of the {$E_8$} and {L}eech lattices and
  interpolation formulas.
\newblock {\em Ann. of Math. (2)}, 196(3):983--1082, 2022.

\bibitem{daley2003introduction}
D.~J. Daley and D.~Vere-Jones.
\newblock {\em An introduction to the theory of point processes: volume I:
  elementary theory and methods}.
\newblock Springer, 2003.

\bibitem{daley2008introduction}
D.~J. Daley and D.~Vere-Jones.
\newblock {\em An introduction to the theory of point processes: volume II:
  general theory and structure}.
\newblock Springer, 2008.

\bibitem{dereudre2019introduction}
D.~Dereudre.
\newblock Introduction to the theory of {G}ibbs point processes.
\newblock In {\em Stochastic Geometry: Modern Research Frontiers}, pages
  181--229. Springer, 2019.

\bibitem{dereudre2024nonuniformity}
D.~Dereudre and D.~Flimmel.
\newblock Non-hyperuniformity of {G}ibbs point processes with short-range
  interactions.
\newblock {\em J. Appl. Probab.}, 61(4):1380--1406, 2024.

\bibitem{dobrushin1985completely}
R.~L. Dobrushin and S.~B. Shlosman.
\newblock Completely analytical {G}ibbs fields.
\newblock In {\em Statistical Physics and Dynamical Systems: Rigorous Results},
  pages 371--403. Springer, 1985.

\bibitem{dobrushin1987completely}
R.~L. Dobrushin and S.~B. Shlosman.
\newblock Completely analytical interactions: constructive description.
\newblock {\em Journal of Statistical Physics}, 46(5):983--1014, 1987.

\bibitem{dyer2004mixing}
M.~Dyer, A.~Sinclair, E.~Vigoda, and D.~Weitz.
\newblock Mixing in time and space for lattice spin systems: A combinatorial
  view.
\newblock {\em Random Structures \& Algorithms}, 24(4):461--479, 2004.

\bibitem{finken2001freezing}
R.~Finken, M.~Schmidt, and H.~L{\"o}wen.
\newblock Freezing transition of hard hyperspheres.
\newblock {\em Physical Review E}, 65(1):016108, 2001.

\bibitem{fisher1964free}
M.~E. Fisher.
\newblock The free energy of a macroscopic system.
\newblock {\em Archive for Rational Mechanics and Analysis}, 17(5):377--410,
  1964.

\bibitem{friedli2017statistical}
S.~Friedli and Y.~Velenik.
\newblock {\em Statistical mechanics of lattice systems: a concrete
  mathematical introduction}.
\newblock Cambridge University Press, 2017.

\bibitem{FGKKP2025}
T.~Friedrich, A.~Göbel, M.~Katzmann, M.~S. Krejca, and M.~Pappik.
\newblock Sampling repulsive {G}ibbs point processes using random graphs.
\newblock {\em Combinatorics, Probability and Computing}, 34(1):63–89, 2025.

\bibitem{friedrich2022spectral}
T.~Friedrich, A.~Göbel, M.~S. Krejca, and M.~Pappik.
\newblock A spectral independence view on hard spheres via block dynamics.
\newblock {\em SIAM Journal on Discrete Mathematics}, 36(3):2282--2322, 2022.

\bibitem{frisch1999high}
H.~Frisch and J.~Percus.
\newblock High dimensionality as an organizing device for classical fluids.
\newblock {\em Physical Review E}, 60(3):2942, 1999.

\bibitem{ginibre1967rigorous}
J.~Ginibre.
\newblock Rigorous lower bound on the compressibility of a classical system.
\newblock {\em Physics Letters A}, 24(4):223--224, 1967.

\bibitem{goebel2026simple}
A.~G\"obel, M.~Jenssen, M.~Michelen, M.~Pappik, W.~Perkins, and L.~Schiller.
\newblock A simple proof of rapid mixing on random regular graphs beyond
  uniqueness.
\newblock {\em preprint}, 2026.

\bibitem{grakaos2014classical}
L.~Grafakos.
\newblock {\em Classical {F}ourier analysis}, volume 249 of {\em Graduate Texts
  in Mathematics}.
\newblock Springer, New York, third edition, 2014.

\bibitem{groeneveld1962two}
J.~Groeneveld.
\newblock Two theorems on classical many-particle systems.
\newblock {\em Physics Letters}, 3, 1962.

\bibitem{helmuth2022correlation}
T.~Helmuth, W.~Perkins, and S.~Petti.
\newblock Correlation decay for hard spheres via {M}arkov chains.
\newblock {\em Ann. Appl. Probab.}, 32(3):2063--2082, 2022.

\bibitem{hofertemel2019disagreement}
C.~Hofer-Temmel.
\newblock Disagreement percolation for the hard-sphere model.
\newblock {\em Electron. J. Probab.}, 24:Paper No. 91, 22, 2019.

\bibitem{houdebart2022explicit}
P.~Houdebert and A.~Zass.
\newblock An explicit {D}obrushin uniqueness region for {G}ibbs point processes
  with repulsive interactions.
\newblock {\em J. Appl. Probab.}, 59(2):541--555, 2022.

\bibitem{jansen2018gibbsian}
S.~Jansen.
\newblock Gibbsian point processes.
\newblock {\em available at author’s website}, 2018.

\bibitem{jansen2019cluster}
S.~Jansen.
\newblock Cluster expansions for {G}ibbs point processes.
\newblock {\em Advances in Applied Probability}, 51(4):1129--1178, 2019.

\bibitem{jenssen2019hard}
M.~Jenssen, F.~Joos, and W.~Perkins.
\newblock On the hard sphere model and sphere packings in high dimensions.
\newblock {\em Forum Math. Sigma}, 7:Paper No. e1, 19, 2019.

\bibitem{jenssen2024quasipolynomial}
M.~Jenssen, M.~Michelen, and M.~Ravichandran.
\newblock Quasipolynomial-time algorithms for {G}ibbs point processes.
\newblock {\em Combinatorics, Probability and Computing}, 33(1):1--15, 2024.

\bibitem{jerrum2019perfect}
M.~Jerrum and H.~Guo.
\newblock Perfect simulation of the hard disks model by partial rejection
  sampling.
\newblock {\em Annales de l’Institut Henri Poincar{\'e} D (AIHPD)}, 2019.

\bibitem{kondratiev2013spectral}
Y.~Kondratiev, T.~Kuna, and N.~Ohlerich.
\newblock Spectral gap for {G}lauber type dynamics for a special class of
  potentials.
\newblock {\em Electronic Journal Of Probability}, 18, 2013.

\bibitem{lee1952statistical}
T.-D. Lee and C.-N. Yang.
\newblock Statistical theory of equations of state and phase transitions. ii.
  lattice gas and {I}sing model.
\newblock {\em Physical Review}, 87(3):410, 1952.

\bibitem{levin2017markov}
D.~A. Levin and Y.~Peres.
\newblock {\em Markov chains and mixing times}, volume 107.
\newblock American Mathematical Soc., 2017.

\bibitem{lindvall1992lectures}
T.~Lindvall.
\newblock {\em Lectures on the coupling method}.
\newblock Wiley Series in Probability and Mathematical Statistics: Probability
  and Mathematical Statistics. John Wiley \& Sons, Inc., New York, 1992.
\newblock A Wiley-Interscience Publication.

\bibitem{lowen2000fun}
H.~L{\"o}wen.
\newblock Fun with hard spheres.
\newblock In {\em Statistical physics and spatial statistics: the art of
  analyzing and modeling spatial structures and pattern formation}, pages
  295--331. Springer, 2000.

\bibitem{meeron1970bounds}
E.~Meeron.
\newblock Bounds, successive approximations, and thermodynamic limits for
  distribution functions, and the question of phase transitions for classical
  systems with non-negative interactions.
\newblock {\em Phys. Rev. Lett.}, 25:152--155, 1970.

\bibitem{mertens2012continuum}
S.~Mertens and C.~Moore.
\newblock Continuum percolation thresholds in two dimensions.
\newblock {\em Physical Review E—Statistical, Nonlinear, and Soft Matter
  Physics}, 86(6):061109, 2012.

\bibitem{michelen2022strong}
M.~Michelen and W.~Perkins.
\newblock Strong spatial mixing for repulsive point processes.
\newblock {\em J. Stat. Phys.}, 189(1):Paper No. 9, 32, 2022.

\bibitem{michelen2023analyticity}
M.~Michelen and W.~Perkins.
\newblock Analyticity for classical gasses via recursion.
\newblock {\em Communications in Mathematical Physics}, 399(1):367--388, 2023.

\bibitem{michelen2025potential}
M.~Michelen and W.~Perkins.
\newblock Potential-weighted connective constants and uniqueness of {G}ibbs
  measures.
\newblock {\em Communications in Mathematical Physics}, 406(2):32, 2025.

\bibitem{michelen2025central}
M.~Michelen and J.~Sahasrabudhe.
\newblock Central limit theorems and the geometry of polynomials.
\newblock {\em J. Eur. Math. Soc. (JEMS)}, 28(5):2261--2305, 2026.

\bibitem{montgomery1988minimal}
H.~L. Montgomery.
\newblock Minimal theta functions.
\newblock {\em Glasgow Math. J.}, 30(1):75--85, 1988.

\bibitem{handook2010}
F.~W.~J. Olver, D.~W. Lozier, R.~F. Boisvert, and C.~W. Clark, editors.
\newblock {\em N{IST} handbook of mathematical functions}.
\newblock {U.S. Department of Commerce, National Institute of Standards and
  Technology, Washington, DC; Cambridge University Press, Cambridge}, 2010.

\bibitem{parisi2010mean}
G.~Parisi and F.~Zamponi.
\newblock Mean-field theory of hard sphere glasses and jamming.
\newblock {\em Reviews of Modern Physics}, 82(1):789--845, 2010.

\bibitem{penrose1963convergence}
O.~Penrose.
\newblock Convergence of fugacity expansions for fluids and lattice gases.
\newblock {\em Journal of Mathematical Physics}, 4(10):1312--1320, 1963.

\bibitem{richthammer2007translation}
T.~Richthammer.
\newblock Translation-invariance of two-dimensional {G}ibbsian point processes.
\newblock {\em Comm. Math. Phys.}, 274(1):81--122, 2007.

\bibitem{ruelle1963correlation}
D.~Ruelle.
\newblock Correlation functions of classical gases.
\newblock {\em Annals of Physics}, 25(1):109--120, 1963.

\bibitem{ruelle1969statistical}
D.~Ruelle.
\newblock {\em Statistical mechanics: {R}igorous results}.
\newblock W. A. Benjamin, Inc., New York-Amsterdam, 1969.

\bibitem{viazovska2017sphere}
M.~S. Viazovska.
\newblock The sphere packing problem in dimension 8.
\newblock {\em Ann. of Math. (2)}, 185(3):991--1015, 2017.

\bibitem{villani2009optimal}
C.~Villani.
\newblock {\em Optimal transport}, volume 338 of {\em Grundlehren der
  mathematischen Wissenschaften [Fundamental Principles of Mathematical
  Sciences]}.
\newblock Springer-Verlag, Berlin, 2009.
\newblock Old and new.

\bibitem{Wei06}
D.~Weitz.
\newblock Counting independent sets up to the tree threshold.
\newblock In {\em Proceedings of the 38th Annual {ACM} Symposium on Theory of
  Computing ({STOC})}, pages 140--149. ACM, 2006.

\bibitem{yang1952statistical}
C.-N. Yang and T.-D. Lee.
\newblock Statistical theory of equations of state and phase transitions. i.
  theory of condensation.
\newblock {\em Physical Review}, 87(3):404, 1952.

\bibitem{ziesche2018sharpness}
S.~Ziesche.
\newblock Sharpness of the phase transition and lower bounds for the critical
  intensity in continuum percolation on {$\Bbb{R}^d$}.
\newblock {\em Ann. Inst. Henri Poincar\'e{} Probab. Stat.}, 54(2):866--878,
  2018.

\end{thebibliography}
	\bibliographystyle{abbrv}

	\appendix
	
	\section{Bessel function calculations} \label{sec:bessel}
	The goal of this section is to prove Lemma \ref{lem:unit-ball-fourier}.   Let $\omega_d = \Vol_{\R^d}(1)$ and note that the radius $s_d$ of the ball of volume $1$ satisfies $s_d = \omega_d^{-1/d}$.   Setting $\nu = d/2$, a classical computation \cite[Appendix B.5]{grakaos2014classical} shows $$\widehat{\boldsymbol{1}}_{B(0,1)}(\xi) = |\xi|^{-\nu} J_{\nu}(2\pi \xi)$$
	where $J_\nu$ is the Bessel function of the first kind of order $\nu$. We thus see that \begin{equation}\label{eq:ad-in-terms-bessel}
		a_d = \omega_d^{-1} (2\pi)^\nu \left(- \min_{t \geq 0} t^{-\nu} J_\nu(t)  \right)\,.
	\end{equation}
	
	In order to understand the minimum in \eqref{eq:ad-in-terms-bessel},  set $F_\nu(t) = t^{-\nu} J_{\nu}(t)$.  By \cite[Eq.~10.6.6]{handook2010} we have 
	\begin{equation}
		F_\nu'(t) = - t^{-\nu} J_{\nu+1}(t)
	\end{equation}
	and so critical points of $F_\nu$ are precisely zeros of $J_{\nu+1}$, which we label $(j_{\nu+1,k})_{k \geq 1}$ in increasing order.  We recall that $(j_{\nu+1,k})_{k \geq 1}$ are distinct \cite[Sec.~10.21(i)]{handook2010}.  We first show that the absolute value of $F_\nu$ at its critical points are strictly decreasing.
	
	\begin{lemma}\label{lem:crit-points-bessel-decreasing}
		Let $F_\nu(t) = t^{-\nu} J_{\nu}(t)$ and let $(j_{\nu+1,k})_{k \geq 1}$ denote its critical points.  Then $|F_\nu(j_{\nu+1,1})| > |F_\nu(j_{\nu+1,2})| > |F_\nu(j_{\nu+1,3})| > \cdots . $.
	\end{lemma}
	\begin{proof}
		Recall that $J_\nu$ satisfies the differential equation (see \cite[Eq.~10.13.4]{handook2010})\begin{equation*} t^2 J''_\nu + t J_\nu' + (t^2 - \nu^2)J_\nu = 0
		\end{equation*}
		and so $F_\nu$ satisfies \begin{equation}\label{eq:F-ODE}
			F_\nu'' + \frac{2\nu + 1}{t} \cdot F_\nu' + F = 0\,.
		\end{equation}
		Define $E(t) = F_\nu(t)^2 + (F_\nu'(t))^2$ and note that \eqref{eq:F-ODE} implies $$E'(t) = -\frac{4\nu + 2}{t} (F_\nu'(t))^2 \leq 0$$
		implying $E$ is monotone decreasing and strictly decreasing except at critical points.  Further, by definition of $E$ we have that at each critical point $j_{\nu+1,k}$ we have $E(j_{\nu+1,k}) = F_\nu(j_{\nu+1,k})^2$.  Since $(j_{\nu+1,k})_{k \geq 1}$ are distinct, we have completed the proof.
	\end{proof}
	
	We are now ready to deduce Lemma \ref{lem:unit-ball-fourier}:
	
	\begin{proof}[Proof of Lemma~\ref{lem:unit-ball-fourier}]
		Since as $t \to 0$ we have the asymptotic $J_\nu(t) \sim (\Gamma(\nu+1))^{-1}(t/2)^\nu$ (see \cite[Eq.~10.7.3]{handook2010}), we see that for $t > 0$ and small we have that $F_\nu(t) > 0$.   Letting $j_{a,k}$ denote the $k$'th positive zero of $J_a$, we recall that the zeros interlace, and so in particular $j_{\nu,1} < j_{\nu+1,1} < j_{\nu,2}$ (see \cite[Eq.~10.21.1]{handook2010}).
		This implies that $F_\nu(j_{\nu+1,1}) < 0$ and so by Lemma \ref{lem:crit-points-bessel-decreasing} we have \begin{equation}\label{eq:min-at-crit-pt}
			\min_{t \geq 0} F_\nu(t) = F_\nu(j_{\nu+1,1})\,. 
		\end{equation}
		Combining \eqref{eq:ad-in-terms-bessel} with \eqref{eq:min-at-crit-pt} completes the proof.
	\end{proof}

	In order to prove Lemma \ref{lem:asymptotics-fourier-ball}, it will be sufficient to prove the following asymptotic:
	
	\begin{lemma}\label{lem:Bessel-minimum}
		Let $J_\nu$ be the Bessel function of the first kind.  Then $$j_{\nu+1,1}^{-\nu}J_\nu(j_{\nu+1,1}) = \nu^{-\nu} \exp\left( -2^{-1/3}\alpha_1 \nu^{1/3}\right)e^{-1} \cdot 2^{2/3} \nu^{-2/3} \left(-\Ai'(-\alpha_1)\right)(1 + O(\nu^{-1/3}))\,.$$
	\end{lemma}
	
	We first deduce Lemma \ref{lem:asymptotics-fourier-ball} from Lemma \ref{lem:Bessel-minimum}:
	\begin{proof}[Proof of Lemma \ref{lem:asymptotics-fourier-ball}]
		Since $\omega_d = \frac{\pi^{\nu}}{\Gamma(\nu+1)}$ we use Lemmas \ref{lem:unit-ball-fourier} and \ref{lem:Bessel-minimum} along with Stirling's formula to see \begin{align*}
			a_d &= \Gamma(\nu+1)2^\nu  \nu^{-\nu} \exp\left( -2^{-1/3}\alpha_1 \nu^{1/3}\right) e^{-1} \cdot 2^{2/3} \nu^{-2/3} \Ai'(-\alpha_1)(1 + O(\nu^{-1/3})) \\
			&=  \left(\frac{2}{e}\right)^\nu \sqrt{2\pi \nu }\exp\left( -2^{-1/3}\alpha_1 \nu^{1/3}\right) e^{-1} \cdot 2^{2/3} \nu^{-2/3} \Ai'(-\alpha_1)(1 + O(\nu^{-1/3}))\,.
		\end{align*}
		Recalling $\nu = d/2$ completes the proof. 
	\end{proof}

	The main tool for proving Lemma \ref{lem:Bessel-minimum} is an approximation in terms of the Airy function $\mathrm{Ai}$.  In particular, for $x$ in a compact set we uniformly have \cite[Eq.~10.21.44]{handook2010} \begin{equation}\label{eq:turning-asympt}
		J_\nu(\nu + 2^{-1/3} x \nu^{1/3}) = 2^{1/3}\nu^{-1/3} \Ai(-x) + O(\nu^{-1})\,.
	\end{equation}

	We first identify $j_{\nu+1,1}$ in terms of a zero of $\Ai$. \begin{lemma}\label{lem:jnu+1-expansion}
		Let $\alpha_1 = \min\{\alpha > 0 : \Ai(-\alpha) = 0\}.$  Then $$j_{\nu+1,1} = \nu+ 2^{-1/3} \alpha_1 \nu^{1/3} + 1 + O(\nu^{-1/3})\,.$$
	\end{lemma}
	\begin{proof}
		We see from \eqref{eq:turning-asympt} that for $\eps = O(\nu^{-2/3})$ we have \begin{align*}
			J_{\nu+1}(\nu+1 + (\alpha_1+ \eps)2^{-1/3}(\nu+1)^{1/3}) &= 2^{1/3} (\nu+1)^{-1/3}\Ai(-(\alpha_1 + \eps)) + O(\nu^{-1}) \\
			&=2^{1/3} (\nu+1)^{-1/3} \left( - \eps\Ai'(-\alpha_1) + O(\eps^2)\right) + O(\nu^{-1})
		\end{align*}
		since $\Ai$ is analytic.  Noting that $\Ai'(-\alpha_1) \neq 0$ (see \cite[Eq.~9.9.3]{handook2010}) completes the proof.
	\end{proof}
	
	\begin{proof}[Proof of Lemma \ref{lem:Bessel-minimum}]
		By \eqref{eq:min-at-crit-pt} we have \begin{align*}
			\min_{t \geq 0} F_\nu(t) = F_\nu(j_{\nu+1,1}) = (j_{\nu+1,1})^{-\nu} J_\nu(j_{\nu+1,1})\,.
		\end{align*}
		Using \eqref{lem:jnu+1-expansion} we have that \begin{align}
			(j_{\nu+1,1})^{-\nu} &= \nu^{-\nu} \left(1 + 2^{-1/3}\alpha_1 \nu^{-2/3} + \nu^{-1} + O(\nu^{-4/3})\right)^{-\nu} \nonumber \\
			&= \nu^{-\nu} \exp\left( -2^{-1/3}\alpha_1 \nu^{1/3} - 1 + O(\nu^{-1/3})\right)\,.\label{eq:j-power-expansion}
		\end{align}
		Similarly, combining \eqref{eq:turning-asympt} with Lemma \ref{lem:jnu+1-expansion} we have \begin{align*}
			J_\nu(j_{\nu+1,1}) &= J_\nu\left( \nu + 2^{-1/3} \nu^{1/3}\left(\alpha_1 + 2^{1/3}\nu^{-1/3} + O(\nu^{-2/3}) \right)  \right) \\
			&= 2^{1/3} \nu^{-1/3} \Ai\left(-\left(\alpha_1 + 2^{1/3}\nu^{-1/3} + O(\nu^{-2/3}) \right) \right) + O(\nu^{-1}) \\
			&= 2^{2/3} \nu^{-2/3} \left(-\Ai'(-\alpha)\right)(1 + O(\nu^{-1/3}))\,.
		\end{align*}
		Combining the previous equation with \eqref{eq:j-power-expansion} completes the proof.
	\end{proof}

\end{document}